\documentclass[12pt]{article}
\usepackage{amsmath, amsthm}
\usepackage{amsfonts}
\usepackage{amssymb}
\usepackage{graphicx}
\usepackage{subcaption}
\usepackage{amscd}
\usepackage{bm}
\usepackage{xcolor, colortbl}

\usepackage{epstopdf}
\usepackage{algorithmic}
\ifpdf
  \DeclareGraphicsExtensions{.eps,.pdf,.png,.jpg}
\else
  \DeclareGraphicsExtensions{.eps}
\fi

\usepackage{enumitem}
\setlist[enumerate]{leftmargin=.5in}
\setlist[itemize]{leftmargin=.5in}

\title{Two-Dimensional Frequency-Difference-of-Arrival Varieties} 

\author{Jeanne Duflot
\and Margaret Cheney\footnotemark[2]
\and James A. Given 
}

\newtheorem{theorem}[equation]{Theorem}
\newtheorem{lemma}[equation]{Lemma}

\newtheorem{corollary}[equation]{Corollary}
\newtheorem{remark}[equation]{Remark}
\newtheorem{example}[equation]{Example}
\newtheorem{definition}[equation]{Definition}

\newcommand{\bxi}{\boldsymbol{\xi}}

\begin{document}
\maketitle

\begin{abstract}
This paper studies  Frequency-Difference-of-Arrival (FDOA) curves for the 2-dimensional, 2-sensor case.  The primary focus of this paper is to give a description of curves associated to the FDOA problem from the algebro-geometric point of view.  
 To be more precise,  the complex projective picture of the family of FDOA curves for all possible relative velocities is described.  
\end{abstract}




\tableofcontents

\setcounter{tocdepth}{1}

\section{Introduction}

 There are a variety of situations in which we might wish to determine the location from which electromagnetic or acoustic waves are emanating:  examples are locating gunshots, radar systems, or the ``black box" of a downed aircraft.  To locate such a source, typically one measures the associated signals at multiple receiver locations, and compares the measurements.  
One possible way to compare the measurements is to look for the time difference of arrival (TDOA) between the two signals.  
This is typically done by cross-correlating \cite{SMGC} the measurements and determining the time delay at which a peak occurs in the cross-correlation function.  For some signals, however, the TDOA resolution may be very poor.    

Another possible way to compare the measurements is to use one or more moving sensors with known positions and velocities.   
Alternatively, fixed sensors can be used to measure emanations from a moving source with known velocity.  
In either case, the relative motion results in signals that are Doppler-shifted, and one can look for the frequency difference of arrival (FDOA) between the signals.  This can be done by Doppler-shifting one signal by various amounts and then cross-correlating \cite{SMGC} with the other signal.  
This process results in the ``cross-ambiguity'' funtion, which has been thoroughly studied \cite{L}, and it turns out that the TDOA and FDOA measurements are complementary in the sense that they satisfy an uncertainty relation.  Thus it can easily happen that a signal that produces poor TDOA resolution will produce good FDOA resolution.  

Consequently it is of great interest to study source localization based on FDOA measurements.  This turns out to be much more difficult than localization based on TDOA measurements.  

Typically, using TDOA measurements to locate a source involves considering the TDOA for each pair of sensors in turn, and for each such pair, mapping out the set of source locations that could have produced the TDOA measurement.  In a free-space environment, the TDOA for each pair produces a constraint that the source must lie on a certain hyperboloid; consequently the source must be located at the intersection of the hyperboloids for the various sensor pairs.  If the source is known to lie on a flat plane, then the TDOA source localization problem \cite{HoCh} becomes a problem of determining the intersection of hyperbolas.  

For FDOA measurements, however the corresponding constraints are much more complicated than hyperbolas.  

This paper investigates  curves of constant FDOA, in a free-space environment,  for the case when both the source and the pair of receivers lie on a flat two-dimensional plane.  

Our approach is motivated by the  TDOA paper of Ho and Chan \cite{HoCh} 
and is based on their point of view;  their work differs from the present paper in that we exploit the 
 tools of algebraic geometry, and we consider the more difficult FDOA problem.
 
 The Conclusions section gives a detailed summary of the specific accomplishments of the document.
 
\section{Problem formulation and strategy}
We assume that the sensor locations and velocities are known, and from the FDOA measurement, we wish to find the location of a stationary source. 

\subsection{Notation and geometry}
We denote the source location by $\bm y \in \mathbb R^2$ and the sensor locations $\bm s_1, \bm s_2 \in \mathbb R^2$.
In the frame of reference in which the source is stationary, we denote the sensor velocities by 
$\bm v_1, \bm v_2 \in \mathbb R^2$, 
respectively.   
We note that the unit vector from the source to the sensor  at $\bm s$ is $(\bm s- \bm y)/|\bm s- \bm y|$ where $|\bm x |$ denotes the Euclidean norm of the vector $\bm x$.  

When the sensor velocities are both equal to $\bm v$, we can switch to the frame of reference in which the sensors are stationary and the source is moving with velocity $-\bm v$.  Thus a special case of the scenario with a stationary source and moving sensors also addresses the case of stationary sensors and a source moving with a known velocity.  

We consider only measurements made at one instant, so we do not address any change in sensor position.


\subsection{TDOAs}
The time difference of arrival of a signal at the two sensors located at $\bm s_1$ and $\bm s_2$ is
\begin{equation}	\label{TDOAdef}
\frac{1}{c}  \left( \underbrace{ |\bm s_1 - \bm y| - |\bm s_2 - \bm y| }_{\rm TDOA}\right).
\end{equation}
We neglect the factors of the (assumed known and constant) speed $c$ of signal propagation; 
alternatively, we choose units so that $c=1$.

\subsection{Doppler shifts and FDOA curves}
Suppose
 a stationary emitter at location $\bm y = (y_1, y_2)$ emits a signal of frequency $\omega_0 \in \mathbb R_+$ that is  then received by moving receivers at locations $\bm s_1, \bm s_2$ and velocities $\bm v_1, \bm v_2$.  
The signal at each receiver will be Doppler-shifted by an amount proportional to the line-of-sight component of its velocity.  
In particular, the received frequency at receiver $j$ will be  
$\omega = \omega_0  \left[ 1 +  \frac{(\bm s_j- \bm y)}{ |\bm s_j- \bm y |} \cdot \frac{\bm v_j}{c} \right]$ where $c$ is the (assumed constant) speed of signal propagation.  

Unfortunately, we cannot measure the Doppler shift itself, because we do not know the frequency $\omega_0$ emitted by the source.  However, if we measure the frequency at two stationary receiver locations, we obtain
$\omega_1=  \omega_0(1 + \frac{(\bm s_1 \bm -\bm y)}{ |\bm s_1-\bm  u |} \cdot \frac{\bm v_1}{c})$ and $\omega_2 = \omega_0(1 + \frac{(\bm s_2- \bm y)}{ |\bm s_2- \bm y |} \cdot \frac{\bm v_2}{c} )$, which we can divide to obtain
\begin{align}
\frac{\omega_2}{\omega_1} &= \frac{1 + \frac{(\bm s_2- \bm y)}{ |\bm s_2- \bm y |} \cdot  \frac{\bm v_2}{c} }{1 
	+ \frac{(\bm s_1-\bm y)}{ |\bm s_1- \bm y |} \cdot \frac{\bm v_1}{c} } 
= \left(1 +   \frac{(\bm s_2-\bm y)}{ |\bm s_2- \bm y |} \cdot \frac{\bm v_2}{c}  \right) 
	\left( 1- \frac{(\bm s_1-\bm y)}{ |\bm s_1- \bm y |} \cdot \frac{\bm v_1}{c}   + \cdots \right) \cr
	&= 1 +  \frac{1}{c} \cdot  \left(  \underbrace{\frac{(\bm s_2-\bm y)\cdot \bm v_2}{ |\bm s_2- \bm y |}   
		- \frac{(\bm s_1-\bm y)\cdot \bm v_1}{ |\bm s_1- \bm y |} }_{\rm FDOA}    + \cdots \right) 
\end{align}

 Here the higher-order terms are smaller in $\bm v_j/c$, a quantity which is typically very small.  


One of the objectives here  is to describe the geometric properties of the set of points $\bm y \in \mathbb{R}^2$ that satisfy the equation

\begin{equation} 	\label{FDOAdef}
d = \frac{(\bm s_2- \bm y)\cdot \bm v_2}{| \bm s_2 - \bm y |} - \frac{(\bm s_1- \bm y)\cdot \bm v_1}{ |\bm s_1- \bm y |}
= \bm e_2 \cdot \bm v_2 - \bm e_1 \cdot \bm v_1,
\end{equation} 
where $d$ is a fixed real number and where $\bm e_j =(\bm s_j - \bm y) /  |\bm s_j- \bm y |$  are unit vectors, depending on $\bm y$.

 Thus we should not be surprised to find that an important role is played by the torus $S^1 \times S^1$, where $S^1$ denotes the unit circle (set of unit vectors) in $\mathbb R^2$. 
 
 There are some constraints on the unit vectors $\bm e_1,\bm e_2$.  If we choose our initial coordinate frame $(y_1,y_2)$ so that the sensor positions are 
 $$\bm s_2 = (-a,0), \qquad \bm s_1 = (a,0),$$ 
 where $a$ is a fixed strictly positive real number, then the unit vectors are 
 \begin{equation}	\label{unitvecs}
 \bm e_1 = \frac{ (a-y_1, -y_2)}{| (a-y_1,-y_2)| }, \qquad  \bm e_2 = \frac{(-a-y_1,-y_2)}{|(-a-y_1,-y_2)|}.
 \end{equation}  
 Note that the second coordinates of $\bm e_1, \bm e_2$ are either both nonnegative or are both nonpositive; in other words, these second coordinates cannot be opposite in sign if nonzero.  In other words, if we describe the unit vectors  as 
 $$\bm e_1 = ( \cos \theta, \sin \theta),  \qquad \bm e_2 = (\cos \tau, \sin \tau),$$ then  
 $$0 \leq \theta \leq \pi \; \mbox{and}\; 0 \leq \tau \leq \pi $$ or
 $$\pi \leq \theta \leq 2 \pi \; \mbox{and}\; \pi \leq \tau \leq 2 \pi.$$
  
As noted in \cite{KC}, there are also some constraints on $\bm v, d$:  using the Cauchy-Schwarz inequality for equation \ref{FDOAdef}, we see that if all parameters are real numbers, then
\begin{equation}  \label{CSbound} 
\mid d \mid = \mid \bm e_2 \cdot \bm v_2 - \bm e_1 \cdot \bm v_1 \mid \leq \mid \bm e_2 \cdot \bm v_2\mid + \mid \bm e_1 \cdot \bm v_1\mid \leq \|\bm v_2\| + \| \bm v_1\|.
\end{equation}  
In other words, for there to be real solutions $\bm y$ to equation \ref{FDOAdef}, the parameters $\bm v,d$ must satsify this ``Cauchy-Schwarz" bound.


\subsection{Strategy and formulation of the modified FDOA problem}	\label{strategy}
One of the main points of this paper is to describe   the   FDOA curves,  related to the solution set for $(y_1,y_2)$ of the equation \eqref{FDOAdef},  using tools from algebraic geometry.  What sort of curves are these, from the geometric point of view?  Since most of the tools used from algebraic geometry require working over the complex numbers, once we get a description of the curves over the complex numbers, how does this help us figure out how the curves behave when working over the real numbers?

The strategy is to  describe a specific process that starts with  the set of points $(y_1,y_2) \in \mathbb{R}^2$ with defining equation \eqref{FDOAdef} and arrives at an  algebraic variety (a geometric object) which is the solution set, in some projective space over the complex numbers, to a system of degree two polynomial equations; describing this new solution set in projective space over the complex numbers will be called ``the modified problem".  This is the primary concern of this paper.    Some geometric properties of the corresponding solution sets for the modified problem over the real numbers are discussed, but more detailed analyses over the real numbers are left for  future work.   These   detailed analyses over the real numbers 
 would use, in crucial ways, the geometric view of the solution sets over the complex numbers which we present in this paper.

\subsubsection{Use of projective space} 

The description of various sorts of ``FDOA" curves is first given in projective spaces:  the projective  line, the projective plane, projective 3-space and projective 4-space.  We usually start by working over the complex numbers, and pass from the description of a curve over the complex numbers to a description over the real numbers.  Allowing all quantities to be complex is useful in dealing with polynomial systems for many reasons; for example, polynomial systems  have solutions over the complex numbers:  a polynomial in one variable of degree $n$ has $n$ complex roots, counted with multiplicity, whereas there are many polynomials in one variable which have real coefficients, but no real solutions.

We summarize here basic definitions and notation used to define projective space.  For example, two-dimensional real projective space $\mathbb R P^2$ is the set of one-dimensional subspaces (lines through the origin) in $\mathbb R^3$.  This means that a point in $\mathbb{R}P^2$ is given by a nonzero vector $(x_0,x_1,x_2) \in \mathbb{R}^3$, up to nonzero scalar multiples, since a line through the origin in $\mathbb{R}^3$ is uniquely determined by a nonzero direction vector  for that line, and any two direction vectors for the same line are nonzero scalar multiples of each other.  This gives us the definition of a set of projective coordinates on $\mathbb{R}P^2$:  we write a point in $\mathbb{R}P^2$ as $[x_0,x_1,x_2]$ (the ``square bracket notation")  where $(x_0,x_1,x_2)$ is a nonzero vector in $\mathbb{R}^3$ and  $[x_0,x_1,x_2] = [\tilde{x}_0,\tilde{x}_1, \tilde{x}_2]$ in $\mathbb{R}P^2$ if and only if there is a nonzero real number $r$ such that $\tilde{x}_j = rx_j$ for each $j$.  
The function
$$\mathbb{R}^3 - \{0\} \rightarrow \mathbb{R}P^2$$ given by
$$(x_0,x_1,x_2) \mapsto [x_0,x_1,x_2],$$  gives $\mathbb{R}P^2$ its topology;  $\mathbb{R}P^2$ is a compact space.  Note that points in the projective plane correspond to lines in  Euclidean 3-space and lines in the projective plane correspond to planes in Euclidean 3-space.

In general, if $K$ is either the set of complex numbers $\mathbb{C}$, or the set of real numbers $\mathbb{R}$, and $1 \leq n \leq 4$,
$$KP^n = \{ [x_0,x_1, \ldots, x_n] \mid (x_0, \ldots, x_n)  \neq \bm 0 \in K^{n+1} \},$$ and 
$$[x_0,x_1, \ldots, x_n] = [y_0,y_1, \ldots, y_n]$$
$$ \mbox{if and only if there is a}\; \lambda \neq 0 \in K \; 
\mbox{such that} \; x_j= \lambda y_j, \; \mbox{ for every}\; j.$$ While we won't give details about the topology of $KP^n$ here, we'll  refer to ``affine open subsets" of $KP^n$, and we give a brief description of these.

We'll just focus on the example of $KP^2$.  There are three special open subsets of $KP^2$,  $U_0,U_1,U_2$, called the ``standard affine open subsets".  These sets are described by
\begin{equation}	\label{R2inP2}
U_j \doteq \{[x_0,x_1,x_2] \in KP^3 \mid x_j \neq 0\}.
\end{equation}  
Each is homeomorphic with $K^2$ via a specific homeomorphism $h_j:U_j\rightarrow K^2$.

For example, in later sections we will use the specific homeomorphism $\mathbb{R}^2  \cong  U_0$ established by the case $j=0$, namely with $h_0 : U_0 \rightarrow \mathbb{R}^2$  defined by
$$[x_0,x_1,x_2] \mapsto (x_1/x_0, x_1/x_0) $$ 
and the inverse  $h_0^{-1}:  \mathbb{R}^2 \rightarrow \mathbb{R} P^2$ given by
$$(y_1,y_2) \mapsto [1,y_1,y_2].$$      

The open sets $U_j$ and homeomorphisms $h_j$ may be used to give $KP^2$ the structure of a two-dimensional (over $K$) compact manifold.  Similarly, $KP^1$ is covered by two standard affine open subsets, $KP^3$ is covered by four open affine subsets, etc.

Why do we study geometric objects in projective space, rather than in $\mathbb{R}^n$ or $\mathbb{C}^n$?  Projective spaces  are compact,  and can be considered as a sort of compactification of Euclidean spaces.  Geometric analyses in projective space are often simpler:  for example, in the projective plane,  distinct lines always intersect in exactly one point, which isn't true in the Euclidean plane.  Moreover, one can pass easily from geometry in projective space  to geometry in Euclidean space using the affine open subsets.   Finally, theorems used from algebraic geometry often use hypotheses that consider geometric objects in complex projective sapce.
\subsubsection{Projective varieties}\label{varietydef}

  If $F_1, \ldots, F_k$ are homogeneous polynomials in variables $x_0, \ldots, x_n$,  with coefficients in the complex numbers,  a {\bf projective variety}
$V \doteq V(F_1, \ldots, F_k) \subseteq \mathbb{C}P^n$  is defined as  the set of  $[\bxi] \in \mathbb{C}P^n$ such that
$$F_1(\bxi) =  \ldots = F_k(\bxi) = 0.$$ 

Note that the definition of $V$ does not depend on the choice of $\bxi$ representing $[\bxi]$ 
 since the polynomials $F_j$ are homogeneous polynomials.  

 If the coefficients of all the polynomials $F_j$ are real numbers, we say that $V$ is a {\bf real} variety.  In this case,  $V(\mathbb{R})$ is the set of points 
  $[\bxi] \in V$ such that there are real numbers $x_j$ such that
$[\bxi] = [x_0, \ldots, x_n]$. 
 We then say ``{\bf $V(\mathbb{R})$ is the set of real points of the real variety $V$}".  The set $V(\mathbb{R})$ is also called a projective variety, since it is a subset of $\mathbb{R}P^n$ that is a solution set of a system of polynomial equations. For some specific examples of varieties considered here, we may use different letters than ``$V$"; e.g.,  in the next section, we define a variety called $Y(Q_1,Q_2)$.

\subsubsection{Using algebraic geometry to study equations involving Euclidean norms}
Algebraic geometry in its most basic conception is the study of solutions to polynomial systems of equations.  While concerned with algebraically finding individual such solutions, algebraic geometers'  greater interest  is  to geometrically understand the set of solutions as a whole.  Now, in order to study, with the tools of algebraic geometry,  solutions of equations like \eqref{FDOAdef} involving norms of vectors, since norms are nonnegative square roots of polynomial expressions of degree two, the square roots must be eliminated, by ``squaring", or another method.  In this paper, as in \cite{HoCh}, we will also eliminate  square roots by adding new variables and new polynomial equations in a systematic way.  We do want to know how to go back and forth between the resulting polynomial system and the original equations which involve norms, and hence have square roots sprinkled throughout.

Let us proceed now to describe a process, used in this  paper, of going from the study of solutions to the FDOA equation (2.2) to a geometric object in projective space, called an ``FDOA variety" or ``FDOA curve",  which can be  directly studied with tools from algebraic geometry.

We begin by moving into projective space, adding the variable $t$ to accomplish this (this is called ``homogenization"), and two additional variables, $r_1$ and $r_2$,   in order to obtain polynomial equations for a solution set which can be studied using algebro-geometric tools.  

So,  we add three new variables, $t, r_1$ and $r_2$, and form the two new second-degree homogeneous polynomials
\begin{equation}	\label{Qjdef}
Q_j(t,  y_1, y_2, r_1,  r_2)  \doteq (s_{j1}t-y_1)^2 +  (s_{j2}t-y_2)^2 - r_j^2  \qquad j = 1, 2.
\end{equation}   

Before looking at how equation (2.2) is incorporated, we first study the solution set of these two homogeneous polynomials in $\mathbb{C}P^4$; this solution set is the {\bf  subvariety} $Y(Q_1, Q_2)$ of $\mathbb{C}P^{4} $ defined by
$$Y(Q_1,Q_2) = \{[t,\bm y, \bm r] \in \mathbb{C}P^{4} \mid Q_j(t,\bm y,\bm r) = 0, \; \mbox{for}\; j = 1,2 \}.$$ Here we have used the  abbreviated notation: $[t,\bm y,\bm r] = [t,y_1,y_2, r_1, r_2].$   

We will
call $Y(Q_1,Q_2)$ the ``ambient" variety, since it holds  FDOA varieties, for every choice of velocity vector pairs; in other words, the ambient variety $Y(Q_1, Q_2)$ is fundamental to the use of algebraic geometry in addressing sensing problems for pairs of ranging sensors, because it holds information about ranges from each of two sensors, for any pair of velocity vectors.  
The particular sensing problem for a pair of sensors, involving for example one or more TDOA equations
 ({\it i.e.,} equation of the form TDOA = measured value, where TDOA is given by \eqref{TDOAdef})
or  FDOA equations  ({\it i.e.,}  equations of the form FDOA = measured value, where FDOA is given by \eqref{FDOAdef}), 
involves adding an extra polynomial condition.  In fact, either a TDOA or FDOA problem can be studied using an appropriate subvariety of $Y(Q_1, Q_2)$.

\begin{remark}{ \bf TDOA Example.} 
The (modified) TDOA subvariety of $Y(Q_1, Q_2)$ is defined by supplementing the conditions $Q_1 = Q_2 = 0$ with the TDOA condition $L=0$, where $L$ is the polynomial  
$L = r_2-r_1-b t.$ 
Here $b$  is a fixed nonzero real number representing the TDOA value computed from the two sensors.  
 \end{remark}

The standard open affine subsets, and the homeomorphisms between these open subsets of projective space and Euclidean space,  give a way of passing from information obtained in the ambient variety, which lies in projective space, to information in Euclidean space.  For example, if $[t,y_1,y_2,r_1,r_2] \in Y(Q_1,Q_2)$ with $t \neq 0$ and all variables real numbers,  and we label   coordinates in $\mathbb{R}^4$  as $(z_1,z_2,R_1,R_2)$, when we pass  to $\mathbb{R}^4$ using the correspondence $z_1 = y_1/t, z_2 = y_2/t, R_1 = r_1/t, R_2 = r_2/t$, then, after doing a little algebra,  the part of $Y(Q_1, Q_2)$ where $t \neq 0$ and all coordinates are real, maps homeomorphically to the solution set $(z_1,z_2,R_1,R_2)$ in $\mathbb{R}^4$ of the  equations
\begin{equation}	\label{Sdef}
R_j^2 = |\bm s_j - (z_1, z_2)|^2 \qquad j = 1,2.
\end{equation}

\subsubsection{FDOA conditions in complex projective space} \label{HCFdef}
The sensors have fixed, known instantaneous velocity vectors $\bm v_1, \bm v_2$ in $\mathbb{R}^2$, where $\bm v_j= (v_{j1}, v_{j2}) $.   If $\bm v_j \neq \bm 0$, for $j = 1,2 $, define a linear polynomial $L_j$ in the variables $t, \bm y$   by
\begin{equation}	\label{vdotprod}
L_j(\bm v) = L_j(\bm v_1,\bm v_2) \doteq \bm v_j \cdot (\bm s_j t - \bm y) = \sum_{l = 1}^2 v_{jl}(s_{jl}t-y_l), \qquad j = 1,2 .
\end{equation}
Next, given the FDOA value $d$  for the sensor pair, transform the FDOA equation \eqref{FDOAdef} by clearing denominators, moving all terms to one side of the equation, and using our new variables.  Thus we arrive at the  polynomial

 \begin{equation}	\label{FDOAcond}
 \tilde{Q}(\bm v,d) \doteq L_2(\bm v)r_1-L_1(\bm v)r_2 - r_2r_1 d .
 \end{equation}
and  the ``FDOA" subvariety of $\mathbb{C}P^{4}$ defined by 
 \begin{equation}	\label{FDOAvariety}
 HC_F(\bm v,d) = \{[t,\bm y,\bm r] \in Y(Q_1, Q_2) \mid \tilde{Q}(\bm v,d)(t, \bm y, \bm r) = 0 \}.
 \end{equation}
 
 The ``$HC$" in the name of the FDOA variety comes from Ho-Chen, since their paper \cite{HoCh} 
 uses the method of finding solutions of \eqref{FDOAdef} by adding the extra variables $r_1,r_2.$  The ``$F$" in the subscript comes from ``FDOA".  Note that the varieties $HC_F$ no longer lie in a real two dimensional Euclidean plane, but rather on a complex surface $Y(Q_1,Q_2)$ in complex  projective four-space $\mathbb{C}P^4$.

\subsection{The FDOA problem: Summary}\label{sumFDOA}
  We are interested in three ``FDOA" problems for a single pair of sensors. 
\begin{enumerate}
\item {\bf The starting problem.}  Given a sensor geometry, velocities $\bm v_1, \bm v_2$, and the FDOA measurement $d$, describe geometrically the set  of points $\bm y$ satisfying the FDOA equation \eqref{FDOAdef}.  This corresponds to describing the set of possible source locations from a single FDOA measurement.

\item {\bf The modified problem.}  In this  paper, much of the focus is on the study of the Ho-Chen varieties $HC_F(\bm v, d).$ Our aims here are as follows.

\begin{enumerate}
	\item  Describe geometrically the ambient variety $Y(Q_1,Q_2)$ in $\mathbb{C}P^{4}$.  (Section \ref{ambientY})
	\item Describe $HC_F(\bm v,d)$ geometrically.  (Sections \ref{Xvariety}, \ref{secIrred}, \ref{specialXvariety})
	\item  If $\bm v_1,\bm v_2,d$ are all real, describe geometrically the real points of the varieties $Y(Q_1,Q_2)$ and $HC_F(\bm v,d)$; namely, describe $Y(Q_1,Q_2)(\mathbb{R}), HC_F(\bm v, d)(\mathbb{R})$.  
	 (Sections \ref{realY} and \ref{realX})  
	 \item Understand  how points on $HC_F(\bm v,d)$ and points $\bm y$ in $\mathbb{R}^2$ satisfying the FDOA equation \eqref{FDOAdef} are related, and begin to analyze how the geometry of $HC_F(\bm v,d)$ helps us in the study of the starting problem.	(Section \ref{realX})
\end{enumerate}

\end{enumerate}

\section{Different  coordinate systems}\label{coords}
 In this section we write down  new coordinate systems we'll use on $\mathbb{C}P^4$.  These changes of coordinates are real changes of coordinates (the matrix of the coordinate change has real entries), and each have their computational uses.
  
\subsection{The ``original" coordinates}

First, we choose our initial coordinate frame  $[t,\bm y, \bm r]$ so that the sensor positions are 
$$\bm y = \bm s_2 = (-a,0), \qquad \bm y = \bm s_1 = (a,0),$$ 
where $a$ is a fixed strictly positive real number. 
The polynomials $Q_1,Q_2$  are then
\begin{equation}	\label{Qjadef}
Q_1   \doteq (at-y_1)^2 + y_2^2  -r_1^2, \qquad 
	Q_2  \doteq (at+y_1)^2 + y_2^2 -r_2^2.
\end{equation}

Next, we  set $u = at$ and leave the other coordinates the same, to obtain the ``original" coordinates $[u, \bm y, \bm r]$.
We now have
\begin{equation}	\label{Qudef}
Q_1   \doteq (u-y_1)^2 + y_2^2  -r_1^2,  \qquad 
	Q_2  \doteq (u+y_1)^2 + y_2^2 -r_2^2.
\end{equation}

 Let's see which points on $HC_F$ ``correspond to" the sensor positions $(\pm a,0)$ using our piece $U_0 \cap HC_F$ of $HC_F$.  We compute that the equation $\tilde{Q}=0$ is $(v_{21}(-u-y_1) -v_{22}y_2) r_1 - (v_{11}(u-y_1)-v_{12}y_2)r_2 -dr_1r_2 = 0$, in the $[u,\bm y, \bm r]$ coordinates.   Since we are in $U_0$, $(a,0)$ in the $\bf y$ coordinates corresponds to two points $[1,1,0,0, \pm 2]$,  while $(-a,0)$ corresponds to $[1,-1,0,\pm 2,0]$ in the $[u,\bm y, \bm r]$ coordinates on $HC_F$ (in doing this computation, remember that we are in projective space).  Thus, in this instance,   by ``corresponds to"  we mean, first, that we are looking in the open subset $U_0 \cap HC_F$ of $HC_F$, and second, we are looking at which points on $HC_F$ map to $(\pm a,0) \leftrightarrow [1, \pm a, 0]$ via the projection $[u,\bm y, \bm r] \mapsto [u,  \bm y]$ of $\mathbb{C}P^4 $ to $\mathbb{C}P^2$ .  The projection maps two points on $HC_F$ to each sensor position.    

\subsection{The $[\bm w,x_1]$-coordinates}
Here, 
$$w_0 = u-y_1-r_1, \quad  w_1 = -u-y_1-r_2,  \quad w_2 = -u-y_1+r_2,  \quad w_3 = u-y_1 + r_1,  \quad x_1 = -y_2.$$  
In these coordinates,  points in $\mathbb{C}P^{4}$ are denoted as
$$[w_0,w_1,w_2,w_3,x_1]$$
and the polynomials $Q_1,Q_2$ of \eqref{Qudef} become
$$Q_1 = w_0w_3 + x_1^2,$$
$$Q_2 = w_1w_2 +x_1^2.$$

In these coordinates,  we see that $Y(Q_1,Q_2) = Y(Q,Q_1)$, 
 where $Q = w_0w_3-w_1w_2.$
We'll use this coordinate system to talk about the geometry of $Y(Q,Q_1) = Y(Q_1,Q_2)$, for example. 

The inverse of the transformation  $(u, y_1, y_2, r_1, r_2) \mapsto (w_0, w_1, w_2, w_3, x_1)$ is
\begin{align}	\label{w2y}
u &= \frac{1}{4} (w_0 - w_1 - w_2 + w_3), \quad y_1 = -  \frac{1}{4} ( w_0 + w_1 + w_2 + w_3) , \quad  y_2 = -x_1 \cr
 r_1 &= \frac{1}{2} (-w_0 + w_3), \quad r_2 = \frac{1}{2} (-w_1 + w_2) . 
\end{align}
Having the inverse will enable us to interpret, in terms of the original coordinates, conclusions obtained from analysis in the $[\bm w, x_1]$ coordinates.

\subsection{The $[\bm z,x]$-coordinates}
This last coordinate system is given by:

$$z_0 = 2(w_3-w_0),  \quad z_1 = 2(w_2-w_1),  \quad z_2 = 2(w_0+w_3),  \quad z_3 = 2(w_2 + w_1),  \quad x = 4x_1.$$ 
In these coordinates, $Q = w_0w_3-w_1w_2=0$ if and only if
$$z_2^2-z_0^2 -(z_3^2-z_1^2) = 0,$$ while $Q_1 = x_1^2 + w_0w_3=0$ is equivalent to
$$x^2 + (z_2^2-z_0^2) = 0.$$

We'll  usually use this coordinate system on $Y(Q,Q_1)$ in discussing the Ho-Chen FDOA subvarieties of $Y(Q,Q_1)$.  The reason for using this coordinate system is that it  simplifies certain computations  for the FDOA subvarieties.

The inverse of the  composite transformation $(u, y_1, y_2, r_1, r_2) \mapsto  (z_0, z_1,z_2, z_3, x)$ is
\begin{align}	\label{z2y}
u = \frac{z_2 - z_3}{8} , \quad y_1 = \frac{-z_2 - z_3}{8} , \quad 
	y_2 = - \frac{x}{4} , \quad 
r_1  = \frac{z_0}{4} , \quad r_2 = \frac{z_1}{4} .
\end{align} 
This mapping enables us to interpret results obtained from analysis in the $[\bm z, x]$ coordinates.

If we are specifically  using  either the $[\bm w, x_1]$-coordinates, or the $[\bm z, x]$-coordinates, we'll refer to the ambient variety as $Y(Q,Q_1)$; otherwise, we'll refer to $Y(Q_1,Q_2)$.

\begin{remark}  We note that the coordinate changes described here are {\bf real, linear} coordinate changes.  Here, a {\bf real, linear} change of coordinates is a change of coordinates defined by a set of homogeneous linear equations with real coefficients.  \end{remark}

\section{The geometry of the ambient variety $Y(Q_1,Q_2)$} \label{ambientY}
 
 Using the $[\bm w, x_1]$-coordinates on $\mathbb{C}P^4$,
$$Y(Q_1,Q_2) = Y(Q,Q_1)$$ is the set of points $[\bm w, x_1] \in \mathbb{C}P^4$ such that 
$$w_0w_3-w_1w_2=0, \qquad  x_1^2 + w_0w_3 = 0.$$

Now, $w_0w_3-w_1w_2, x_1^2 + w_0w_3$ are homogeneous, irreducible polynomials of degree two  that are not constant multiples of each other; this enables one to show that $w_0w_3-w_1w_2,  \ x_1^2 + w_0w_3$ is a {\bf regular sequence} of polynomials.  Commutative algebra (dimension theory) then tells us that  $Y(Q,Q_1)$ is a surface in $\mathbb{C}P^4$ (a del Pezzo surface, as we shall see); in other words, $Y(Q,Q_1)$ has complex dimension equal to two.  

\subsection{Singularities of $Y(Q, Q_1)$} 
 For the purposes of this paper, the singularities of an algebraic variety are those points on the variety at which the matrix of partial derivatives of an appropriate set of  defining equations for the variety drops in rank.  These correspond to points at which the surface does not have a well-defined tangent plane. 
 It's straightforward to see that $Y(Q,Q_1)$ has  exactly four singular points; 
 in the $[\bm w, x_1]$-coordinates these points are:
\begin{align}	\label{Ysings}
 \bm p_1= [ \tilde{ \bm p}_1, 0] = [1,0,0,0,0], \quad \bm p_2 = [ \tilde{\bm p}_2,0] = [0,0,0,1,0], \cr
 \bm p_3 = [\tilde{\bm p}_3,0] = [0,1,0,0,0], \quad \bm p_4 = [\tilde{\bm p}_4,0] = [0,0,1,0,0].
 \end{align}

 From \eqref{w2y},  these points translate back into the $[t, \bm y, \bm r]$ coordinates as 
 \begin{align}	\label{sensorSings}
 \bm p_1 \leftrightarrow  \left[\frac{1}{4a} , -\frac{1}{4}, 0, -\frac{1}{2}, 0 \right] = [1, -a, 0, -2a, 0] \qquad 
 	&\bm p_2 \leftrightarrow \left[\frac{1}{4a} , -\frac{1}{4}, 0, \frac{1}{2}, 0 \right] =  [1, -a, 0, -2a, 0]\cr
 \bm p_3 \leftrightarrow  \left[-\frac{1}{4a} , -\frac{1}{4}, 0, 0, -\frac{1}{2} \right] = [1, a, 0, 0, 2a]   \qquad 
 	&\bm p_4  \leftrightarrow  \left[-\frac{1}{4a} , -\frac{1}{4}, 0, 0, \frac{1}{2}\right] = [1, a, 0, 0, -2a]
 \end{align}
We see that these singularities correspond precisely to the sensor locations, as noted at the end of Section 3.1; allowing both positive and negative value for $r_j$ means that two points in $Y(Q, Q_1)$ correspond to each sensor location.

{\bf From now on, whatever the coordinates we choose, we will call the four singular points of $Y(Q,Q_1)$ by these names:  $\bm p_1,\bm p_2,\bm p_3,\bm p_4$.}

\subsection{The classification of $Y(Q,Q_1)$}\label{classY}

Here, we classify the surface $Y(Q,Q_1)$ in $\mathbb{C}P^4$.  This classification is not used in the rest of the paper, though; it is included here for context.  We will use some vocabulary and theorems from the classification of algebraic surfaces here without much explanation.

The smooth quadric surface $Y(Q)$ in $\mathbb{C}P^3$ is discussed in Appendix \ref{quadric}, and we use the notation given there. 

First, we note that if a smooth surface $S$ in $\mathbb{C}P^4$ (a surface $S$  is ``smooth"  if it has no singularities) is the complete intersection of two quadrics, then surface theory (see, e.g. \cite{AB}) tells us that $S$ is a ``del Pezzo" surface.   

Our surface $Y(Q,Q_1)$ is a complete intersection of two quadrics in $\mathbb{C}P^4$, but it has four singularities. Let's see briefly how we can still classify $Y(Q,Q_1)$ as a del Pezzo surface.

\begin{lemma}
The projection of  real  varieties $\Pi:Y(Q,Q_1) \rightarrow Y(Q)$ defined (everywhere) by
\begin{equation}	\label{PiDef}
\Pi: [\bm w,x_1] \mapsto [\bm w]
\end{equation}
 gives $Y(Q,Q_1)$ as a double cover of  the smooth variety $Y(Q)$ branched over the reduced and reducible curve $L$ 
 defined by $w_0w_3=0$.
 \end{lemma}
 \begin{proof}
The reduced curve $L$  can be written as the union of the four lines 
$L = \ell_1 \cup \ell_2 \cup \ell_3 \cup \ell_4,$
where
$\ell_1 = \{ [w_0,w_1,0,0] \}, \ \ell_2 = \{ [w_0,0,w_2,0] \} , \ \ell_3 = \{[0,0,w_2,w_3] \} , \ \ell_4 =  \{ [0,w_1,0,w_3] \}. $ 
The singularities $\tilde{\bm p}_j$ of $L$, $1 \leq j \leq 4$, are ordinary double points:
 in a neighborhood of  each point $\tilde{\bm p}_j$, $L$ looks like exactly two distinct lines intersecting transversely in the single point $\tilde{\bm p}_j$.

If $[\bm w] \in Y(Q)-L$, then there are exactly two points in the 
 set $\Pi^{-1}([\bm w])$, which is called the  {\em fibre} of $\Pi$ over $[\bm w]$, 
namely  $[\bm w,\pm R]$, where $R$ is any square root of $-w_0w_3$.  

If $[\bm w] \in L$ then there is exactly one point in the fibre over $[\bm w]$, namely $[\bm w,0].$
\end{proof}

The theory of double covers  \cite{BHPV} now tells us a couple of facts:
\begin{itemize}
\item For each $j$, $\bm p_j$ must also be an ordinary double point of  the surface $Y(Q,Q_1)$ (see,e.g., \cite{BHPV}, for the definition of an ordinary double point of a surface).
\item Since $Y(Q)$ is irreducible and the curve $L$ is reduced, the variety $Y(Q,Q_1)$ must also be irreducible.
\end{itemize}

The surface $Y(Q,Q_1)$ has at least four lines lying on it, the lines lying over $L$ via $\Pi$:

$$ \hat{\ell}_1 \doteq \{[w_0,w_1,0,0,0] \mid [w_0,w_1] \in \mathbb{C}P^1 \}, $$
$$\hat{\ell}_2  \doteq \{[w_0,0,w_2,0,0] \mid [w_0,w_2] \in \mathbb{C}P^1\}, $$
$$\hat{\ell}_3 \doteq \{[0,0,w_2,w_3,0] \mid [w_2,w_3] \in \mathbb{C}P^1 \}, $$
\begin{equation} \label{lines}  \hat{\ell}_4 \doteq \{[0,w_1,0,w_3,0] \mid [w_1,w_3] \in \mathbb{C}P^1\}. 
\end{equation}

Note that 
$$\hat{L}  \doteq \cup_{i=1}^4 \hat{\ell}_i = \{[w_0,w_1,w_2,w_3,x_1] \in Y(Q,Q_1) \mid w_0w_3=0\}.$$ 

Using the projection $\Pi$, we can classify the surface $Y(Q,Q_1)$  geometrically.
\begin{theorem}	\label{delPezzo}
$Y(Q,Q_1)$ is a  del Pezzo surface; in other words, $Y(Q,Q_1)$  is a normal surface and  the desingularization of $Y(Q,Q_1)$ is a smooth del Pezzo surface.  
\end{theorem}
\begin{proof} See \cite{BHPV} for elaboration of the undefined terms here.

Now, the line bundle defining the (abstract) double cover of $Y(Q)$, branched over $L$,  is $\mathcal{L} = (1,1)$.  
If $\bar{Y}$ is the surface obtained by canonical resolution of singularities $\sigma:\bar{Y} \rightarrow Y(Q,Q_1)$ (\cite{BHPV}) of $Y(Q,Q_1)$, since $L$ is a $(2,2)$ curve in $Y(Q)$ (see  Appendix \ref{quadric}),  then  surface theory (e.g., see \cite{BHPV}, pp. 236-237)  says that  the canonical divisor $K_{\bar{Y}}$ is the pullback via $\bar{Y} \stackrel{\sigma}{\rightarrow }Y(Q,Q_1) \rightarrow Y(Q)$ of $K_{Y(Q)} + \mathcal{L} = \mathcal{O}_{Y(Q)}(-1,-1)$.  

The computation of $K_{\bar{Y}}$ tells us that, by definition, $S = \bar{Y}$ is a {\em del Pezzo} surface \cite{AB}. \end{proof}

 Surface theory (see, e.g., \cite{AB}) tells us that a smooth del Pezzo surface is an intersection of two quadrics in $\mathbb{C}P^4$ if and only if the surface may be obtained by ``blowing up" five points in $\mathbb{C}P^2$.  Therefore,  every smooth del Pezzo surface  is a {\bf rational } surface.  Since our surface $Y(Q,Q_1)$ is ``birational"  to a smooth del Pezzo surface $\bar{Y}$, it is  also a rational surface.  In the next section, we write down an explicit rational map to $\mathbb{C}P^2;$ this map and its rational inverse will be used later in this paper, for other calculations.

\subsection{Rational maps $\alpha:  Y(Q,Q_1) \dashrightarrow \mathbb{C}P^2$
 and $\beta:   \mathbb{C}P^2 \dashrightarrow Y(Q,Q_1)$ }  \label{alphabeta}

A {\bf rational} map $\gamma:V \dashrightarrow W$ between subvarieties $V,W$ of projective spaces is a function given by homogeneous polynomials of the same degree in the domain variables, which is well-defined (or just ``defined") on an open subset of the form $V-F$ of $V$, where $F$ is a proper subvariety of $V$; the subvariety $F$ could be the empty set.  The use of the dashed arrow is to warn the reader that the map $\gamma$ may not be defined everywhere on $V$.

Surface theory \cite{AB} tells us that $Y(Q,Q_1)$ can be ``described by" cubics in $[u_0,u_1,u_2]$.   This motivates the definitions of the following maps $\alpha$ and $\beta$.

{\bf Define  rational maps $\alpha:Y(Q,Q_1) \dashrightarrow \mathbb{C}P^2$ and $ \beta:\mathbb{C}P^2 \dashrightarrow Y(Q,Q_1)$ by}
\begin{align}	\label{alphadef}
\alpha([w_0,w_1,w_2,w_3,x_1]) &=  [w_3x_1, w_2x_1, w_2w_3], \\ 
\beta([u_0,u_1,u_2]) & = [-u_0u_1^2, -u_0^2u_1,  u_1u_2^2,u_0u_2^2,u_0u_1u_2].	\label{betadef}
\end{align}
Here and throughout this section we use  coordinates $[u_0,u_1,u_2]$ on $\mathbb{C}P^2$, and the coordinates $[\bm w, x_1]$ on $Y(Q,Q_1)$.  

  Note that $\alpha$ is not defined on $\hat{\ell}_1 \cup \hat{\ell}_2 \cup \hat{\ell}_4$ and $\beta$ is not defined at the points
$$[1,0,0],\quad [0,1,0], \quad [0,0,1]  \ \in  \ \mathbb{C}P^2.$$
 In particular, the map $\alpha$ is not defined at points corresponding to the sensor positions. 
Note that the maps $\alpha, \beta$ are {\bf real}, nonlinear maps of real projective varieties:  the polynomials defining each coordinate of the map are homogenous polynomials with real coefficients. 

Now, if $u_0u_1u_2 \neq 0$ (i.e. $[u_0,u_1,u_2]$ does not lie on one of the coordinate lines on $\mathbb{C}P^2$), then
$$\alpha(\beta([u_0,u_1,u_2]) = [u_0^2u_1u_2^3, u_0u_1^2 u_2^3, u_0u_1u_2^4] = [(u_0u_1u_2^3)u_0, (u_0u_1u_2^3)u_1, (u_0u_1u_2^3)u_2] = [u_0,u_1,u_2].$$

Also, suppose that $[w_0,w_1,w_2,w_3,x_1] \in Y(Q,Q_1) - (\cup_{i=1}^4 \hat{\ell}_i)$.  Then, as we've noted, $x_1 \neq 0$.  Since $x_1^2 = w_0w_3$, neither $w_0$ nor $w_3$ can be zero, and since $w_1w_2 = w_0w_3,$ neither $w_2$ nor $w_1$ can be zero.

Therefore, 
$$\beta(\alpha([w_0,w_1,w_2,w_3,x_1])) = [-w_2^2w_3x_1^3,-w_2w_3^2 x_1^3, w_2^3w_3^2x_1, w_2^3w_3^3x_1, w_2^2w_3^2x_1^2];$$ since $x_1^2 = -w_0w_3$ this last point is equal to
\begin{align*}
[w_0w_2^2w_3^2x_1, &w_0w_2w_3^3x_1, w_2^3w_3^2x_1, w_2^2w_3^3x_1, w_2^2w_3^2x_1^2]  \cr
&= [(x_1w_2w_3^2)w_0w_2, (x_1w_2w_3^2)w_0w_3, (x_1w_2w_3^2)w_2^2, (x_1w_2w_3^2)w_2w_3, (x_2w_2w_3^2)w_2x_1]  \cr
& = [w_0w_2, w_1w_2, w_2^2, w_2w_3, x_1w_2] = [w_0,w_1,w_2,w_3,x_1],
\end{align*} where the next-to-last equality holds because $w_1w_2=w_0w_3$.

Let $H_0,H_1, H_2$ be the three coordinate lines in $\mathbb{C}P^2$:  e.g., 
\begin{equation}	\label{coordLines}
H_j  = \{[u_0,u_1,u_2] \in \mathbb{C}P^2 \mid u_j=0\}.
\end{equation}

The above calculations show that
$$\alpha:Y(Q,Q_1) - (\cup_{i=1}^4 \hat{\ell}_i) \rightarrow \mathbb{C}P^2-(H_0 \cup H_1 \cup H_2)$$ is a biholomorphism,  with inverse
$$\beta:\mathbb{C}P^2 - (H_0 \cup H_1 \cup H_2) \rightarrow Y(Q,Q_1) - (\cup_{i=1}^4 \hat{\ell}_i).$$


\subsubsection{Connection between the original coordinates and coordinates $\bm u$ on $\mathbb{C}P^2$}
From the definition of $\beta$ we have, for $[u_0, u_1, u_2] \in \mathbb{C}P^2$ such that $\beta([u_0,u_1,u_2])$ is defined, then
\begin{align}
w_0 = -u_0u_1^2, \qquad w_1 = -u_0^2u_1, \qquad w_2 = u_1u_2^2, \qquad w_3 = u_0u_2^2, \qquad x_1 = u_0u_1u_2
\end{align}
and hence from \eqref{w2y}, we have 
\begin{align}	\label{u2y}
u &= \frac{1}{4} (-u_0u_1^2 +  u_0^2u_1 - u_1u_2^2 + u_0u_2^2), 
\quad y_1 = -  \frac{1}{4} (-u_0u_1^2 - u_0^2u_1 + u_1u_2^2 + u_0u_2^2) , \cr
 \quad  y_2 &= -u_0u_1u_2, \qquad  
 r_1 = \frac{1}{2} (-u_0u_1^2+ u_0u_2^2), \quad r_2 = \frac{1}{2} (u_0^2u_1 + u_1u_2^2) . 
\end{align}
Note that while $\beta$ is not defined at the points $[1,0,0],[0,1,0], [0,0,1]$, then using the coordinates $[\bm w, x_1]$ on $\mathbb{C}P^4$, if
$$[0,u_1,u_2] \in H_0,\; \mbox{forcing}\;u_1u_2 \neq 0, \beta([0,u_1,u_2]) = [0,0,1,0,0] = \bm p_4,$$
$$[u_0, 0,u_2] \in H_1, \;\mbox{forcing}\; u_0u_2 \neq 0, \beta([u_0,0,u_2]) = [0,0,0,1,0] = \bm p_2$$ and if
$$[u_0,u_1,0] \in H_2, \; \mbox{forcing}\; u_0u_1 \neq 0, \beta([u_0,u_1,0]) = [u_1,u_0,0,0,0].$$ In other words,
\begin{itemize}
\item for every point $\bm q$ on $H_0$ where $\beta$ is defined, $\beta(\bm q) = \bm p_4,$
\item for every point $\bm q $ on $H_1$ where $\beta$ is defined, $\beta( \bm q ) = \bm p_2$ and
\item $\beta$ maps $H_2-\{[1,0,0],[0,1,0]\}$ linearly and isomorphically onto the line $\hat{\ell}_1$, minus the two points $\bm p_1, \bm p_3$.
\end{itemize}

If we use the coordinate system $[t,y_1,y_2,r_1,r_2]$ on $\mathbb{C}P^4$,  and consider the projection to $\mathbb{C}P^2$ given by $[t,y_1,y_2,r_1,r_2] \mapsto [t, y_1,y_2]$ (which is also not defined everywhere), then the composite with $\beta$, where defined,  maps $H_0-\{ [0,1,0], [0, 0,1]\}$ to the single point $[1,a, 0]$, while mapping $H_1 - \{ [1,0,0], [0,0,1]\}$ to the single point $[1,-a,0]$. The points on  $H_2-\{[1,0,0], [0,1,0]\}$  map isomorphically to the line in $\mathbb{C}P^2$ defined by $y_2 = 0$, minus the two points $[1,a,0] , [1,-a,0]$.

The maps $\alpha$ and $ \beta$ will be used in Sections 6 and 7.

\subsection{The real points of $Y(Q,Q_1)$.} \label{realY}
In the last Section, we saw that $Y = Y(Q,Q_1)$ is an irreducible rational surface in $\mathbb{C}P^4$.  
This tells us that 
 the real part $Y(\mathbb{R})$ of $Y$  
 is topologically connected \cite{K}.  We can see why this is true explicitly, though. First,  the map $\beta$   restricts to a map
$$ \beta(\mathbb{R}): \mathbb{R}P^2 - \{[1,0,0],[0,1,0],[0,0,1]\} \rightarrow Y(Q,Q_1)(\mathbb{R}).$$   Now, $O \doteq \mathbb{R}P^2-\{[1,0,0],[0,1,0],[0,0,1]\}$ is topologically connected.  Therefore, $\beta(\mathbb{R})(O)$ is a connected subset of $Y(Q,Q_1)(\mathbb{R})$.  However, using our computations in Section \ref{alphabeta}, we see that $\beta(\mathbb{R})(O)$ is a dense subset of $Y(Q,Q_1)(\mathbb{R})$, so that $Y(Q,Q_1)(\mathbb{R})$ is topologically connected.

\section{The Ho-Chen FDOA varieties $HC_F(\bm v,d)$} \label{Xvariety}

 In Section \ref{HCFdef},  given the vectors $\bm v_1, \bm v_2 \in \mathbb{R}^2$ and a real number $d $, setting  $\bm v = (\bm v_1, \bm v_2) = ((v_{11}, v_{12}), (v_{21}, v_{22}))$,  we defined a homogeneous polynomial $\tilde{Q}(\bm v,d)$ of degree two 
 and, in \eqref{FDOAvariety}, the FDOA 
 subvariety $HC_F(\bm v,d)$ of $\mathbb{C}P^4$.  Using the ``original"  coordinates (Section \ref{coords}) on $\mathbb{C}P^4$, the polynomials $L_1,L_2$ become
$$L_1 =  v_{11}(at-y_1)- v_{12}y_2 , \qquad L_2 = v_{21}(-at-y_1) - v_{22} y_2;$$ while $\tilde{Q}$ is now
$$\tilde{Q}  = L_2r_1-L_1r_2-d r_1r_2.$$  

In other words,  the subvariety $HC_F(\bm v,d)$ of $Y(Q_1,Q_2)$ is defined as
\begin{equation}
HC_F(\bm v,d) \doteq \{ [\bxi ]= [t,y_1,y_2,r_1,r_2] \in Y(Q_1,Q_2) \mid \tilde{Q}(\bxi) = 0 \}.
\end{equation}

\begin{remark}  In order to define $HC_F(\bm v,d)$ as a subvariety of $\mathbb{C}P^4$, it is not necessary to assume that 
$$v_{11},v_{12},v_{21},v_{22},d$$ are real numbers.  For applications  such as the source localization problem, we do require these parameters to be real. In the following, however,  we don't assume these parameters are real numbers, unless we are specifically interested in this case.  If the parameters $\bm v,d$ are real, then the variety $HC_F(\bm v,d)$ is a real variety.

Furthermore, if $\bm v_1 = \bm v_2 = \bm 0, d=0$, then $HC_F(\bm 0, \bm 0, 0) = Y(Q,Q_1)$, and this is clearly not a curve.  Thus, from now on, we do not consider this case.  Another way of saying this is:  we consider only parameter choices $[\bm v_1, \bm v_2, d]$ in $\mathbb{C}P^4$, or, if we want all parameters to be real, $[\bm v_1, \bm v_2, d] \in\mathbb{R}P^4$.  In other words, the parameters themselves come from a  projective parameter space.
\end{remark}

Our variety $HC_F(\bm v,d)$ is an intersection of three quadric hypersurfaces in $\mathbb{C}P^4$.   So, we ``expect" the Ho-Chen FDOA variety to be a curve of degree 8 in $\mathbb{C}P^4$: this means that, if we intersect this curve with a hyperplane in $\mathbb{C}P^4$, we expect to get 8 points.  When we ``count" these intersection points, we must acknowledge:  1) if the hyperplane contains the entire curve, or even a component of the  curve with an infinite number of points, we don't  get a finite number of points and 2) if we limit ourselves to hyperplanes containing no component of the curve, we must ``count" with ``multiplicity".

For some choices of parameters $\bm v, d$, $HC_F(\bm v,d)$ will be a {\bf reducible} variety; in other words, the union of two or more proper subvarieties.  One of the principle aims of Section \ref{secIrred}  is to prove that, for most choices of $\bm v, d$, $HC_F(\bm v,d)$ is an irreducible variety.  What do we mean, here, by ``most'' choices of $\bm v,d$? To be precise, this means that there is a finite set $\mathcal{F}$ of homogeneous polynomials in the parameters $v_{11}, v_{12}, v_{21},v_{22},d$, such that for every selection  $\bm v,d$ NOT in the solution set of this finite set of polynomials,  $HC_F(\bm v,d)$ is irreducible.  Another way of saying this is:  for {\bf generic} choices of $\bm v,d$, or {\bf generically}, $HC_F(\bm v,d)$ is irreducible.  We will not usually specify exactly what this finite set of polynomials $\mathcal{F}$ is; an exception to this is the extended example discussed in Section \ref{specialXvariety}.

 In this section we collect various properties of $HC_F(\bm v,d)$.

\subsection{Changes of coordinates}
Here, we write down the equations defining $HC_F(\bm v,d)$ in our other coordinate systems on $\mathbb{C}P^4$.

In the $[\bm w, x_1]$-coordinates,  $\tilde{Q}$ is
 $$\tilde{Q} = A(\bm w)x_1  + C( \bm w),$$ where
 \begin{align*}
 A(\bm w) &= 2(v_{22}(w_3-w_0)-v_{12}(w_2-w_1)), \cr
 C(\bm w) & = v_{21}(w_1+w_2)(w_3-w_0) - v_{11}(w_0+w_3)(w_2-w_1)-d(w_3-w_0)(w_2-w_1) \cr
  &= (v_{11} -v_{21}-d)w_0w_1 + (-v_{11}-v_{21} +d)w_0w_2 + (v_{11}+ v_{21} +d)w_1w_3 \cr
  	& \qquad + (-v_{11} + v_{21} -d)w_2w_3\cr
 & = a_1 w_0w_1 + a_2 w_0w_2 + a_3 w_1 w_3 + a_4 w_2 w_3,
  \end{align*}
  where
\begin{equation} \label{adef} 
a_1 = v_{11}-v_{21}-d,  \quad a_2 = -v_{11}-v_{21}+d, \quad a_3 = v_{11}+v_{21} + d, \quad a_4 = -v_{11}+v_{21}-d. \end{equation}
   In these coordinates,
 $$HC_F(\bm v,d) = \{ [\bxi ]= [w_0,w_1,w_2,w_3,x_1] \in Y(Q,Q_1) \mid \tilde{Q}(\bxi) = 0 \};$$ in other words,
 $HC_F(\bm v,d)$ is the set of points  in $\mathbb{C}P^4$  such that

$$w_0w_3=w_1w_2,  \qquad x_1^2 + w_0w_3 = 0$$ and
$$A(\bm w)x_1 + C(\bm w)= 0.$$

We will  usually use the $[\bm z,x]$-coordinates on $HC_F$ in order to discuss singular points, since the matrix of partials is easier to manage with these coordinates.  In these coordinates,    $HC_F(\bm v,d)$ is the subvariety of $\mathbb{C}P^4$ defined by (with a small abuse of notation in the use of  the symbols $Q,Q_1$):
 $$Q = z_2^2-z_0^2 - (z_3^2-z_1^2) = 0,$$
 $$Q_1 = x^2 +(z_2^2-z_0^2)=0,$$
 $$\tilde{Q}_1 = (v_{22}z_0-v_{12}z_1)x + (v_{21}z_0z_3-v_{11}z_1z_2 - dz_0z_1)=0;$$
 setting $$C_1 = v_{21}z_0z_3-v_{11}z_1z_2-dz_0z_1, \qquad A_1 = v_{22}z_0-v_{12}z_1,$$ this last equation is
 $$\tilde{Q}_1 = A_1x + C_1=0.$$

 \subsection{General remarks on $HC_F(\bm v,d)$}

We will consider all varieties $HC_F(\bm v,d)$, for varying choices of $[\bm v,d] \in \mathbb{C}P^4$, and refer to this collection of varieties as  a {\bf family} $\mathcal{H}$ of subvarieties of $Y(Q_1,Q_2)$.
    
 Note that if $[\bm v_1,\bm v_2,d] \in \mathbb{C}P^4$, then the polynomial $\tilde{Q} _1= A_1(z_0,z_1) x + C_1(z_0,z_1,z_2,z_3)$ (using the $[\bm z,x]$-coordinates) is not the zero polynomial:  the distinct monomials of the form $xz_j$ or $z_kz_l$ which can possibly occur in the nonzero summands of of the degree-two polynomial $A_1x + C_1$, when written out, form a linearly independent set of monomials.
  
  \begin{lemma} 	\label{irredFDOAs}
  For every $[\bm v_1,\bm v_2,d] \in \mathbb{C}P^4$,  every irreducible component of $HC_F(\bm v_1,\bm v_2,d)$ is an irreducible curve in $Y(Q_1,Q_2)$.
\end{lemma}
\begin{proof}  We have seen that the projective variety $Y =Y(Q_1,Q_2)$ is an irreducible variety of dimension 2. Since $HC_F = HC_F(\bm v_1,\bm v_2,d)$ is defined by one nonzero homogeneous polynomial equation $\tilde{Q}$ on $Y$, every irreducible component of $HC_F$ has dimension at least one, using Krull's Principal Ideal Theorem and Hilbert's Nullstellensatz.  However, again,  since $Y$ is irreducible, the only way any irreducible component of $HC_F$ could have dimension 2 would be if that component were equal to $Y$, in which case we would have $HC_F = Y$ as well.  Thus, we need only produce points on $Y$ that do not lie on $HC_F$.  Use the $[\bm z,x]$-coordinates.

If the polynomial $C_1 = v_{21}z_0z_3-v_{11}z_1z_2-dz_0z_1$ is not the zero polynomial, then at least one of $v_{21},v_{11}, d$ is nonzero, so at least one of $v_{21}-v_{11}-d, \ v_{21} + v_{11} +d, \ -v_{21} + v_{11}-d$ is nonzero.   If $v_{21}-v_{11} -d \neq 0$, then $\bxi = [1,1,1,1,0] \in Y$, $C_1(\bxi) \neq 0$ and $\tilde{Q}_1(\bxi) = C_1( \bxi) \neq 0$; similarly,  if $v_{21} + v_{11} + d, -v_{21} + v_{11} -d$, respectively,  is nonzero, then $[1,-1,1,1,0] , [-1,-1,1,1,0] $, respectively,  are points in $Y$ for which $\tilde{Q}(\bxi) \neq 0$.  All of these points have real coordinates.

If $v_{11} = v_{21} = d = 0$, then  $\tilde{Q}_1(z,x) = A_1(z_0,z_1)x$ (as a polynomial).  In this case, we must have $v_{12} \neq 0$ or $v_{22} \neq 0$.   In either case, we may choose a real number $e >1$ such that $v_{22}e + v_{12} \neq 0$. Then, for example, one can check that 
$$[e,1,\sqrt{e^2-1}, 0 , 1]\  \in  \ Y(Q,Q_1) - HC_F((0,v_{12}),(0,v_{22},0).$$   Note that this point has real coordinates.
 \end{proof}

The argument of Lemma \ref{irredFDOAs}  
fails for real varieties; for example, $\mathbb{R}P^2$, with projective coordinates $[u,v,w]$,  is an irreducible variety of dimension 2 over the real numbers, but the subvariety defined by the nonzero homogeneous polynomial $u^2+(v+w)^2$ has only one point $[0,1,-1]$ in it.

Therefore,   
since the polynomials $\tilde{Q}$ have coefficients varying linearly with respect to the parameters $[\bm v,d]$:  

\begin{corollary} The family of varieties $HC_F(\bm v_1,\bm v_2,d)$, as $[\bm v_1,\bm v_2,d]$ varies over $\mathbb{C}P^4$, defines a linear system $\mathcal{H}$ whose members are (possibly reducible) curves on $Y(Q_1,Q_2)$.  \end{corollary}


\subsubsection{Base points of the family $HC_F(\bm v_1,\bm v_2,d)$ and Bertini's Theorems} 
The {\bf fixed locus} of the family $\mathcal{H}$ (or, more generally, of any family of subvarieties of a variety $Y$) is the set of points in the ambient variety $Y$ which lie on every variety in the family. 
 Sometimes the fixed locus of a collection of curves is called the {\bf base locus} of the system; if the fixed locus is a finite set of points, the points in the fixed locus are called the {\bf base points} of the collection of curves.
  
   For the source localization problem, it is important to know the real base points, because they always appear in the set of possible source locations, and thus can generally be eliminated from  consideration.

 It's not hard to see  that the fixed locus of the family of curves $HC_F(\bm v,d)$ on $Y$ consists of eight points.  In the $[\bm z, x]$-coordinates on $Y$, these eight points are:
  $$[-1,0,1,0,0],\qquad [0,-1,0,1,0], \qquad [0,1,0,1,0],  \qquad [1,0,1,0,0]$$ 
  (which are the points $\bm p_1,\bm p_2,\bm p_3,\bm p_4$  and, as before, correspond to the sensor positions) and 
    $$ [0,0,1,1,\sqrt{-1}],\quad [0,0,1,1,-\sqrt{-1}], \quad [0,0,1,-1, \sqrt{-1}],  \quad  [0,0,1,-1, -\sqrt{-1}].$$   
  The first four points are always singular points of $HC_F$  (since they are singular points on the ambient variety $Y(Q_1,Q_2)$)  and the last four points, which are not real points,  are usually not singular points of $HC_F$. 
  
   Another reason to determine the base points is in order to make use of the {\bf Bertini Theorems}; see \cite{EOM}, for example, for the statements of these theorems, as well as references for further discussions of the theorems and their proofs.   In  reference \cite{EOM}, statement 1) will be called {\bf Bertini's Theorem-Generic Irreducibility or BTGI}, and statement 2) will be called {\bf Bertini's Theorem--Generic  Smoothness or BTGS } in this paper.  Note that we may attempt to apply these theorems to the family of curves $\mathcal{H}$, since, as we have seen, this is a linear system of curves on the variety $Y(Q_1,Q_2)$ and, as such, constitutes a linear system of ``divisors" on $Y(Q_1,Q_2)$.  
     
\begin{remark}	\label{singsRbase} We've determined the base points of the linear system of curves $\mathcal{H}$ on $Y(Q_1,Q_2)$; since these basepoints include the singular points of the variety $Y(Q_1,Q_2)$, we may use BTGS to conclude:  for generic choices of $[\bm v_1,\bm v_2,d]$, the singularities of  the curve $HC_F(\bm v_1,\bm v_2,d)$ lie amongst the base points of the system, namely, amongst the eight points listed above. 
\end{remark}

Consequently, to determine the singularities of a generic curve $HC_F(\bm v_1,\bm v_2,d)$, we need only check the base points.

 In fact, we have
\begin{lemma}  \label{foursingsH} In the $[\bm z, x]$-coordinates, the points $\bm p_1,\bm p_2,\bm p_3,\bm p_4$ are the points 
$$[-1,0,1,0,0],\qquad [0,-1,0,1,0], \qquad [0,1,0,1,0], \qquad [1,0,1,0,0]$$ 
and are singular points of every curve $HC_F(\bm v,d)$.  The points 
$$ [0,0,1,1,\sqrt{-1}], \quad [0,0,1,1, \quad -\sqrt{-1}], \quad [0,0,1,-1, \sqrt{-1}], \quad  [0,0,1,-1, -\sqrt{-1}]$$ 
are not singular points of $HC_F(\bm v,d)$, if
$$v_{j1}^2 + v_{j2}^2 \neq 0$$ for $j = 1$ or for $j =2$.  Therefore, using BTGS,  for generic choices of $[\bm v,d]$, the only singular points of $HC_F(\bm v, d)$ are the first four points: $\bm p_1,\bm p_2,\bm p_3,\bm p_4$.  Moreover, the conditions determining ``generic" must include $v_{j1}^2 + v_{j2}^2 \neq 0$ for $j = 1$ or for $j =2$. 
\end{lemma}
\begin{proof}
The first four points are singular points of $HC_F(\bm v,d)$ since they are singular points on the ambient variety $Y(Q,Q_1)$.

As for the last four points, we need to compute the Jacobian of the set of three equations defining $HC_F(v_1,v_2,d)$; after eliminating unnecessary multipliers, this is the matrix
$$\left[ \begin{array}{ccccc}
-z_0 & z_1 & z_2 & -z_3 & 0 \\
-z_0 & 0 & z_2 & 0 & x \\
v_{22}x + v_{21}z_3-dz_1 & -v_{12}x-v_{11}z_2-dz_0 & -v_{11}z_1 & v_{21}z_0 & v_{22}z_0-v_{12}z_1 \end{array} \right].$$

Evaluating this matrix at the points $[0,0,1,\pm 1, \pm  \sqrt{-1}]$, we have
$$\left[ \begin{array}{ccccc}
0 & 0 & 1 & - (\pm 1) & 0 \\
0 & 0 & 1 & 0 & \pm \sqrt{-1} \\
\pm v_{22}\sqrt{-1} \pm  v_{21} & -(\pm v_{12})\sqrt{-1}-v_{11} & 0 & 0 & 0 \end{array} \right].$$

Looking at columns 4,5,2 or columns 4,5,1, according as to whether $v_{j1}^2 + v_{j2}^2 \neq 0$ for $j = 1$ or for $j =2$, we get submatrices that are invertible.  

\end{proof}

\begin{remark} For those who are interested, we may invoke some theorems from algebraic geometry now.  See, for example, \cite{AB}, \cite{BHPV}, \cite{GH} for statements and discussion of these theorems; it is not within the scope of this paper to make these statements and discussions here.  As we've seen, $HC_F(\bm v,d)$ is an intersection of three quadrics in $\mathbb{C}P^4$.  The {\bf adjunction formula} and the {\bf Riemann-Roch Theorem} (see cited references) then tell us, that if $HC_F$ were smooth, its genus would be
$$1 + 1/2(2 + 2+ 2-5)(2\cdot 2 \cdot 2)  = 5.$$ However, generically, $HC_F$ has exactly four singularities, which we suspect are all {\bf nodes}.  In this case, using the genus formula for a curve with such singularities, the genus of this generic curve should be
$$5-(1 + 1 + 1 + 1) = 1.$$

We will compute this genus in other ways in this paper.
\end{remark}

\subsection{Points on $HC_F(\bm v,d)$ where $z_0z_1=0$}

In this section, we  consider the points $[\bm z, x] \in HC_F(\bm v, d)$ where $z_0z_1 = 0$, if $$v_{j1}^2 + v_{j2}^2 \neq 0$$ for $j = 1$ and $j=2$.   Note that, in the original coordinates
$[u, \bm y, r_1,r_2]$, this is the set of points in $HC_F$ where $r_1r_2 = 0$.

\begin{lemma} \label{r1r2zeroH} Suppose that $v_{j1}^2 + v_{j2}^2 \neq 0$ for $j = 1$ and $j=2$.  Using the $[\bm z, x]$-coordinates,    the only points on $HC_F( \bm v, d)$ with $z_0z_1=0$ are the eight basepoints
$$[-1,0,1,0,0],\qquad [0,-1,0,1,0], \qquad [0,1,0,1,0], \qquad [1,0,1,0,0]$$ 
(the points $\bm p_1,\bm p_2,\bm p_3, \bm p_4$) and
  $$ [0,0,1,1,\sqrt{-1}], \quad [0,0,1,1,-\sqrt{-1}], \quad [0,0,1,-1, \sqrt{-1}],  \quad [0,0,1,-1, -\sqrt{-1}].$$  
  
 Restating this using the original coordinates, the set of points $[u, \bm y, r_1,r_2] \in HC_F(\bm v,d)$ where  $r_1r_2 =0$ is equal to the set of eight basepoints.  
\end{lemma}

\begin{proof}  Suppose that $[\bm z,x] \in HC_F(\bm v,d)$ and $z_0z_1=0$.

If $z_0=z_1=0$, then $x^2 = \pm \sqrt{-1}z_2, z_2^2 = z_3^2, A_1 = 0, C_1 = 0.$    Therefore $[\bm z,x] $ is one of the four points $ [0,0,1,1,\sqrt{-1}],[0,0,1,1,-\sqrt{-1}], [0,0,1,-1, \sqrt{-1}],  [0,0,1,-1, -\sqrt{-1}].$

If $z_0 = 0, z_1 \neq 0$, then $x = \pm \sqrt{-1}z_2$, $A = -v_{12} z_1, C = -v_{11}z_1z_2$, so, using $\tilde{Q}_1 = 0$,  
$$0 = (-v_{12}z_1)(\pm \sqrt{-1}z_2) - v_{11}z_1z_2 =  -(v_{11} \pm \sqrt{-1}v_{12})z_1z_2.$$  Since $v_{11}^2 + v_{12}^2 \neq 0, z_1 \neq 0$, we must have $z_2 = 0$.  This means that $x = 0$ and, since $z_3^2-z_1^2 = 0$, 
$[\bm z,x]$ is one of the two points $\bm p_2,\bm p_3$. Similarly, if $z_0 \neq 0, z_1 = 0$, then, since $v_{12}^2 + v_{22}^2 \neq 0$, we get one of the two points
$\bm p_1, \bm p_4$.
\end{proof}

This Lemma is used in the proof of Lemma \ref{HintersectsZ} in the next section.

\subsection{Intersections of $HC_F(\bm v,d)$ with the lines $\hat{\ell_j}$}

One can compute that the lines $\hat{\ell}_j$, defined in Section \ref{classY}, on $Y(Q,Q_1)$, in the $[\bm z, x]$-coordinates,  are 
\begin{align*}
\hat{\ell}_1 &\doteq \{[z_0,z_1,-z_0, -z_1,0] \mid [z_0,z_1] \in \mathbb{C}P^1 \}, \cr
\hat{\ell}_2  &\doteq \{[z_0,z_1, -z_0, z_1,0] \mid [z_0,z_1] \in \mathbb{C}P^1\},\cr
\hat{\ell}_3 &\doteq \{[z_0,z_1,z_0,z_1,0] \mid [z_0,z_1] \in \mathbb{C}P^1 \}, \cr 
\hat{\ell}_4 & \doteq \{[z_0,z_1,z_0,-z_1,0] \mid [z_0,z_1] \in \mathbb{C}P^1\}. 
\end{align*}

Now, the only remaining equation that a point on one of these lines must satisfy is the equation
$$A_1 x + C_1 = 0,$$ where
$$C_1 = v_{21}z_0z_3-v_{11}z_1z_2-dz_0z_1, \qquad A_1 = v_{22}z_0-v_{12}z_1.$$   
However, since $x=0$ for any point on any $\hat{\ell}_j$, we need only decide when, for a point $[\bm z,0]$ on some $\hat{\ell}_j$,
$$C_1(\bm z) = 0.$$

Evaluating $C_1$ at $[\bm z] \in \hat{\ell}_j$, we  must have, respectively,
$$a_1z_0z_1 = (-v_{21}+v_{11} -d)z_0z_1=0,$$
$$a_2 z_0z_1 = -(v_{21} + v_{11} -d)z_0z_1 = 0,$$
$$a_4 z_0z_1 =(v_{21} - v_{11} -d)z_0z_1 = 0,$$
$$a_3 z_0z_1 = (v_{21} + v_{11} + d)z_0z_1=0.$$  (See \eqref{adef} for the definition of the $a_j$s.)

Therefore,  using Lemma \ref{r1r2zeroH}, 
\begin{itemize}
\item If $a_1 \neq 0$, $HC_F \cap \hat{\ell}_1 = \{[0,1,0,-1], [1,0,-1,0]\}$; if $a_1 = 0$,  $\hat{\ell}_1$ is a component of $HC_F$.
\item If $a_2 \neq 0$, $HC_F \cap \hat{\ell}_2 = \{[0,1,0,1],[1,0,-1,0]\};$ if $a_2 = 0$, $\hat{\ell}_2$ is a component of $HC_F$.
\item If $a_4 \neq 0$, $HC_F \cap \hat{\ell}_3 = \{[0,1,0,1],[1,0,1,0]\};$ if $a_4=0$, $\hat{\ell}_3$ is a component of $HC_F$.
\item If $a_3 \neq 0$, $HC_F \cap \hat{\ell}_4 = \{0,1,0,-1],[1,0,1,0]\};$ if $a_3 = 0$, $\hat{\ell}_4$ is a component of $HC_F$.
\end{itemize}
We summarize this analysis as follows:
\begin{lemma} \label{HintersectsZ} If $a_j \neq 0, 1 \leq j \leq 4$,  then 
  $HC_F(\bm v,d) \cap (\cup_{j=1}^4 \hat{\ell}_j) = \{ \bm p_1, \bm p_2, \bm p_3, \bm p_4\},$ these four points are also singular points of $HC_F$.
\end{lemma}

This lemma will be used in the proof of Corollary \ref{Hirred}.

\section{Geometric analysis of the curves $HC_F(\bm v,d)$}	\label{secIrred}

One of the Bertini theorems in reference \cite{EOM}, namely BTGI, tells us that we expect that, generically, the curve $HC_F(\bm v, d)$, is reduced and irreducible.  We could invoke BTGI here, after verifying the hypotheses of that theorem.  However, in this section, we will prove the generic irreducibility of $HC_F(\bm v,d)$ without invoking BTGI.  The methods here have independent interest, and also give some idea about what conditions ``generic" must include.  Note that the discussion preceding Lemma \ref{HintersectsZ} tells us that $HC_F(\bm v,d)$ is not irreducible if one of the quantities $a_j$ is zero.

The strategy for the proof is to study the images of the curves $HC_F$ under the map $\alpha$ of section \ref{alphabeta}.   
The algebraic varieties corresponding to these images comprise a linear system $\mathcal{V}$ of plane quartics, which we show are generically irreducible.   
Because, as shown in section \ref{alphabeta},  the maps $\alpha$ and $\beta$ provide a birational equivalence between curves $HC_F(\bm v, d)$ and the curves $V$ of $\mathcal V$, we can conclude that the curves $HC_F(\bm v, d)$  on $Y(Q,Q_1)$ are also generically irreducible.


\subsection{ Birational equivalence of $HC_F$ with $V$}		\label{biratFDOAwithV}

In this subsection, we examine the images $V(\bm v,d)$ of  the curves $HC_F(\bm v, d)$ under the mapping $\alpha:  \mathbb{C}P^4 \dashrightarrow \mathbb{C}P^2$ of  \eqref{alphadef}.    

We show in Lemmas \ref{alphaonHC}  and \ref{betatoHC} below that $V(\bm v,d)$ is birationally equivalent to $HC_F(\bm v,d)$.   
Consequently, properties of $HC_F$ can be inferred from the properties of $V$.

\subsubsection{The map $\alpha:HC_F(\bm v,d) \dashrightarrow \mathbb{C}P^2$ }


Let's recall the definition of $\alpha:Y(Q,Q_1) \dashrightarrow \mathbb{C}P^2$ from \eqref{alphadef}, using the $[\bm w,x_1]$-coordinates:
$$\alpha([w_0,w_1,w_2,w_3,x_1]) =  [w_3x_1, w_2x_1, w_2w_3];$$  $\alpha$ is not defined on $\hat{\ell}_1 \cup \hat{\ell}_2 \cup \hat{\ell}_4$ and $\alpha$ sends the part of the line $\hat{\ell}_3$ on which it is defined to $[0,0,1]$. 

The rational map $\alpha$ restricts to a rational map, also called $\alpha$,
$$\alpha:HC_F(\bm v,d) \dashrightarrow \mathbb{C}P^2.$$

Define an open subset $O_1(\bm v,d)$ of $HC_F(\bm v,d)$ by
\begin{equation} \label{O1}
O_1(\bm v,d) = HC_F(\bm v,d) - (\hat{\ell}_1 \cup \hat{\ell}_2 \cup \hat{\ell}_4);
\end{equation} 
$O_1$ is the domain of definition of $\alpha$ (when restricted to $HC_F$).   Moreover, we have
\begin{lemma}  \label{genericdomainalpha} For generic choices of $[\bm v, d]$, 
$$O_1(\bm v,d)  = HC_F(\bm v, d) - \{ \bm p_1, \bm p_2, \bm p_3, \bm p_4\},$$  every point of $O_1$ is a smooth point of $HC_F(\bm v,d)$ and $\alpha$ is  injective when restricted to $O_1$.
\end{lemma}
\begin{proof}  This follows directly from Lemma \ref{HintersectsZ}, Lemma \ref{foursingsH} and the discussion in Section 4.3.  We point out that those Lemmas specify exactly what conditions ``generic" must include.
\end{proof} 

We will use the $[\bm w,x_1]$-coordinates on $HC_F(\bm v,d) \subset \mathbb{C}P^4$ and the coordinates $[u_0,u_1,u_2]$ on $\mathbb{C}P^2$, as in Section \ref{alphabeta},  to study $\alpha$ on $HC_F$.

\subsubsection{The linear system $\mathcal V$ of plane curves}
 In this subsection, we define  the family $\mathcal V$ of varieties $V(\bm v,d)$ as the image, under $\alpha$, of the family $\mathcal{H}$.   

We show that the maps  $\alpha$ and $\beta$ of section \ref{alphabeta}, restricted to $HC_F$ and $\mathcal V$, respectively, provide a birational equivalence between curves $V(\bm v, d)$ and the corresponding curves $HC_F(\bm v, d)$.   Consequently, generic irreducibility of $V(\bm v,d)$ implies irreducibility of the corresponding curve $HC_F(\bm v, d)$.  

 Lemma \ref{betatoHC} below and its proof motivate the definition of the following
 homogeneous polynomial $P(\bm v,d)(u_0,u_1,u_2)$ of degree 4:  
\begin{align}
P(\bm v, d)(u_0,u_1,u_2) & = 2v_{22}u_0u_2(u_1^2+u_2^2) - 2v_{12}u_1u_2(u_0^2+u_2^2)   \\
 & + v_{21}(u_2^2-u_0^2)(u_2^2+u_1^2)-v_{11}(u_2^2-u_1^2)(u_0^2+u_2^2)-d(u_1^2+u_2^2)(u_0^2+u_2^2). \nonumber
 \end{align} 
Recalling the definitions of the $a_j$s \eqref{adef}, we may rewrite $P$ as
$$P =  X_1 u_0^2 + 2v_{22}(u_2^3 + u_1^2u_2) u_0 + X_2 u_2^2,$$ where, 
$$X_1 = -a_3u_2^2-2v_{12}u_1u_2 + a_1u_1^2,  \qquad X_2 = a_4u_2^2 -2v_{12}u_1u_2 -a_2u_1^2.$$  
Alternatively,
\begin{equation}	\label{polyP}
P  = a_4u_2^4 + 2(v_{22}u_0-v_{12}u_1)u_2^3 -(a_3u_0^2 +a_2u_1^2)u_2^2 + 2u_0u_1(v_{22}u_1 - v_{12}u_0)u_2 + a_1u_0^2u_1^2.
\end{equation}

The polynomial $P$, in turn, defines the 
 family of varieties $\mathcal{V} = \{ V(\bm v,d) \mid [\bm v,d] \in \mathbb{C}P^4\}$, where
\begin{equation}
V(\bm v,d) = \{ [u_0,u_1,u_2] \in \mathbb{C}P^2 \mid P(\bm v,d)(u_0,u_1,u_2) = 0 \}.
\end{equation} 
The family $\mathcal V$ 
is a linear family  in $\mathbb{C}P^2$; since at least one of $v_{11},v_{12},v_{21},v_{22},d$ is nonzero, this is a linear system of quartic plane curves (not always irreducible curves).  In other words, for every choice of parameters $[\bm v,d] \in \mathbb{C}P^4$, $V(\bm v,d)$ is a subvariety of  the plane whose irreducible components all have (complex) dimension one.

Let $O_1$ be defined as in \eqref{O1}.

\begin{lemma}  \label{alphaonHC}
The closure of $\alpha(O_1(\bm v,d))$ in $\mathbb{C}P^2$ is  a subset of $V(\bm v,d)$ and 
$$\alpha:HC_F(\bm v,d) \dashrightarrow V(\bm v,d)$$ is a rational map.   
\end{lemma}	
\begin{proof}  Use the $[\bm w, x_1]$-coordinates on $\mathbb{C}P^4$;  if $[\bm w,x_1] \in O_1(\bm v,d)$, then $\alpha([\bm w,x_1]) = [w_3x_1,w_2x_1, w_2w_3]$, and we must verify that 
$P(\bm v,d)(w_3x_1,w_2x_1,w_2 w_3) = 0$, which is a straightforward calculation:  
\begin{align*}	
P(w_3x_1, &w_2x_1,w_2w_3) =
 w_2^2w_3^2(2v_{22}w_2x_1(x_1^2+w_3^2) - 2v_{21}w_3x_1(x_1^2 + w_2^2)  \cr
&  + v_{21}(w_2^2-x_1^2)(w_3^2+x_1^2)-v_{11}(w_3^2-x_1^2)(x_1^2+w_2^2) -d(x_1^2+w_3^2)(x_1^2+w_2^2)).
\end{align*}
  Since $w_0w_3=w_1w_2, x_1^2 = -w_0w_3=-w_1w_2$, making substitutions appropriately,  the expression above is equal to
$$w_2^3w_3^3(2v_{22}x_1(w_3-w_0)-2v_{12}x_1(w_2-w_1) - v_{11}(w_0+w_3)(w_2-w_1) - d(w_3-w_0)(w_2-w_1)) $$
$$= w_2^3w_3^3(A(\bm w)x_1 + C(\bm w)) = 0.$$ Thus, $\alpha(O_1) \subseteq V$; since $V$ is closed, the closure of $\alpha(O_1)$ in the plane must be contained in $V$. 

\end{proof}

Similarly, the map $\beta: \mathbb{C}P^2 \dashrightarrow Y(Q,Q_1)$ of Section \ref{alphabeta} restricts to 
$$\beta:V(\bm v,d) \dashrightarrow  Y(Q,Q_1);$$
$$\beta([u_0,u_1,u_2]) =  [-u_0u_1^2, -u_0^2u_1,  u_1u_2^2,u_0u_2^2,u_0u_1u_2].$$ Now, $\beta$ is defined everywhere on the plane except for the points $[1,0,0],[0,1,0],[0,0,1]$; if 
$$\tilde{U}(\bm v,d) = V(\bm v,d) -\{[1,0,0],[0,1,0],[0,0,1] \},$$ then the open subset $\tilde{U}$ of $V(\bm v,d)$ must be dense in $V(\bm v,d)$. 

\begin{lemma} 	\label{betatoHC}
The closure of $\beta(\tilde{U}(\bm v,d))$ is a subset of $HC_F(\bm v,d)$ and 
$$\beta:V(\bm v,d) \dashrightarrow HC_F(\bm v,d)$$ is a rational map.
\end{lemma}
\begin{proof}  If $[u_0,u_1,u_2] \in \tilde{U}$, then we know that 
$$\beta([u_0,u_1,u_2]) = [-u_0u_1^2, - u_0^2u_1, u_1u_2^2, u_0u_2^2, u_0u_1u_2]  = [w_0,w_1,w_2,w_3,x_1] \in Y(Q,Q_1).$$  So, we only need to verify that $A(\bm w)x_1 + C(\bm w) = 0$.  Computing,
$$A(\bm w)x_1 + C(\bm w) = $$
$$2(v_{22}(u_0u_2^2 + u_0u_1^2) -v_{12}(u_1u_2^2 + u_0u_1^2))u_0u_1u_2 + a_1u_0^3u_1^3 - a_2u_0u_1^3u_2^2 -a_3u_0^3u_1u_2^2 + a_4a_0u_1u_2^4$$
$$= u_0u_1P(u_0,u_1,u_2) = 0.$$

Therefore, $\beta(\tilde{U}) \subseteq HC_F$; since $HC_F$ is closed, we are done.
\end{proof}

\subsection{Properties of the linear system $\mathcal V$ of plane curves}	\label{systemV}

In this subsection we study the linear system $\mathcal V$ and determine (in subsection \ref{summaryV}) that generically, the varieties $V \in \mathcal V$ are irreducible.  

Up to this point, we've been presenting some details of calculations as well as how the calculations lead to conclusions about geometry.  We'll omit details at times in subsequent sections; often, the computational details and their geometric conclusions  can be made using a software package that can do symbolic computation.  We take the point of view that it's still useful to see some details, even so.
 
\subsubsection{Intersections of $V$ with the coordinate lines  in $\mathbb{C}P^2$ } \label{Vintersectslines}

Understanding the intersections of $V$ with coordinate lines gives information about the varieties $V$, contributes to the proof of Corollary \ref{Hirred}, and is needed in section \ref{desing}.

We recall that the coordinate lines  $H_j$ are defined in \eqref{coordLines}.  
If $[0,u_1,u_2] \in V \cap H_0$, evaluating $P$ at the points $[0,u_1,u_2]$, we get
$$0 =  u_2^2 X_2   = u_2^2(a_4u_2^2 -2v_{12}u_1u_2 -a_2u_1^2).$$  If $u_2 = 0$, we get the point $[0,1,0]$ on  $V(\bm v,d)$.  Otherwise, we must have 
$$a_4u_2^2 -2v_{12} u_1u_2 -a_2u_1^2 = 0.$$  

Similarly, if $[u_0,0, u_2] \in V \cap H_1$, evaluating $P$ at the points $[u_0,0,u_2]$, we get
$$0= a_4u_2^4 + 2v_{22}u_0u_2^3  -a_3u_0^2u_2^2 = u_2^2(a_4u_2^2 + 2v_{22}u_0u_2 - a_3u_0^2).$$ If $u_2 = 0$, we get the point $[1,0,0]$; otherwise,
$$a_4u_2^2 +2v_{22}u_0u_2 -a_3u_0^2 = 0.$$ 

Finally, if $[u_0,u_1,0] \in V \cap H_2$, we must have
$$0 = a_1u_0^2u_1^2.$$

Particular cases of interest  are given by: 

\begin{lemma}  \label{lemmaVintersectslines} If  $a_j \neq 0$ for $1 \leq j \leq 4,  v_{12}^2 + a_2a_4 \neq 0, v_{22}^2 + a_3a_4 \neq 0$,  where the $a_j$s are defined in \eqref{adef}, then 
\begin{itemize}
\item $V(\bm v,d) \cap H_0$ consists of three distinct points
$$ [0,1,0], \  [0,r_1,R_1], \  [0,r_2,R_2],$$ where $r_j \neq 0, R_j \neq 0$ and
$$a_4 R_j^2 -2v_{12} r_jR_j - a_2 r_j^2 =0,$$ for $j=1,2.$
\item $V \cap H_1$ consists of three distinct points
$$[1,0,0], \  [s_1,0,S_1], \  [s_2,0,S_2],$$ where $s_j \neq 0, S_j \neq 0$ and 
$$a_4S_j^2 +2v_{22}s_jS_j - a_3s_j^2 = 0,$$ for $j=1,2$.
\item Since $a_1 \neq 0$,   $V \cap H_2$ contains only the two points $[1,0,0],[0,1,0]$.
\item The four points $[0,r_j,R_j], [s_j,0,S_j]$ are smooth points of $V$.  

\end{itemize}
\end{lemma} 

\begin{proof}  The only part of this yet to prove is the last.

The four points $[0,r_j,R_j],[s_j,0,S_j]$ lie in the affine part of $V$ corresponding to $u_2 \neq 0$.  Choosing affine coordinates $(X,Y)$ in this affine piece, $V$ is defined by
$$f(X,Y) = (a_1 Y^2 - 2v_{12} Y - a_3)X^2 + 2v_{22}X(1 + Y^2) + (-a_2 Y^2 -2v_{12} Y + a_4)=  0.$$  The points $[0,r_j,R_j]$ correspond to $(0,r_j/R_j) \doteq (0,\alpha_j)$ and the points $[s_j,0,S_j]$ correspond to $(s_j/S_j, 0) = (\beta_j, 0)$; we know that, given the hypotheses,  $\alpha_1 \neq \alpha_2$, $\beta_1 \neq \beta_2$.

Let's compute the normal vector to the curve at the points $(0,\alpha_j), (\beta_j,0)$; for this we need to compute the gradient vector:
$$\nabla f = \langle 2X(a_1Y^2 -2v_{12}Y-a_3)+ 2v_{22}(1+Y^2), 2X^2(a_1Y-v_{12}) + 4v_{22}XY -2(a_2Y+v_{12}) \rangle.$$  So,
\begin{itemize}
\item $\nabla f \mid_{(0, \alpha_j)} = \langle 2v_{22}(1 + \alpha_j^2), -2(a_2 \alpha_j + v_{12})\rangle$,
\item $\nabla f \mid_{(\beta_j,0)} = \langle -2 (a_3 \beta_j -  v_{22}) , -2v_{12}(1 + \beta_j^2)  \rangle.$
\end{itemize}

Now, if $a_2 \alpha_j + v_{12} = 0, $ then since $0 = a_4-2v_{12}\alpha_j -a_2 \alpha_j^2,$ we must have $0 = a_4-2v_{12} \alpha_j + v_{12} \alpha_j$, or
$a_4 -v_{12} \alpha_j = 0$.  In this case, $a_2a_4 + v_{12}^2 = a_2(v_{12} \alpha_j) + v_{12}(-a_2 \alpha_j) = 0,$ which we are assuming is not true.  Similarly, if $a_3 \beta_j - v_{22} = 0$, we also contradict a hypothesis.    Therefore, the four points $(0, \alpha_j),(\beta_j,0)$ are smooth points of $V$.

\end{proof}

\begin{corollary} \label{genericdomainbeta}  For generic choices of $[\bm v,d]$, the open subset
 $$O_2(\bm v,d) \doteq V(\bm v,d) - (H_0 \cup H_1 \cup H_2)$$  of $V(\bm v,d)$ is obtained by removing a finite number of points in $V(\bm v,d)$, $[0,0,1]$ does not lie on $V(\bm v,d)$, every point of $O_2(\bm v,d)$ is a smooth point of $V(\bm v,d)$ and $\beta$, when restricted to $O_2$, is injective.
\end{corollary}
\begin{proof}  This follows directly from Lemma \ref{lemmaVintersectslines} and the discussion in Section 4.3.  Moreover,  conditions under which ``genericity" are achieved have been specified in the hypotheses of this Lemma.
\end{proof}
\subsubsection{The basepoints of the linear system $\mathcal{V}$}

Determining the base points of the linear system $ \mathcal{V} = \{V(\bm v,d) \mid [\bm v,d] \in \mathbb{C}P^4\}$ is not too difficult; for $[u_0,u_1,u_2]$ to be a basepoint, it must lie on $V(\bm 0,\bm 0,d), d \neq 0$.
Therefore, $-u_1^2 = u_2^2 $ or $-u_0^2 = u_2^2$.

If $-u_1^2= u_2^2$, setting $d=0, v_{12} = 0, v_{11} \neq 0$, we must have $2v_{11}u_1^2(u_0^2+u_2^2) = 0$, so $u_2 = u_1 =0$ or $u_2^2=-u_0^2$.   So, our basepoint is either $[1,0,0]$ or of the form
$$[1, \pm 1, \pm \sqrt{-1}].$$

If $-u_0^2=u_2^2$, setting $d=0, v_{21} \neq 0, v_{22} = 0$, we must have the point $[0,1,0]$ or the four points already listed.  Therefore, we have 6 basepoints for the linear system:
$$[1,0,0],\quad [0,1,0], \quad [1,1, \sqrt{-1}], \quad [1,1,-\sqrt{-1}], \quad [1,-1,\sqrt{-1}],  \quad [1,-1,-\sqrt{-1}].$$

\subsubsection{The singularities of $V(\bm v,d)$}

Recalling that the defining polynomial $P$ for $V(\bm v, d)$, written as \eqref{polyP}, is
$$ a_4u_2^4 + 2(v_{22}u_0-v_{12}u_1)u_2^3 -(a_3u_0^2 +a_2u_1^2)u_2^2 + 2u_0u_1(v_{22}u_1 - v_{12}u_0)u_2 + a_1u_0^2u_1^2,$$ we can compute that
\begin{align*}
\frac{\partial P}{\partial u_0} &= 2v_{22}u_2^3-2a_3u_0u_2^2 - 2v_{12}u_0u_1u_2 +  2u_1u_2(v_{22}u_1 - v_{12}u_0) + 2a_1u_0u_1^2, \cr
\frac{\partial P}{\partial u_1} &= -2v_{12}u_2^3 - 2a_2u_1u_2^2 + 2v_{22}u_0u_1u_2 + 2u_0u_2(v_{22}u_1-v_{12}u_0) + 2a_1u_0^2u_1. \cr
\frac{\partial P}{\partial u_2} &= 4a_4u_2^3 + 6(v_{22}u_0-v_{12}u_1)u_2^2 -2(a_3u_0^2+a_2u_1^2)u_2 +2u_0u_1(v_{22}u_1-v_{12}u_0).
\end{align*}

One can check that the basepoints $[1,0,0],[0,1,0]$ are always singular points of the variety $V(\bm v,d)$.
For the points $[1, \pm 1, \pm \sqrt{-1}]$, on the other hand, we have  
\begin{lemma} \label{bpspV} If $v_{j1}^2 + v_{j2}^2 \neq 0$, for $j = 1$ or $j=2$, then none of the basepoints 
$$[1,1, \sqrt{-1}], \quad [1,1,-\sqrt{-1}], \quad [1,-1,\sqrt{-1}], \quad [1,-1,-\sqrt{-1}]$$ are singular points for $V(\bm v,d)$.  
\end{lemma}

\begin{proof}
For these points $u_0^2+u_2^2 = 0 = u_1^2 + u_2^2$.  Therefore, evaluating the partial derivatives, we get
$$\frac{\partial P}{\partial u_0} =  -4v_{12}(\pm \sqrt{-1})  +4 v_{11}, \qquad 
	\frac{\partial P}{\partial u_1} = 4v_{22}(\pm \sqrt{-1})  - 4v_{21}(\pm 1)$$
which are nonzero.   
\end{proof}

\subsubsection{What type of singularities are the points $[1,0,0],[0,1,0]$?} \label{singtypeV}

The analysis in this subsection will be used in Corollary \ref{Hirred}.  

Let's consider  first the singular point $[1,0,0]$.  In the affine open set defined by $u_0 \neq 0$, we must consider the affine plane curve, with variables $Y,Z$, defined by setting the quartic
$$a_4Z^4 + 2(v_{22}-v_{12}Y)Z^3 - (a_3 + a_2Y^2)Z^2 + 2Y(v_{22}Y-v_{12})Z + a_1 Y^2 $$ equal to zero.  Of course, this polynomial is obtained from $P$ by dividing by $u_0^4$ and setting $Y = u_1/u_0, Z = u_2/u_0.$

We immediately assume that $a_1 \neq 0$, since we do not wish to consider a case where we are sure that $V$ is reducible; then, 
$$P = (-(a_3 + a_2Y^2) Z^2 + a_1 Y^2 - 2v_{12}YZ) + Z(Z^2(2v_{22} - 2v_{12} Y + a_4Z) + 2 v_{22}Y^2).$$  Define 
$$U_1 \doteq  -a_3 + a_2Y^2 + (2v_{22} - 2 v_{12} Y + a_4Z)Z,  $$
$$U_2 \doteq  a_1 + 2v_{22} Z. $$  

In order to continue the analysis, we need to consider these polynomials in the formal power series ring $\mathbb{C}[[Y,Z]]$. Assume that $a_3 \neq 0$.   Recall that a formal power series is invertible in the ring $\mathbb{C}[[Y,Z]]$  if and only if its constant term is nonzero. Therefore, $U_1, U_2$ are both units in $\mathbb{C}[[Y,Z]]$ and
 $$P = U_1Z^2 + U_2 Y^2 - 2 v_{12}YZ$$ in $\mathbb{C}[[Y,Z]]$.  The discriminant of this quadratic is $4(v_{12}^2 - U_1 U_2)$ and, given the definitions of $U_1,U_2$, if $v_{12}^2 + a_1 a_3 \neq 0$, the discriminant is a unit in $\mathbb{C}[[Y,Z]]$.  Therefore, $P$ factors into two distinct factors, or is of the form $L^2$ ($L \neq 0$)   if $v_{12}^2 + a_1a_3 = 0$. We have proven

\begin{lemma}  \label{lemmasingtypeV1}  If $a_1 \neq 0, a_3 \neq 0, v_{12}^2 +a_1a_3 \neq 0$, then $[1,0,0]$ is an ordinary double point (i.e., a nodal singularity) of $V(\bm v,d)$.  If $v_{12}^2 + a_1a_3 = 0$, then $[1,0,0]$ is an ordinary cusp.
\end{lemma}

Similarly for the singular point $[0,1,0]$ we have
\begin{lemma} \label{lemmasingtypeV2}  If $a_1 \neq 0, a_2 \neq 0, v_{22}^2 + a_1a_2 \neq 0$, then $[0,1,0]$ is an ordinary double point (i.e., a nodal singularity) of $V(\bm v,d)$.  If $v_{22}^2 + a_1a_2 = 0$, then $[0,1,0]$ is an ordinary cusp.
\end{lemma}

\subsubsection{Summary of properties of $\mathcal V$}	\label{summaryV}

\begin{remark} \label{V2nodes} Using Bertini's theorem-BTGS for the linear system $\mathcal{V}$ of plane quartics, for generic choices of $\bm v_1,\bm v_2,d$,  $V(\bm v,d)$ is a plane quartic with exactly two  singularities $[1,0,0],[0,1,0]$, each of which are nodes (Lemmas \ref{bpspV}, \ref{lemmasingtypeV1}, \ref{lemmasingtypeV2}). 

Moreover, the hypotheses of Lemmas \ref{bpspV}, \ref{lemmasingtypeV1}, \ref{lemmasingtypeV2} give explicit conditions under which we are guaranteed that $[1,0,0], [0,1,0]$ are nodal singularities and the only basepoints of the linear system $\mathcal{V}$ which can be singularities.  \end{remark}

Now, a plane quartic with exactly two nodal singularities, and no other singularities, must be an irreducible plane curve.  
\begin{itemize}
\item[a)] The curve must be reduced, since otherwise the singularity calculation using the Jacobian would yield an infinite set of singularities. 
\item[b)] The curve can't be a union of a line and an irreducible cubic curve, for then it would have either three distinct nodal singularities, or at least one non-nodal singularity. 
\item[c)] The curve can't be a union of two distinct irreducible conics, because in that case, the curve would either have four distinct nodal singularities, or at least one non-nodal singularity.  
 \item[d)] The curve can't be a union of an irreducible conic and two lines, since then it would have five nodes or at least one non-nodal singularity. 
\item[e)]  Finally, the curve can't be a union of four distinct lines, because in this case there would be six distinct nodal singularities, or a triple point.

\end{itemize}

Thus we see
\begin{lemma}\label{Velliptic} 
For generic choices of $\bm v_1,\bm v_2,d$ the plane quartic curve $V(\bm v,d)$ is an irreducible curve with exactly two nodal singularities $[1,0,0],[0,1,0]$ and the point $[0,0,1]$ does not lie on $V(\bm v,d)$.  The genus-degree formula for singular  plane curves \cite{GH} tells us that a  desingularization  of $V(\bm v,d)$ has genus equal to $\frac{(4-1)(4-2)}{2} - 2= 1$; so that  a desingularization of $V(\bm v,d)$ is a smooth irreducible  curve of genus 1.
\end{lemma}

Smooth irreducible curves of genus one, in some projective space, are also called {\bf elliptic} curves; sometimes the adjective ``elliptic" means that one has also specified a point on the curve, which we have not done here.

\subsection{Conclusion:   generic irreducibility of  $HC_F(\bm v,d)$} 

\begin{corollary}  \label{Hirred} 
For generic choices of $\bm v_1,\bm v_2,d$, 
  $HC_F(\bm v,d)$ is  an irreducible curve, birational to the irreducible curve $V(\bm v,d)$ (which has exactly two nodes), with exactly four singularities $\bm p_1,\bm p_2,\bm p_3,\bm p_4.$ 

\end{corollary}

\begin{proof} For generic choices of $\bm v_1,\bm v_2,d$,   Lemma \ref{Velliptic} says that $V(\bm v,d)$ is an irreducible curve with exactly two nodal singularities $[1,0,0],[0,1,0]$.  Lemma  \ref{foursingsH} says that  $HC_F(\bm v,d)$ has exactly four singularities $\bm p_1,\bm p_2,\bm p_3,\bm p_4$. 
Referring back to results in Section 6,  using Lemma \ref{genericdomainbeta}, the holomorphic map $\beta:O_2(\bm v,d) \rightarrow HC_F(\bm v,d)$ is injective and smooth, also $O_2(\bm v,d)$ is obtained from $V$ by deleting a finite number of points; using \ref{genericdomainalpha}, the map $\alpha:O_1(\bm v,d) = HC_F(\bm v,d) - \{\bm p_1,\bm p_2,\bm p_3,\bm p_4\} \rightarrow V(\bm v,d)$ is injective and every point of $HC_F(\bm v,d)$, except for the four points $\bm p_j$, is a smooth point of $HC_F$.   Moreover, $\alpha$  is an inverse to $\beta$ on $O_1, O_2$.  Therefore, since $V(\bm v,d)$ is is irreducible, so is $HC_F(\bm v,d)$ and $V(\bm v,d)$ is a curve birational to $HC_F(\bm v,d)$.  
\end{proof}

Generically, as noted previously, $HC_F(\bm v,d)$ should also have only nodal singularities, so that this curve has a desingularization that is a smooth irreducible curve of genus 1, but we do not present the proof of this here.

\section{Example:  An equal-velocity case}  
\label{specialXvariety}

This section is an extended discussion, completely avoiding the use of Bertini's Theorems,  of a special example; but similar analyses could be carried out in many other cases.  
This equal-velocity example is an illustrative example that is of practical importance. 

In particular, the equal-velocity case $\bm v_1 = \bm v_2 = (0,v)$ corresponds to two important physical scenarios:   
1) the sensors are stationary, and the source is moving in the direction perpendicular  to the line joining the sensors; or 
2) the source is stationary and the sensors are moving in tandem, perhaps both sensors located symmetrically on the same vehicle.

 The equal-velocity case is not a generic case in the entire linear system of curves $HC_F$; however, within the sub-linear system we consider here, we may also speak of ``genericity".

In subsections \ref{d=0,vnonzero} and \ref{v=0} we consider the non-generic cases when $d=0$ or $v=0$, respectively.   
In subsection \ref{genericEqualV} we determine the singularities of the ``generic" equal-velocity case.
In subsection \ref{desing}, we construct a desingularization of a generic equal-velocity curve $HC_F(v,d)$.

We consider here  the  linear system of curves on $Y(Q,Q_1)$
$$\mathcal{P}H = \{ H((0,v),(0,v),d) \mid [v,d] \in \mathbb{C}P^1\}$$ and the related linear system of plane quartics 
$$\mathcal{P}V = \{ V((0,v),(0,v),d) \mid [v,d] \in  \mathbb{C}P^1\}.$$ 
These linear subsystems of $\mathcal{H}, \mathcal{V}$ are called ``{\bf pencils}" because  the parameter set $[v,d]$ is the projective line.

One can check that the basepoints of these linear subsystems of $\mathcal{H},\mathcal{V}$, respectively, have exactly the same set of basepoints as $\mathcal{H}, \mathcal{V}$.

Since we will be determining the singular points of $HC_F$ without using Bertini's Theorem, we will initially use the $[\bm z, x]$-coordinates on $Y(Q,Q_1)$.  Recall that the equations defining $HC_F(v,d) \doteq H((0,v),(0,v),d)$, for $[v,d] \in \mathbb{C}P^1$, are, in the $[\bm z, x]$-coordinates 
$$z_2^2 -z_0^2 = z_3^2-z_1^2,  \qquad  x^2 + z_2^2-z_0^2 = 0,  \qquad  v(z_0-z_1)x - dz_0z_1=0$$ 
and the equation defining $V(v,d) \doteq V((0,v),(0,v),d)$ is
$$P =  X u_0^2 -2v(u_2^2+ u_1^2)u_2 u_0 + X u_2^2= 0,$$ where
$$X= du_2^2+2vu_1u_2 + du_1^2.$$ 

In this case, \eqref{adef} says that
$$a_1 = -d, \quad  a_2 = d,  \quad  a_3 = d, \quad  a_4 = -d, \quad  v_{12} = v_{22} = v.$$

We note that, following Remark 5.27, we are not considering the case $d=v=0$, in other words, $[d,v]$ is in a projective parameter space.
\subsection{The case $\bm{d = 0, v \neq 0}$}	\label{d=0,vnonzero}
In the notation of \eqref{FDOAdef}, we see that the equal-velocity case $d=0, v \neq 0$ corresponds to 
$$0 = \bm e_1 \cdot \bm v_1 = \bm e_2 \cdot \bm v_2  = v (e_{1,2} - e_{2,2})$$
where $e_{j,2}$ denotes the second coordinate of the unit vector from the source  to  sensor $j$.   
This occurs for the stationary-sensor case when the source is moving along the perpendicular bisector of the line between the sensors; in the stationary-source case, it occurs when both sensors are flying directly towards the source or directly away from it.   In these cases,  there is extra symmetry in the  physical problem. 

When $d = 0$, 
 $HC_F(v,0)$ 
is defined on $Y(Q,Q_1)$ by
$$(z_0-z_1)x = 0$$ and the plane curve
$V$ is defined by the equation
$$(u_1-u_0)(u_2^2 - u_0u_1)u_2 = 0.$$ 

Both $HC_F(v, 0)$ and $V$ are reducible curves.  

For $HC_F(v,0)$, we have six irreducible curves on $Y(Q,Q_1)$:
$$C_1^+ = \{[z_0,z_0, z_2,z_2,x] \in \mathbb{C}P^4 \mid [z_0,z_2,x] \in \mathbb{C}P^2, x^2 + z_2^2-z_0^2 = 0\},$$
$$C_1^- = \{[z_0,z_0, z_2,-z_2,x] \in \mathbb{C}P^4 \mid [z_0,z_2,x] \in \mathbb{C}P^2, x^2 + z_2^2-z_0^2 = 0\},$$
$$\hat{\ell}_1, \hat{\ell}_2, \hat{\ell}_3, \hat{\ell}_4$$ and
$$HC_F(v,0) = C_1^+ \cup C_2^- \cup (\cup_{j = 1}^4 \hat{\ell}_j).$$ 
 Each of $C_1^+,C_2^+$ is isomorphic to an irreducible plane conic, and of course the curves $\hat{\ell}_j$ are all lines.

The curve $V$ in this case is defined by
$$(u_1-u_0)u_2(u_2^2-u_0u_1)=0,$$  so that  $V$ is reducible and the union of three curves--two (different) lines and an irreducible (smooth) conic.  The lines are 
$$H = \{ [u_0,u_0,u_2] \mid [u_0,u_2] \in \mathbb{C}P^1\}, \qquad 
\tilde{H} = \{ [u_0,u_1,0] \mid [u_0,u_1] \in \mathbb{C}P^2 \} ;$$ the conic is
$$C = \{[u_0,u_1,u_2] \mid u_2^2 -u_0u_1 = 0\}.$$

\subsection{The stationary case $\bm{v =0}$}	\label{v=0}
 Referring to \eqref{FDOAdef}, we see that the case $\bm v_1 = \bm v_2 = \bm 0$ is the case when there is no relative motion between the sensors and source.    In this case the Doppler shifts, which are proportional to $\bm e_j \cdot \bm v $, both vanish, and  consequently their difference  is  $d=0$.    Consequently the case $\bm v = 0, d \neq 0$ does not correspond to a physical problem.   
 
 However, we still get some algebraic curves.  Let's see how this plays out.

If $v=0$, then $HC_F(0,d) $ is defined on $Y(Q,Q_1)$ by $$z_0z_1=0$$ and $HC_F = E_1^+ \cup E_1^- \cup E_2^+ \cup E_2^-$ is the union of four irreducible curves, each isomorphic to a smooth, irreducible plane conic:
$$E_1^{\pm}  = \{[0, z_1,z_2,z_3, \pm \sqrt{-1}z_2]  \in \mathbb{C}P^4 \mid [z_1,z_2,z_3] \in \mathbb{C}P^2, z_1^2 + z_2^2 = z_3^2 \},$$
$$E_2^{\pm} = \{[z_0,0,z_2,z_3, \pm \sqrt{-1}z_3] \in \mathbb{C}P^4 \mid [z_0,z_2,z_3] \in \mathbb{C}P^2, z_0^2 + z_3^3 = z_2^2 \}.$$
We note that the only real points on each of these curves lie amongst the points $\bm p_j$, $1 \leq j \leq 4$; this can be seen as telling us that for this case of $v = 0$, in ``reality" we only have the points corresponding to sensor locations.

On the other hand,  $V$ is defined by the equation
$$d(u_2^4 + (u_0^2 + u_1^2)u_2^2 + u_0^2 u_1^2) = 0$$ or
$$(u_2^2 + u_1^2)(u_2^2 + u_0^2) = 0.$$ Again, $V$ is reducible and is the union of four lines
\begin{align*} 
K_1^+ &= \{ [u_0,u_1,\sqrt{-1}u_1] \mid [u_0,u_1] \in \mathbb{C}P^1\}, \cr
K_1^- &= \{ [u_0,u_1,-\sqrt{-1}u_1] \mid [u_0,u_1] \in \mathbb{C}P^1\}, \cr
K_2^+ &= \{ [u_0,u_1,\sqrt{-1}u_0] \mid [u_0,u_1] \in \mathbb{C}P^1\}, \cr
K_2^- &= \{ [u_0,u_1,-\sqrt{-1}u_0] \mid [u_0,u_1] \in \mathbb{C}P^1\}.
\end{align*}
 Each of these lines contain exactly one real point, which is either $[1,0,0]$ or $[0,1,0]$.  Again, we see that in ``reality" we only have the points corresponding to sensor locations (in some sense).

\subsection{The equal-velocity case $\bm{v \neq 0, d \neq 0}$}	\label{genericEqualV}


\subsubsection{Singularities of $HC_F(v,d)$}

\begin{lemma} \label{Hedsing} Suppose that  $d,v \neq 0$.  Then, 
\begin{itemize}
\item if $d^2 \neq 4v^2$, then the only singular points of $HC_F(v,d)$ are $\bm p_1,\bm p_2, \bm p_3, \bm p_4$.  
\item If $d = 2v$, then $HC_F(v,d) $ has exactly five singular points $\bm p_1,\bm p_2,\bm p_3,\bm p_4, \bm q$.
\item If $d = -2v$, then $HC_F(v,d)$ has exactly five singular points $\bm p_1,\bm p_2,\bm p_3,\bm p_4,\tilde{\bm q}$.
\end{itemize}
Here $\bm q \doteq [1,-1,0,0,1]$  and $\tilde{\bm q} \doteq [1,-1,0,0,-1]$ in the $[\bm z, x]$-coordinates. 
\end{lemma}

\begin{proof}
We examine the Jacobian for the defining equations in the $[\bm z, x]$-coordinates, which has the same rank as
$$\left[ \begin{array}{ccccc}
-z_0 & z_1 & z_2 & -z_3 & 0 \\
-z_0 & 0 & z_2 & 0 & x \\
vx-dz_1 & -vx-dz_0 & 0 & 0 & e(z_0-z_1)\end{array} \right].$$  Without presenting full details, if $x = 0$, then we obtain only the four already known singular points $p_1,p_2,p_3,p_4$.  So, assume that $x \neq 0$.  Multiplying the last column by $x$ and using the equations for $HC_F$, the matrix becomes
$$\left[ \begin{array}{ccccc}
-z_0 & z_1 & z_2 & -z_3 & 0 \\
-z_0 & 0 & z_2 & 0 & z_0^2-z_2^2 \\
vx-dz_1 & -vx-dz_0 & 0 & 0 & dz_0z_1\end{array} \right].$$

Looking at columns 3,4,5, for a singular point $[\bm z, x]$ on $HC_F$, we must have
$$dz_0z_1z_2z_3 = 0,$$ so that at least one of the $z_j$s must be zero.

If $z_0  = 0$,  the the equation for $HC_F$ on $Y(Q,Q_1)$ says that  $z_1 x=0$ so that $z_1=0$.  Similarly, if $z_1 = 0$, then $z_0 = 0.$   In either case, looking at columns 1,3,4, for a singular point, we must have $z_2z_3=0$.  However, now, $x^2 + z_2^2 = x^2 + z_3^2 = 0$, so that either $z_2=0$ or $z_3 = 0$ forces $x=0$, which we are assuming is not true, and no singular points are found here.

If $z_2 = 0$, then $x^2 = z_0^2, z_1^2-z_0^2 = z_3^2$; since $x \neq 0, z_0 \neq 0$ and the Jacobian has the same rank as
$$\left[ \begin{array}{cccc}
-z_0 & z_1 &  -z_3 & 0 \\
-z_0 & 0 &  0 & z_0 \\
vx-dz_1 & -vx-dz_0 & 0 & dz_1\end{array} \right],$$ which in turn has the same rank as
$$\left[ \begin{array}{cccc}
-z_0 & z_1 &  -z_3 & 0 \\
0 & 0 &  0 & z_0 \\
vx & -vx-dz_0 & 0 & dz_1\end{array} \right].$$  
Looking at columns 1,3,4, for a singular point, we must have
$$z_3z_0 vx = 0,$$ 
which means that $z_3 = 0$.  Looking at columns 1,2, 4, we get
$$z_0(vx + dz_0)-vxz_1=0$$ or
$$vx(z_0-z_1) + dz_0^2 = dz_0z_1 + dz_0^2 = dz_0(z_0+z_1) = 0.$$  This means that $z_1=-z_0$, and our possible singular point looks like
$$[z_0,-z_0, 0,0,x],$$ and $2vz_0x - dz_0^2 = 0$, or $2vx = dz_0$.  Since $x^2 = z_0^2$, this forces $4v^2 = d^2.$  Thus, if $d \neq \pm 2v$, we get no singular points here.  

If $d =2v,$ then $\bm q \doteq [1,-1,0,0,1]$  (in the $[\bm z, x]$-coordinates) is a singular point, and if $d = -2v$, then $\tilde{\bm q} \doteq [1,-1,0,0,-1]$ is a singular point.  

\end{proof}
We have seen that the points $\bm p_1,\bm p_2, \bm p_3, \bm p_4$ correspond to the sensor positions; the points $\bm q$ and $\tilde{\bm q}$ correspond, via \eqref{z2y}, to $[u, \bm y, \bm r] = [0, (0, \mp 1), (1, -1)]$.  
These points are on the line that bisects the sensor axis, but on the ``nonphysical'' part of the curve with $r_1$ or $r_2$ negative.

\subsubsection{Singularities of $V(v,d)$.}

\begin{lemma} Suppose that  $ d \neq 0, v \neq 0.$

If $ d \neq \pm 2v$, the only singular points of $V(v,d)$ are $[1,0,0], [0,1,0]$ and $V(v,d)$ is an irreducible quartic curve.  Moreover, if $d^2 \neq v^2$ then both singularities are nodes, and if $d^2 =v^2$, then both singularities are ordinary cusps.   

If $d=2v$, then $V(v,2v)$ is irreducible,  there is exactly one more singular point $[1,-1,1]$, in addition to the points $[1,0,0],[0,1,0]$, and this singular point is also a node.

If $d = -2v$, then $V(v,-2v)$ is irreducible,  there is exactly one more singular point $[-1,1,1]$, in addition to the points $[1,0,0],[0,1,0]$, and this singular point is also a node.

\end{lemma}

\begin{proof}

These results could be checked by a software package.

Recall that $a_1 = -d, a_2 = d, a_3 = d, a_4 = -d, v_{12} = v_{22} = v$, so that 
$$v_{12}^2 + a_1a_3 = v^2-d^2, v_{22}^2 + a_1a_2 = v^2-d^2.$$  Therefore, Lemmas \ref{lemmasingtypeV1}, \ref{lemmasingtypeV2} say that
\begin{itemize}
\item if $v^2 \neq d^2$, then $[1,0,0],[0,1,0]$ are both nodes on $V(e,d)$ and
\item if $v^2 = d^2$, then $[1,0,0],[0,1,0]$ are both ordinary cusps on $V(e,d)$.
\end{itemize}

To look for other singularities, recall that  the equation for $V$ is $$P =  X u_0^2 -2v(u_2^3 + u_1^2u_2) u_0 + X u_2^2= 0,$$ where
$$X= du_2^2+2vu_1u_2 + du_1^2.$$  Now, the first partials of $P$ are
\begin{align*}
P_{u_0} &= 2u_0 X - 2v(u_2^3 + u_1^2u_2),\cr
P_{u_1} &= X_{u_1} u_0^2 - 4vu_0u_1u_2 + X_{u_1}u_2^2, \cr
P_{u_2} &= X_{u_2}u_0^2 - 2v(3u_2^2 + u_1^2) u_0 + u_2^2 X_{u_2} + 2 u_2 X.
\end{align*}

Let's assume that $u_2 \neq 0$, for if $u_2 = 0$ we get only the already known singularities $[1,0,0],[0,1,0]$.

If $P_{u_0} = 0$, then 
\begin{equation} \label{eqsingV1}
u_0X = vu_2(u_1^2 + u_2^2).
\end{equation}

Given that we have a point on $P = 0$, we see that
$$vu_0(u_2^3 + u_1^2 u_2) - 2v(u_2^3 + u_1^2u_2)u_0 + u_2^2X = 0, $$ or
\begin{equation} \label{eqsingV2} u_2X = vu_0(u_1^2 + u_2^2). \end{equation}

Setting $P_{u_2} = 0$ as well,  rearranging, and using \eqref{eqsingV2} in the last term, 
\begin{align*}
0 
&= X_{u_2} ( u_0^2 + u_2^2) -  2v(3u_2^2 + u_1^2) u_0  + 2 u_2 X  \cr
&= X_{u_2}(u_0^2 + u_2^2) - 4vu_0u_2^2.
\end{align*}
Since $X_{u_2} = 2du_2+2vu_1$, we see that
\begin{equation} \label{eqsingV3} (du_2 +v u_1)(u_0^2 + u_2^2) = 2v u_0u_2^2.
\end{equation}

Finally, setting $P_{u_1} = 0$, we have
$$X_{u_1}(u_0^2 + u_2^2) = 4 vu_0u_1u_2.$$  Since $X_{u_1} = 2du_1+2vu_2,$
\begin{equation} \label{eqsingV4} (du_1 +v u_2)(u_0^2 + u_2^2) = 2vu_0u_1u_2. \end{equation}

Taken together, \eqref{eqsingV3} and \eqref{eqsingV4} imply that
$$(u_2^2-u_0^2)(u_1^2 + u_2^2) = 0,$$ so that
$$u_2^2 = u_0^2$$ or
$$u_2^2 = -u_1^2.$$

If $u_2^2 = u_0^2$, then we must have $u_0 \neq 0$.  In this case, \eqref{eqsingV3} implies that
\begin{equation} \label{eqsingV5} du_2 + vu_1 =v u_0, \end{equation} so that, using \eqref{eqsingV4},
$$(du_1 +vu_2)u_2 = (du_2 + vu_1)u_1$$ or
$$u_2^2 = u_1^2.$$  Since $u_0^2 = u_1^2 = u_2^2$, \eqref{eqsingV5} says that
$$(du_2 + vu_1)^2  = v^2u_1^2$$ and
$$d^2 u_2^2 + 2dvu_1u_2 = 0$$ yielding $du_2 +2vu_1=0$, or $u_2 = -(2v/d)u_1$; finally, this means that $d^2 = 4v^2.$  So, if $d^2 \neq 4v^2$, we get no singular points here.

On the other hand, if $d = 2v$, then $u_2 = -u_1$, and, since $u_0^2 = u_1^2 = u_2^2$ and $u_0 = 2u_2+u_1 = -u_1$ and we have only the point $[1,-1,1]$.  If $d = -2v$, then $u_2 = u_1$ and $u_0 = -2u_2 +u_1 = -u_1$ again, so we have only the point $[1,-1,-1] = [-1,1,1]$.

Suppose now that $u_2^2 = -u_1^2$.  Then, \eqref{eqsingV2} says that, since $u_2 \neq 0$, $X = 0$.  Since $X = du_2^2 + 2vu_1u_2 -du_2^2$ now, this forces $u_1u_2=0$.  However, since $u_2 \neq 0$ and $u_2^2 = -u_1^2$, $u_1\neq 0$ as well.  Thus we get a contradiction and no singular points in this case.

  If $d \neq 0, v \neq 0, d \neq \pm 2v$, we've shown in the calculations above that the only singular points of $V(v,d)$ are $[1,0,0],[0,1,0]$. As pointed out above, Lemmas \ref{lemmasingtypeV1} and \ref{lemmasingtypeV2} show that the two singular points are both nodes if $d^2 \neq v^2$,  and are both ordinary cusps if $d^2 = v^2$.  The argument after Remark \ref{V2nodes} shows that $V(v,d)$ is an irreducible quartic curve.

If $d = 2v$, we've already seen that a singular point $[u_0,u_1,u_2]$, other than the two points $[1,0,0],[0,1,0]$ must be equal to $[1,-1,1]$; while if $d = -2v$, the ``new" singular point must be $[-1,1,1]$.   Let's show that, in either case, these new singular points are nodes.  This we can determine by looking at the curve in the affine piece where $u_2 \neq 0$, since the new singular point lies in this affine piece.  We will give details only in the case $d = 2v$.  In this case, $V$  is defined by
$$ X u_0^2 -2v(u_2^3 + u_1^2u_2) u_0 + X u_2^2= 0,$$ where
$$X= 2vu_2^2+2vu_1u_2 + 2eu_1^2,$$   so that the equation for $V(2v,v)$ is 
$$ (u_2^2 + u_1u_2 + u_1^2)u_0^2 -(u_2^3 +u_1^2u_2)u_0 + (u_2^2 + u_1u_2 + u_1^2)u_2^2 = 0.$$   In the affine plane with coordinates $(U_0,U_1)$, where $u_2 \neq 0$, $V$ is defined by
$$g = (U_1^2 + U_1 + 1)U_0^2 -(1 + U_1^2)U_0 + (U_1^2 +U_1 + 1) = 0;$$ expanding the polynomial $g$ on the left about $(1,-1)$, we have
$$g = (U_0-1)^2 + (U_1+1)^2 + (U_0-1)(U_1+1)^2 - (U_0-1)^2(U_1+1) + (U_0-1)^2(U_1+1)^2.$$  Therefore, $[1,-1,1]$ is a node. 

Revisiting the discussion immediately after Remark \ref{V2nodes}, a plane quartic, with exactly three nodes and no other singularities, is reducible if and only if it is a union of a line and an irreducible cubic curve.  Moreover, the line must pass through all three nodes.  Let's say the equation of the line is $au_0 + bu_1 + cu_2 = 0$.  Since the line must pass through $[1,0,0]$, $a = 0$, since it passes through $[0,1,0]$, $b = 0$ and since it passes through $[1,-1,1]$, we must have $c = 0$, a contradiction.  Therefore, $V(2v,v)$ is irreducible.

A similar argument shows the claim in the case $d = -2v$.

 \end{proof}

The analogues of Lemma \ref{Velliptic} and Corollary \ref{Hirred} for these pencils are the following; the proofs of these facts do not invoke Bertini's theorem.

\begin{lemma} \label{Vedelliptic} Suppose that $d \neq 0, v \neq 0$.  Then,  the plane quartic curve $V(v,d)$ is an irreducible  curve.

If $d^2 \neq 4v^2$, $V(v,d)$ has exactly two singularities $[1,0,0],[0,1,0]$ which are either both nodes, or both cusps.  The genus-degree formula for singular  plane curves tells us that, in either case,  the desingularization $\tilde{V}(v,d)$ of $V(v,d)$ has genus equal to $\frac{(4-1)(4-2)}{2} - 2= 1$; so that  $\tilde{V}(v,d)$ is a smooth (irreducible)  curve of genus 1.

If $d^2 = 4v^2$, $V(v,d)$ has exactly three singularities.  Either all three singularities are nodes or two are cusps and one is a node.  In either case, the genus-degree formula tells us that the desingularization $\tilde{V}(v,d)$ in this case has genus equal to zero, and is thus a smooth, rational curve.
\end{lemma}

Using Lemma \ref{Vedelliptic} we have
\begin{corollary}   Suppose that $d \neq 0, v \neq 0$.  Then,
  $HC_F(v,d)$ is  an irreducible curve, birational to the irreducible curve $V(v,d)$.
  
  If $d^2 \neq 4v^2$, $HC_F(v,d)$ has exactly four singularities, $\bm p_1,\bm p_2, \bm p_3,\bm p_4$.
  
  If $d^2 = 4v^2$, $HC_F(v,d)$ has exactly five singularities. If $d = 2v$, the singularities are $\bm p_1,\bm p_2,\bm p_3,\bm p_4, \bm q$; while if $d = -2v$, the singularities are $\bm p_1,\bm p_2,\bm p_3,\bm p_4, \tilde{\bm q}$; the definitions of $\bm q,\tilde{\bm q}$ are given in  Lemma \ref{Hedsing}.
\end{corollary}

\begin{proof}  Lemma \ref{Vedelliptic} says that $V(v,d)$ is an irreducible curve. 

Referring back to Section 6.1,  the holomorphic map $\beta:V(v,d) \dashrightarrow HC_F(v,d)$ defines a rational map  between the two curves $V(v,d)$ and $HC_F(v,d)$; the holomorphic map $\alpha:HC_F(v,d) \dashrightarrow V(v,d)$ defines a rational map between the two curves $HC_F(v,d )$ and $V(v,d)$.  Moreover, $\alpha$  is an inverse to $\beta$ after removing a finite set of points from both domains, using Lemmas \ref{HintersectsZ} and \ref{lemmaVintersectslines}.  Therefore, since $V(v,d)$ is irreducible, so is $HC_F(v,d)$ and $V(v,d)$ is a curve birational to $HC_F(v,d)$.  

Lemma \ref{Hedsing} proves the statements about the singularities of $HC_F(v,d)$.  

\end{proof}

\subsection{Desingularization of $HC_F(v,d)$} 
\label{desing}

This section is included for those interested, but it's not really necessary to read for subsequent arguments.

 In this subsection, we consider only the case $d \neq 0, v \neq 0, d^2 \neq 4v^2, d^2 \neq v^2$. 
In this case, we've seen that the only singularities $[1,0,0],[0,1,0]$ of $V(e,d)$ are both nodes and $V(v,d)$ is an irreducible quartic plane curve defined by 
$$P =  X u_0^2 -2v(u_2^2+ u_1^2)u_2 u_0 + X u_2^2= 0,$$ where
$$X(u_1,u_2)= du_2^2+2vu_1u_2 + du_1^2.$$ 

The four points $[0,r_j,R_j], [s_j,0,S_j]$ of Lemma \ref{lemmaVintersectslines}, which are the intersections of $V$ with the lines $H_0, H_1$, respectively,  that are not the two nodes, are of the form
\begin{equation} \label{foursmoothpoints} 
[0,1,t],[0,t,1], [1,0,-t],[-t,0,1] 
\end{equation} 
where $dt^2 + 2v t + d = 0$. 
 (In this discussion, $t$ is a constant, and is not one of the variables in the ``original coordinates".  The original coordinates are not used  anywhere in this Section.) 
 Note that $t \neq \pm 1$ since $d^2 \neq v^2$, and 
 consequently the points \eqref{foursmoothpoints}   
 are really four distinct points. Moreover, Lemma \ref{lemmaVintersectslines} says that all four of these points are smooth points of $V$ (which we already know,  since $d \neq 0$ implies $t \neq 0$, independently of that Lemma).

Note also that $t$ is real if and only if $d^2 \leq v^2$ and $v,d$ are real.  Here, we see the ``Cauchy-Schwarz bound" \eqref{CSbound} popping up again to govern the real points on the curve.

 Choose one of these four points \eqref{foursmoothpoints}, say the point $[0,1,t]$.

Now, curve theory tells us, that since $V$ is defined by an irreducible plane quartic curve, with exactly two nodal singularities, we expect that, after doing a {\bf Cremona transformation} on the plane, chosen with two coordinates corresponding to the nodes, and one coordinate chosen with respect to a smooth point on the curve, there is a smooth, irreducible plane cubic curve $\tilde{V}$, birational to the  curve $V$; moreover,  we may take $\tilde{V}$ to be a  desingularization  of  $HC_F$.  (The curve $\tilde{V}$ is also a desingularization of $V$, but, ultimately, our interest is in $HC_F$.)

Let's find such a plane cubic $\tilde{V}$ explicitly, and see how this all works.  

\begin{remark} In this section, we change coordinates a couple of times in the plane:  note that the new coordinate names here have nothing to do with the coordinate names used elsewhere in this paper, in other contexts. \end{remark}

\subsubsection{The Cremona transformation and the plane cubic}

First, we have three  points $A = [1,0,0], B = [0,1,0], C = C(v,d) = [0,1,t]$ on $V$; we have seen that $A,B$ are the two nodes, and $C$ is a smooth point on $V$.  Let us change coordinates using the  linear coordinate change
$$u_0 = v_0, \ \  u_1 = v_1 + v_2,  \  \ u_2 = tv_2,$$ 
so that the polynomial defining $V$ becomes
\begin{equation} X(v_1,v_2)v_0^2 - 2e((v_1+v_2)^2 + t^2 v_2^2) tv_0v_2 + X(v_1,v_2)t^2v_2^2 ,
\end{equation}  where
\begin{align*}
X(v_1,v_2)   &  = dt^2 v_2^2 + 2vt(v_1v_2 + v_2^2) + d(v_1 + v_2)^2  \cr
& = (dt^2 + 2vt + d) v_2^2 + 2(vt +d) v_1v_2 +  dv_1^2  \cr
& = 2(vt+d) v_1v_2 + dv_1^2,
\end{align*} 
since $dt^2 + 2vt + d = 0.$   Recall that $v_0, v_1, v_2$ are  coordinate names and do not refer to components of a velocity vector.   

 The three  points $A,B,C$, in the new coordinates, are
$$A = [1,0,0], \qquad B = [0,1,0],  \qquad C = [0,0,1].$$ 
Applying the rational map (the Cremona transformation) of the plane with coordinates $[v_0,v_1,v_2]$  to the plane with coordinates $[q_0,q_1,q_2]$
$$[v_0,v_1,v_2] \mapsto [1/v_0,1/v_1,1/v_2],$$ the equation defining $\tilde{V}$ is obtained from that defining $V$ by computing the image of the polynomial defining $V$  and ``clearing denominators".  

Starting this process, we begin with 
$$X\left (\frac{1}{q_1}, \frac{1}{q_2} \right) \frac{1}{q_0^2} - 2vt\frac{1}{q_0}\frac{1}{q_2}\left( \left (\frac{1}{q_1} 
	+ \frac{1}{q_2}\right)^2 +  \frac{t^2}{q_2^2} \right)  + X\left(\frac{1}{q_1} , \frac{1}{q_2} \right) \frac{t^2}{q_2^2}.$$

Clearing denominators and collecting terms,  we arrive at the cubic polynomial
$$(2(v t+d)q_1 +  dq_2)q_2^2 - 2vtq_0((q_1+q_2)^2 + t^2 q_1^2) + (2(vt+d)q_1 + dq_2)t^2 q_0^2.$$
In other words, 
$$ \tilde{V}(v,d) = $$
$$\{[q_0,q_1,q_2]  \in \mathbb{C}P^2 \mid (2(vt+d)q_1 +  dq_2)q_2^2 - 2vtq_0((q_1+q_2)^2 + t^2 q_1^2) + (2(vt+d)q_1 + dq_2)t^2 q_0^2=0\}.$$

The curve $\tilde{V}(v,d)$ also depends on the choice of $t$, but we do not denote this dependence.

Note that we have defined changes of coordinates, where the coordinates are  $[u_0, u_1, u_2], [v_0, v_1, v_2]$, and $[q_0, q_1, q_2]$
in the first, second and third planes, respectively:  
$$\mathbb{C}P^2 \stackrel{A}{\longrightarrow} \mathbb{C}P^2 \stackrel{\gamma}{\dashrightarrow} \mathbb{C}P^2$$ 
given by
\begin{align*}
A([u_0,u_1,u_2]) = \left[u_0, u_1-\frac{u_2}{t}, \frac{u_2}{t} \right], \qquad 
\gamma([v_0,v_1,v_2]) = \left[\frac{1}{v_0}, \frac{1}{v_1}, \frac{1}{v_2} \right].
\end{align*}
 The map $A$ is an invertible, projective linear transformation defined everywhere.  
 The map $\gamma$  is defined on the complement of the (standard) coordinate lines in the 
 $[v_0, v_1, v_2]$-plane;  
 moreover, $\gamma$ is invertible where defined and is its own  inverse:
$$[q_0,q_1,q_2] \mapsto \left[\frac{1}{q_0}, \frac{1}{q_1}, \frac{1}{q_2} \right],$$ 
which is also defined on the complement of the coordinate lines in the $[q]$-plane.

\subsubsection{$\tilde{V}(v,d)$ is a desingularization of $HC_F(v,d)$}

Our aim is now to define a map
$$\rho:\tilde{V}(v,d) \rightarrow HC_F(v,d),$$ defined everywhere, that is bijective outside of eight distinct points on $\tilde{V}(v,d)$, and maps pairs of these eight points onto the singular points $\bm p_1,\bm p_2,\bm p_3,\bm p_4$ of $HC_F(e,d)$.  

We already have most of the pieces necessary to define $\rho$.  Give $HC_F(v,d)$ the $[\bm w, x_1]$-coordinates.

Then define $\rho$ on the points $[q]$ of $\tilde{V}(v,d)$ with $q_0q_1q_2 \neq 0$ by
\begin{equation} \label{rhodef} \rho([q_0,q_1,q_2]) 
= [-q_0q_2(q_1+q_2)^2, -q_1q_2^2(q_1 + q_2),\  t^2 q_0^2 q_1(q_1 + q_2),\  t^2 q_0q_1^2 q_2, \  t q_0q_1q_2(q_1 + q_2)].
\end{equation}

Note that even if $v,d$ are real numbers, $\rho$ is not a real map if $t$ is not real; i.e. if $v^2<d^2$.

Of course, the map $\rho$ on $\tilde{V}$ is the restriction of a map defined on $\mathbb{C}P^2-\{[q_0,q_1,q_2] \mid q_0q_1q_2=0\}$; we'll call this map $\rho$ as well.

The reader can check that the image of $\rho$, where defined, is a subset of $HC_F$.  Moreover, there is a map $\hat{\rho}$,
$$\hat{\rho}:HC_F(v,d) \dashrightarrow \tilde{V}(v,d),$$ defined everywhere on $HC_F$ except for $\bm p_1,\bm p_2,\bm p_3,\bm p_4$, given by 
$$\hat{\rho}([w_0,w_1,w_2,w_3,x_1]) = [w_2(tx_1-w_3),\  tw_3x_1, tx_1(tx_1-w_3)].$$

Lemma \ref{HintersectsZ} proves that $\hat{\rho}$ is defined everywhere except at $\bm p_1,\bm p_2,\bm p_3,\bm p_4$.  The reader can check that the image of $\hat{\rho}$ really is a subset of $\tilde{V}(e,d)$.  
The maps $\rho$ and $\hat{\rho}$ 
are constructed using the maps $\alpha, \beta, A$ and $\gamma$; for example, $\rho$  is a rational map corresponding to $\beta \circ A^{-1} \circ \gamma^{-1}$,  where that composite makes sense.  One can check that $\rho$ and $\hat{\rho}$ are inverse to each other, where the composite map makes sense.

It remains to check that, in fact, we can extend the definition of $\rho$ to ALL points of $\tilde{V}(v,d)$, and still get a holomorphic map.  What this means is that we have to examine all the points on $\tilde{V}(v,d)$ where one of the coordinates is equal to zero.  For each of these points for which formula \eqref{rhodef} does not apply,  we have to look at the tangent directions to $\tilde{V}(v,d)$ at each of those points in order to figure out where to send those points.  It's worth presenting the details of that, so we do that in the following.

First, note that the domain of $\rho$ naturally includes some points other than those where $q_0q_1q_2\neq 0$.  For,  if $q_0=0$, then $\rho([0,q_1,q_2]) = [0,-q_1q_2^2(q_1+q_2), 0,0,0], $ so that if $q_1q_2(q_1+q_2) \neq 0$, then
$$\rho([0,q_1,q_2]) = [0,1,0,0,0] = \bm p_3.$$ Similarly, if $q_1=0$ and $q_0q_2 \neq 0$, then 
$$\rho([q_0,0,q_2]) = [1,0,0,0,0] = \bm p_1;$$ if $q_2=0$ and $q_0q_1 \neq 0$, 
$$\rho([q_0,q_1,0]) = [0,0,1,0,0]= \bm p_4.$$

Finally, if $q_1+q_2=0$, but $q_0q_1q_2 \neq 0,$ then
$$\rho([q_0,q_1,-q_1]) = [0,0,0,1,0] = \bm p_2.$$ These calculations help tell us where the points on $\tilde{V}(v,d)$ where $q_0,q_1$ or $q_2 $ equal zero should go.

If $[0,q_1,q_2] \in \tilde{V}(e,d),$ then $(2(vt+d)q_1 + dq_2) q_2^2 = 0$, so that $q_2=0$ or $q_2 = -\frac{2(vt+d)q_1}{d}$ and we get exactly two points on $\tilde{V}$ where $q_0=0$:
$$[0,1,0], \quad \left[0,1, -\frac{2(vt+d)}{d} \right].$$   
Using the equation $dt^2 + 2vt + d = 0$, we may rewrite
$$\frac{2vt + 2d}{d} = \frac{-d-dt^2 + 2d}{d} = 1-t^2 \neq 0,$$ so that the second point above is 
$$[0,1,t^2 -1].$$

If $[q_0,0,q_2] \in \tilde{V}$, then $q_2 = 0$ or $dq_2^2 -2vtq_0q_2 + dt^2 q_0^2 = 0$. In this case, we get exactly three points on $\tilde{V}$ with $q_1=0$:
$$[1,0,0], \quad [r,0,R], \quad [s,0,S],$$ 
where $dR^2 - 2vtrR + dt^2 r^2 = 0, dS^2 - 2vtsS + dt^2 s^2 = 0$. Note that none of $r,R,s,S$ is equal to zero, and $[r,0,R] \neq [s,0,S]$ since $e^2 \neq d^2$, so that we really have three distinct points here.

If $[q_0,q_1,0] \in \tilde{V}$, then $2tq_0q_1(-v(1+t^2) q_1 + (vt+d)tq_0) = 0$ so that we get exactly three points on $\tilde{V}$ with $q_2=0$:
$$[1,0,0],\quad [0,1,0], \quad \left[1,\frac{t(vt+d)}{v(1+t^2)},0 \right].$$
 Since $d^2 \neq v^2$, $t(vt+d) \neq 0$, and there are really three distinct points here.  

In conclusion, we have exactly six points on $\tilde{V}$ where $q_0,q_1$ or $q_2$ is zero:
$$[0,1,0], \quad [0,1, t^2-1], \quad [1,0,0], \quad [r,0,R],\quad [s,0,S], \quad \left[ 1,\frac{t(et+d)}{e(1+t^2)},0 \right],$$

Now, we also need to compute the points on $\tilde{V}$ of the form $[q_0,q_1,-q_1]$.  For such a point, we need
$q_1((2vt+d)q_1^2 - 2vt^3q_0q_1 + (2vt+d)t^2q_0^2) = 0$. Therefore, we get three distinct such points:
$$[1,0,0], \quad [u,U,-U],\quad [z,Z,-Z],$$ 
where $(2vt+d)U^2 - 2vt^3 uU + (2vt+d)u^2 = 0, (2vt+v)Z^2 - 2vt^3zZ + (2vt+d) z^2 = 0$. Again, our assumptions on $v,d,t$ say that the eight points
$$[0,1,0],  \ \  [0,1, t^2-1], \ \  [1,0,0], \ \  [r,0,R],\ \  [s,0,S], \ \  \left[1,\frac{t(vt+d)}{v(1+t^2)},0 \right],  \ \  [u,U,-U], \ \  [z,Z,-Z]$$ are all distinct.

We've seen that
$$[0,1, t^2-1] \stackrel{\rho}{\longmapsto} \bm p_3,  \quad [r,0,R] \stackrel{\rho}{\longmapsto} \bm p_1, \quad [s,0,S] \stackrel{\rho}{\longmapsto} \bm p_1,$$
$$\left[1,\frac{t(et+d)}{e(1+t^2)},0\right] \stackrel{\rho}{\longmapsto} \bm p_4,  \quad [u,U,-U] \stackrel{\rho}{\longmapsto} \bm p_2, 
	\quad [z,Z,-Z] \stackrel{\rho}{\longmapsto} \bm p_2.$$

Thus, the only two points on $\tilde{V}$ on which $\rho$ isn't  defined  by  formula (5.12) are $[0,1,0],[1,0,0]$.  We claim that if we define
\begin{equation}\label{rhoextended} \rho([0,1,0]) = \bm p_3,  \quad \rho([1,0,0]) =\bm  p_4,
\end{equation}  then $\rho$ is holomorphic everywhere on $\tilde{V}$.

Once this is verified, we have
\begin{lemma} The  map $\rho:\tilde{V}(v,d) \rightarrow HC_F(v,d)$, defined by \eqref{rhodef} and \eqref{rhoextended}, is holomorphic and restricts to a bijective map
$$\tilde{V}(v,d) - F \rightarrow HC_F(v,d) - \{\bm p_1,\bm p_2,\bm p_3,\bm p_4\},$$ where
$F = \{[0,1,0], [0,1, t^2-1], [1,0,0], [r,0,R],[s,0,S], [1,\frac{t(vt+d)}{v(1+t^2)},0], [u,U,-U],[z,Z,-Z]\}$ is the set of eight points defined in the above discussion.  

Moreover, for each $j$, the inverse image of $\bm p_j$ with respect to $\rho$ consists of exactly two points and $\tilde{V}(v,d)$ is the  desingularization of $HC_F(v,d)$.

The curve $\tilde{V}(v,d)$ is an irreducible smooth curve of genus 1.
\end{lemma}
\begin{proof} We need only verify that, using definition \eqref{rhoextended},  $\rho$ is holomorphic at $[0,1,0]$ and $[1,0,0]$.

Examining the Taylor expansion about $(0,0)$  for the polynomial defining $\tilde{V}$  in the affine open set $q_1 \neq 0$, we have the polynomial
$$-2vt(t^2 + 1)q_0 + 2(vt+d)q_2^2 - 4vtq_0q_2 + 2(vt+d)q_0^2 + dt^2 q_0^2 q_2 - 2v q_0q_2^2 + dq_2^3.$$ Thus, the tangent line to the curve at $(0,0)$ is the $q_2$-axis (in the affine $(q_0,q_2)$ plane).  
However, for every point $[0,1,q_2]$  with  $q_2 \neq 0$,  we have  $\rho([0,1,q_2]) = \bm p_3$, so defining $\rho([0,1,0]) =\bm p_3$ gives a map which is holomorphic at $[0,1,0]$.

Examining the Taylor expansion about $(0,0)$ for the polynomial defining $\tilde{V}$ in the affine open set $q_0 \neq 0$, we have the polynomial
$$2(vt+d)t^2 q_1 + dt^2 q_2 - 2vt(t^2+1) q_1^2 - 4vt q_1q_2 - 2v t q_2^2 + 2(vt + d)q_1q_2^2 + d q_2^3.$$ Therefore, the tangent line to the curve at $(0,0)$ is the line
$$q_2 = -\frac{2(vt+d)}{d} q_1 = (t^2-1)q_1$$ in the affine $(q_1,q_2)$-plane.  For points on this tangent line, not equal to $(0,0)$, we see that $\rho$ of such a point is equal to
$$[-a(1+a^2)q_1, \ -a^2(1+a) q_1^2, \ t^2(1+a),\  t^2aq_1,\  ta(1+a)q_1],$$ where $a = t^2-1.$ Letting $q_1$ go to zero, these points must approach
$$[0,0,1,0,0] =\bm  p_4.$$

The last statement follows from applying the genus-degree formula to the smooth plane cubic curve $\tilde{V}(v,d).$
\end{proof}


\section{The FDOA problem and real varieties}	\label{realX}

We first review the notation established in Section \ref{varietydef}.  If $F_1, \ldots, F_k$ are homogeneous polynomials in variables $x_0, \ldots, x_n$,  we have the  projective variety
$V(F_1, \ldots, F_k) \subseteq \mathbb{C}P^n$ which is defined as  the set of solutions to the system of polynomial equations $F_1 =  \ldots = F_k = 0$.  If the coefficients of all the polynomials $F_j$ are real numbers, we say that $V$ is a {\bf real} variety.  In this case,  $V(\mathbb{R})$ is the set of points $[\bxi] \in V$ such that there are real numbers $x_j$ such that
$[\bxi] = [x_0, \ldots, x_n]$.  We then say ``{\bf $V(\mathbb{R})$ is the set of real points of the real variety $V$}".

If $(\bm v,d)$ is a choice of parameters, where the coordinates of $$\bm v = (\bm v_1, \bm v_2) = ((v_{11},v_{12}), (v_{21}, v_{22}))$$ are all real numbers and $d$ is also real, then the varieties 
$$HC_F(\bm v,d), \quad V(\bm v,d)$$ 
from Sections 5,6 and 7 are all real varieties, and it makes sense to talk about their real points
$$HC_F(\bm v,d)(\mathbb{R}),  \quad V(\bm v,d)(\mathbb{R}).$$  

The varieties $\tilde{V}(v,d)$, of Section 7.4.1, however, are only real varieties in the case $0<d^2 < v^2$, for this is the condition that ensures that the scalar $t$ defined therein is a real number.  

In this section, we assume that $\bm v_1, \bm v_2 \in \mathbb{R}^2$, $d \in \mathbb{R}$.  

\subsection{The starting problem}

Let's first go back to the beginning and consider the starting problem in Section 2, from which we have seemingly strayed so far in Section 4, 5,6 and 7.   This is the set of points in $\mathbb{R}^2$ that is the solution set to \eqref{FDOAdef}; giving this set a name, define 
the FDOA isocurve $A_0$ as
$$A_0(\bm v_1, \bm v_2, d) = \left\{\bm y = (y_1,y_2) \in \mathbb{R}^2 \mid  d = \frac{(\bm s_2- \bm y)\cdot \bm v_2}{| \bm s_2 - \bm y |} - \frac{(\bm s_1- \bm y)\cdot \bm v_1}{ |\bm s_1- \bm y |}\right\}.$$  Note that, strictly speaking, the points corresponding to the sensor location $\bm s_1,\bm s_2$ do not lie on $A_0$, since the denominators are zero at those points. 
Using the homeomorphism $\bm y \mapsto [1, \bm y]$ from $\mathbb{R}^2$ to the affine open subset $U_0 \subseteq \mathbb{R}P^2$, we may consider $A_0(\bm v_1,\bm v_2,d)$ as a non-closed subset of $\mathbb{R}P^2$.   Now, the set $A_0 \subseteq \mathbb{R}^2$ is not defined by polynomial equations, so its solution set is not an algebraic variety.  Another problem is that there are variables in the denominators, certainly not a polynomial-like situation.  However, we'll show that $A_0$ is a real piece of a complex algebraic curve $Z$ in $\mathbb{C}P^2$ in this section.

An analogous situation happens in the much simpler case of plane hyperbolas:  if we fix two points $\bm s_1, \bm s_2$ in the plane $\mathbb{R}^2$ and a positive constant $c$, then the set of points $\bm y$ such that
$$c = | \bm s_2-\bm y| - | \bm s_1 - \bm y|$$ is not an algebraic variety, even over the real numbers.  The smallest algebraic variety in $\mathbb{R}^2$ containing this curve is the algebraic variety, a conic section, that we usually 
call a ``hyperbola".  We obtain the polynomial equation defining the conic section in $\mathbb{R}^2$  from the equation above by squaring twice, and then simplifying.  The resulting curve, defined by a polynomial of degree two, has more points that the original curve defined by the equation above.  Fortunately, there are no denominators to worry about in this case of hyperbolas.

So, let's follow a similar process with our set $A_0$, which is a much more complicated curve than a piece of a hyperbola, moreover there are denominators!

Adding the variable $u$, taking us into projective space, 
 we change coordinates to the $[u,\bm y]$ coordinates described in Section 3.1.  Since this is a real, linear coordinate change, it ``keeps real points real".
So we want to take the equation  
\begin{equation} \label{22projectivized} 
d = \frac{(-u-y_1,-y_2)\cdot \bm v_2}{| (-u-y_1,-y_2) |} - \frac{ (u-y_1,-y_2)\cdot \bm v_1}{ | (u-y_1,-y_2) |}, 
\end{equation} 
and first, get rid of the denominators:
\begin{align*}
d | (-u-y_1,-y_2) | &  | (u-y_1,-y_2) | =  \cr
 & ((-u-y_1,-y_2) \cdot \bm v_2)| (u-y_1,-y_2) |- ((u-y_1,-y_2) \cdot \bm v_1)  | (-u-y_1,-y_2) |.
 \end{align*}
    Set 
\begin{equation}	\label{Ldef}
L_1 = (u-y_1,-y_2) \cdot \bm v_1,  \qquad L_2 = (-u-y_1,-y_2) \cdot \bm v_2.
\end{equation}

Next, square both sides, arriving at
$$L_2^2 f_1 -2L_1L_2 | (u-y_1,-y_2) || (-u-y_1,-y_2) | + L_1^2 f_2 = d^2 f_1 f_2,$$ where

\begin{equation}	\label{fdef}
 f_1 = (u-y_1)^2 + y_2^2, \quad  f_2 = (u+y_1)^2 + y_2^2.
 \end{equation}
   Solving for the part that still involves square roots,
we get
$$2L_1L_2| (u-y_1,-y_2) || (-u-y_1,-y_2) |  = L_2^2 f_1 + L_1^2 f_2 - d^2f_1 f_2,$$ squaring again, we get
$$4L_1^2 L_2^2 f_1 f_2 = (L_2^2 f_1 + L_1^2 f_2-d^2f_1f_2)^2$$ 
 so that we obtain the ``FDOA polynomial" 
\begin{equation} \label{polyZ} 
h(u,y_1,y_2) \doteq (L_2^2 f_1 + L_1^2 f_2-d^2f_1f_2)^2 - 4L_1^2 L_2^2 f_1 f_2 .
\end{equation} 
This FDOA polynomial 
appears to be a homogeneous polynomial of degree 8 in the variables $u,y_1,y_2,$ with real coefficients.   
We have shown that if  $(y_1,y_2) $  is on the FDOA isocurve $A_0$,  or, equivalently, if 
$[1,y_1,y_2] \in \mathbb{R}P^2$  is a solution to \eqref{22projectivized} (thus not equal to $[1,1,0],[1,-1,0]$), 
 then $[1,y_1,y_2]$ is also a solution to 
 $h(1,y_1,y_2) = 0$.   Of course, $h(1,1,0) = h(1,-1,0) = 0$ so we see that $Z$ contains two points that $A_0$ does not:  points corresponding to the sensor points.  As we shall see, there are many points $[1,y_1,y_2] $, not equal to $[1, \pm 1,0]$,  such that $h(1,y_1,y_2) = 0$, but $(y_1,y_2) \notin A_0$.

 We cannot 
 simplify the polynomial $h$ unless we know more about its coefficients, which are nonlinear polynomial expressions in the parameters $v_{11},v_{12},v_{21},v_{22},d$.  This nonlinearity means that the family of curves $\mathcal{Z}$, defined by the curves $Z(\bm v,d)$ described below, is not a linear family in $\mathbb{C}P^2$.

Define a   variety in $\mathbb{C}P^2$, whether $\bm v,d$ are real parameters or not,  by
$$Z(\bm v,d)  = \{ [u,y_1,y_2] \in \mathbb{C}P^2 \mid h(u,y_1,y_2) = 0\};$$ 
 if $\bm v_1,\bm v_2$ are real vectors and $d \in \mathbb{R}$, this is a real variety and the set of real points in this variety is
 the real FDOA variety 
$$Z(\bm v,d) (\mathbb{R})= \{[u,y_1,y_2] \in \mathbb{R}P^2 \mid h(u,y_1,y_2) = 0\}.$$   In this case,  the calculations above show that, after the change of variables to $[u,y_1,y_2]$,  $A_0(\bm v,d)$  may be regarded as a subset of  an open subset of $Z(\bm v,d)(\mathbb{R})$ using the homeomorphism $(y_1,y_2) \mapsto [1,y_1,y_2]$.

\subsection{The projection $\rho:HC_F(\bm v,d) \rightarrow Z(\bm v,d)$}

In this section, we discuss in detail the correspondence between the solutions of the FDOA equation \eqref{FDOAdef}, 
which correspond as above to points of the real FDOA variety $Z(\bm v,d)(\mathbb{R})$,  and the points of the 
 Ho-Chen FDOA variety $HC_F(\bm v,d)$ studied in earlier sections. 
 
Consider the projection $\rho : HC_F(\bm v,d) \rightarrow \mathbb{C}P^2$, given in the original coordinates $[u,\bm y, r_1, r_2]$ by
$$\rho([u,y_1,y_2,r_1,r_2]) = [u, y_1,y_2].$$ This projection is defined everywhere on $HC_F$.   The image of $\rho$ is contained in $Z(\bm v,d)$, since if $[u,y_1,y_2,r_1,r_2] \in HC_F$, then $f_1(u,y_1,y_2) = r_1^2, \ f_2(u,y_1,y_2) = r_2^2$ and $L_2(u,y_1,y_2)r_1-L_1(u,y_1,y_2) r_2 - dr_1r_2 = 0$.  Therefore, following exactly the same ``squaring" process on the equation $L_2(u,y_1,y_2)r_1-L_1(u,y_1,y_2) r_2 - dr_1r_2 = 0$ as in Section 8.1, we must have
$$h(u,y_1,y_2) = 0.$$  

If $\bm v,d$ are real, then when restricted to the real points of $HC_F$,
$$\rho : HC_F(\bm v,d)(\mathbb{R}) \rightarrow Z(\bm v,d)(\mathbb{R}).$$

We show below that $\rho$ is a bijection between a certain subset of $HC_F(\bm v, d)$ and  a certain subset of $Z(\bm v, d)$. 

\subsubsection{The projection $\rho:HC_F(\bm v,d) \rightarrow Z(\bm v,d)$ is surjective}	\label{sec:rhosurjective}

\begin{lemma} \label{rhosurjective}  The map $\rho:HC_F(\bm v,d)\rightarrow Z(\bm v, d)$ is surjective.  Moreover, if $\bm v_1,\bm v_2$ are real vectors, and $d \in \mathbb{R}$, then 
$$\rho:HC_F(\bm v,d)(\mathbb{R}) \rightarrow Z(\bm v,d)(\mathbb{R})$$ is surjective.
\end{lemma}
\begin{proof}
Start with a point $[u,y_1,y_2] \in Z(\bm v, d)$.  Then, we may choose numbers $R_1,R_2$ such that 
$R_1^2 = f_1(u,y_1,y_2), R_2^2 = f_2(u,y_1,y_2)$; if $u,y_1,y_2 \in \mathbb{R}$, since $f_j(u,y_1,y_2)$ are always nonnegative real numbers, for $j=1,2$,  then $R_1,R_2$ are real numbers.

Substituting these values into \eqref{polyZ}, we get 
\begin{align}	
0 &= (L_2^2 R_1^2 + L_1^2 R_2^2-d^2R_1^2R_2^2)^2 - 4L_1^2 L_2^2 R_1^2 R_2^2  \cr
& = (L_2R_1-L_1R_2-dR_1R_2)(L_2(- R_1)-L_1(-R_2) - d(-R_1)(- R_2))   \cr
&\qquad  \cdot (L_2 (-R_1) - L_1 R_2 - d(-R_1)R_2)(L_2R_1 - L_1(- R_2) - dR_1(- R_2))  \label{factor} \\
& = g_1(u,y_1,y_2,R_1,R_2)g_1(u,y_1,y_2,-R_1,-R_2) g_1(u,y_1,y_2,-R_1,R_2)g_1(u,y_1,y_2,R_1,-R_2) \cr
& \doteq g_1 g_2 g_3 g_4    \nonumber 
\end{align}
where  $L_j = L_j(u,y_1,y_2)$ are given by \eqref{Ldef},
$$g_1(u,y_1,y_2,R_1,R_2) \doteq L_2 R_1 - L_1 R_2 - dR_1R_2$$  and  the $g_j(u,y_1, y_2,R_1, R_2), j >1$ are defined by the last equation. 

Therefore, given our choice of $R_1,R_2$, one of the four  quantities $g_1,g_2,g_3,g_4$ is zero when evaluated at $(u,y_1,y_2,R_1,R_2)$.

If $g_1(u,y_1,y_2,R_1,R_2) = 0$, we are done:  $[u,y_1,y_2,R_1,R_2] \in HC_F(\bm v,d)$ by definition.

If $g_1(u,y_1,y_2,-R_1,-R_2) = 0$, then one can compute that $[u,y_1,y_2,-R_1,-R_2] \in HC_F(\bm v,d)$; if $g_1(u,y_1,y_2,-R_1,R_2) = 0,$ then $[u,y_1,y_2,-R_1,R_2] \in HC_F(\bm v,d)$ and finally if $g_1(u,y_1,y_2,R_1,-R_2) = 0$, then $[u,y_1,y_2,R_1,-R_2] \in HC_F(\bm v,d)$.  

Since any of these points $[u,y_1,y_2, \pm R_1,\pm R_2]$ maps to $[u,y_1,y_2]$ via $\rho$, $\rho$ must be surjective. 
\end{proof}

Now, using Lemma \ref{rhosurjective}, and Corollary \ref{Hirred}, since $HC_F(\bm v,d)$ is irreducible for generic choices of $\bm v,d$, and $Z(\bm v,d)$ is the image of $HC_F(\bm v, d)$ with respect to the projection $\rho$, we must have that $Z(\bm v,d)$ is also irreducible for generic choices of $\bm v,d$.  Thus, generically, $Z$ cannot have the lines $L_1$ or $L_2$ as a component.  This tells us that, generically, neither $L_1$ nor $L_2$ is a polynomial factor of the FDOA polynomial $h$.  In the next section, we will refine this and give precise constraints on the parameters $\bm v,d$ that guarantee that neither $L_1$ nor $L_2$ is a factor of $h$.

 We can also consider  the map (with the same name and formula) $\rho:Y(Q_1,Q_2) \rightarrow \mathbb{C}P^2$, defined everywhere on $Y$.  The proof above of Lemma \ref{rhosurjective} shows us that the set-theoretic inverse image of $Z(\bm v,d)$ with respect to $\rho:Y(Q_1,Q_2) \rightarrow \mathbb{C}P^2$ is a union of four curves, one of which is $HC_F(\bm v,d)$, on $Y(Q_1,Q_2)$.  The other three curves have third defining equations (other than $Q_2 = 0,Q_1=0$)
$$L_2(-r_1) - L_1(-r_2) - d(-r_1)(-r_2) = 0,$$
$$L_2(-r_1)-L_1r_2-d(-r_1)r_2 = 0,$$ or
$$ L_2r_1-L_1(-r_2)-dr_1(-r_2)=0,$$ respectively, in $\mathbb{C}P^4$ (coordinates here are $[u,y_1,y_2,r_1,r_2]$).

Let's call these four curves $HC_F(\bm v,d) \doteq HC_F^{++}, HC_F^{--}, HC_F^{-+}, HC_F^{+-},$ respectively.  There's a projective  automorphism of $\mathbb{C}P^4$ of the form $[u,y_1,y_2,r_1,r_2] \mapsto [u, y_1,y_2,\pm r_2, \pm r_2]$ that maps any one of these four curves to any other of the four, so the curves are all isomorphic. They are not mutually disjoint however.

\subsubsection{The fibres of $\rho$} \label{sec:Gchar}

Here, we consider the map $\rho:HC_F(\bm v,d) \rightarrow Z(\bm v,d)$ and, for a fixed point $[u,y_1,y_2] \in Z$, we set ourselves the problem of determining how many points are in the fibre $\rho^{-1}(\{[u,y_1,y_2]\})$.  We don't want to consider every single choice of parameters $\bm v,d$ for brevity's sake.  So:  we assume that $\bm v_1, \bm v_2$ are both nonzero vectors, and $d \neq 0$ as well.  Then, Lemma \ref{rhosurjective} tells us that there is at least one point in the fibre.

Let $G(\bm v,d) = \{[u,y_1,y_2] \in \mathbb{C}P^2 \mid L_1(u,y_1,y_2)L_2(u,y_1,y_2) = 0\}.$  Since  $\bm v_1,\bm v_2$ are nonzero vectors,  $L_1$ and $L_2$ are nonzero  homogeneous polynomials of degree 1, and  $L_1L_2$ has degree 2.    

\begin{definition} \label{noLfactors} Given parameters $[\bm v, d] \in \mathbb{C}P^4$ with $\bm v_1 \neq \bm 0, \bm v_2 \neq \bm 0, d \neq 0$, consider the following additional requirements:
\begin{itemize}
\item $v_{11}^2 + v_{12}^2 \neq 0,$ and $v_{21}^2 + v_{22}^2 \neq 0$ 
\item $v_{21} \neq 0$ or $ v_{11}^2 -d^2 \neq 0$ 
\item $v_{11} \neq 0$ or $v_{21}^2 -d^2 \neq 0$.
\end{itemize}
\end{definition}

\begin{lemma} \label{linearfactorsh} If the parameters $\bm v,d$ satisfy all conditions in \ref{noLfactors},  then neither $L_1$ nor $L_2$ is a polynomial factor of the FDOA polynomial $h$.  Moreover, if any one of the three conditions in Definition \ref{noLfactors} is not true, then one of $L_1$ or $L_2$ is a polynomial factor of $h$.
\end{lemma}
\begin{proof}  First, note that the conditions $v_{11}^2 + v_{12}^2 \neq 0, v_{21}^2 + v_{22}^2 \neq 0$ are equivalent to $\bm v_1 \neq \bm 0, \bm v_2 \neq \bm 0$ for real vectors $\bm v_j$.

Recall that \eqref{Ldef} defines $L_1,L_2$ as
$$L_1 = (u-y_1,-y_2) \cdot \bm v_1 = v_{11}(u-y_2)-v_{12}y_2,   L_2 = (-u-y_1,-y_2) \cdot \bm v_2 = -v_{21}(u+y_1) -v_{22}y_2,$$
while
\ref{fdef} defines $f_1,f_2$ as
$$ f_1 = (u-y_1)^2 + y_2^2,   f_2 = (u+y_1)^2 + y_2^2.$$

Then, the ``FDOA polynomial"  \eqref{polyZ}
is
$$h(u,y_1,y_2) \doteq (L_2^2 f_1 + L_1^2 f_2-d^2f_1f_2)^2 - 4L_1^2 L_2^2 f_1 f_2 .$$  Since $L_1,L_2$ are nonzero linear polynomials, they are irreducible as polynomials. 

Examining $h$, we see then that $L_1$ is a factor of $h$ if and only if $L_1$ divides $(L_2^2-d^2f_2)f_1$, in other words, if and only if $L_1$ divides $f_1$
or $L_1$ divides $L_2^2 -d^2f_2$.  Similarly, $L_2$ is a factor of $h$ if and only if $L_2$ divides $f_2$ or $L_2$ divides $L_1^2-d^2f_1$.

Changing variables to $T = u-y_1, S = u+y_1, y_2 = y_2$, $f_1 = T^2 + y_2^2, f_2 = S^2 + y_2^2, L_1 = v_{11}T -v_{12}y_2, L_2 = -v_{21}S -v_{22}y_2.$ Since the only factors of $f_1$ are $T \pm \sqrt{-1}y_2$,  we must have $v_{12} = \pm \sqrt{-1} v_{11},$ or $v_{11}^2 + v_{12}^2 = 0$ in order for $L_1$ to divide $f_1$.  Similarly, in order for $L_2$ to divide $f_2$, we must have $v_{22} = \pm \sqrt{-1} v_{21},$ or $v_{21}^2 + v_{22}^2 = 0$.

Consider $L_2^2 -d^2f_2 = (v_{21}^2-d^2)S^2 + 2v_{21}v_{22}Sy_2 + (v_{22}^2-d^2)y_2^2$.  If $v_{11} \neq 0$, then clearly $v_{11}T-v_{12}y_2$ can't divide $L_2^2-d^2f_2$, since $T$ does not appear in this polynomial.  If $v_{11} = 0$, then $v_{12} \neq 0$ and  the only way $v_{12}y_2$ can divide $L_2^2-d^2f_2$ is if $v_{21}^2-d^2 = 0$.  Similarly, the only way that $L_2$ can divide $L_1^2 -d^2 f_1$ is if $v_{21} = 0$ and $v_{11}^2-d^2 =0$.

\end{proof}
\begin{corollary}\label{bezout} With the hypotheses of Lemma \ref{linearfactorsh}, there are at most 16 distinct points in $Z(\bm v,d) \cap G(\bm v,d)$. 
\end{corollary} 
\begin{proof}  Apply Bezout's Theorem.
\end{proof}

 We return to the consideration of the relationship between $HC_F$ and $Z$, using the projection $\rho$.

\begin{lemma}  \label{fibresrho} If $\bm v_j \neq \bm 0, j = 1,2$ and $d \neq 0$, then for $[u,y_1,y_2] \notin Z(\bm v,d) \cap G(\bm v,d)$ there is exactly one point in $HC_F(\bm v,d)$ that maps to $[u,y_1,y_2]$ via $\rho$.
\end{lemma}
\begin{proof}
 To see how this works,  choose a point $[u,y_1,y_2] \in Z(\bm v,d) $ such that $L_1(u,y_1,y_2) \neq 0, L_2(u,y_1,y_2) \neq 0$.  Suppose that $[u,y_1,y_2,a,b], [\tilde{u}, \tilde{y}_1,\tilde{y}_2, r,s] \in HC_F(\bm v,d) $ both map to $[u,y_1,y_2]$ via $\rho$.  Then, there is a $\lambda \neq 0$ (if all coordinates are real, $\lambda$ may be chosen to be real) such that
$$\tilde{y}_j = \lambda y_j, \quad \tilde{u} = \lambda u.$$ 
In this case, $\lambda^2[(u-y_1)^2 + y_2^2] = r^2 = \lambda^2 a^2, \ \lambda^2[(u+y_1)^2 + y_2^2] = s^2 = \lambda^2 b^2,$ by definition of $HC_F$.  Therefore, $r = \epsilon_1 \lambda a, s =  \epsilon_2 \lambda b,$ where $\epsilon_j  = \pm 1$.  Since both points lie on $HC_F$, we must have
\begin{equation} \label{epsilon1} L_2(u,y_1,y_2)a-L_1(u,y_1,y_2) b - dab = 0, \end{equation}
$$ L_2(\lambda u, \lambda y_1, \lambda y_2) \epsilon_1 \lambda a - L_1(\lambda u, \lambda y_1, \lambda y_2) \epsilon_2 \lambda b - d \epsilon_1\epsilon_2 \lambda^2 ab =0.$$ After multiplying by $\epsilon_1$ and dividing by $\lambda^2$, the second equation simplifies to
\begin{equation}\label{epsilon2} 
L_2( u,  y_1,  y_2)   a - L_1(u,  y_1,  y_2) \epsilon_1 \epsilon_2 b - d \epsilon_2 ab =0,
\end{equation} 
since $\lambda \neq 0$.

  If $\epsilon_1=\epsilon_2 = 1$, then $[u, y_1,y_2, a,b] = [\lambda u, \lambda y_1, \lambda y_2, \lambda a, \lambda b] = [\tilde{u}, \tilde{y}_1, \tilde{y}_2, r,s]$ as desired.

If $\epsilon_1 = 1, \epsilon_2 = -1$, \eqref{epsilon1},\eqref{epsilon2} become
$$L_2 a-L_1b -dab= 0, \qquad  L_2 a +L_1 b +dab = 0,$$ 
which implies that $L_2a = 0$.  Since we are assuming that $(u,y_1,y_2)$ is such that $L_2(u,y_1,y_2) \neq 0$, this means that $a = r=0$.  Thus,  $f_1 = (u-y_1)^2 + y_2^2 = 0$, so that the equation for $[u,y_1,y_2]$ to be on $Z$ is
$$(L_1^2f_2)^2 = 0.$$  Since $L_1 \neq 0$, we must have $f_2= 0$, or $b=s=0$.  Therefore, again, our two points are the same:  $[u,y_1,y_2,0,0]$.
.

Similarly, if $\epsilon_1 = -1, \epsilon_2 = 1$, we deduce that $a=r=b=s=0$ and our two points are the same: $[u,y_1,y_2,0,0]$.

Finally, if $\epsilon_1=-1,\epsilon_2 = -1$,  then $dab = 0$; this means that $a=r=0$ or $b=s=0$; and just as in the two previous paragraphs, either case implies the other and our two points are the same: $[u,y_1,y_2,0,0]$.

Since the computations did not depend on whether we are considering real points or not, the statement about the real points, when the parameters are real, has been proved as well.
\end{proof} 

\begin{corollary}  With the hypotheses of Lemma \ref{fibresrho},  there are at most 16 points on $Z(\bm v,d)$ above which there are multiple points in the fibre of $\rho:HC_F(\bm v,d) \rightarrow Z(\bm v,d)$
\end{corollary}

Note that there are definitely points on $Z \cap G$ with multiple points in the fibre:  for example, considering $[1,1,0] \in Z \cap G$, two of the singular points $\bm p_j$ map to $[1,1,0]$.

Also, since $L_1,L_2$ are linear, for any given $\bm v,d$ actually finding the points in $Z(\bm v,d) \cap G(\bm v,d)$ amounts to solving quadratic equations, given the nature of the polynomial $h$ that defines $Z$.


\subsection{Summary of the relationship of $HC_F(\bm v, d)$ to $Z(\bm v, d)$}	\label{realSummary}

Looking back at Sections \ref{sec:rhosurjective} - \ref{sec:Gchar}, 
we've seen that, under the assumptions $d \neq 0$, $\bm v_1 \neq \bm 0, \bm v_2 \neq \bm 0$,  we know that 
$$\rho:HC_F(\bm v,d) \rightarrow Z(\bm v,d)$$ is surjective and there is a finite set of at most 16 points $\mathcal{F}(\bm v,d) \doteq Z(\bm v,d) \cap G(\bm v,d)$ on $Z(\bm v,d)$ such that for every point $[u,y_1,y_2]$ of $Z$, not equal to one of these points in $\mathcal{F}$, there is a unique point on $HC_F(\bm v,d)$ which maps to $[u,y_1,y_2]$ via $\rho$. 
  This is independent of whether or not $(\bm v,d)$ has real coordinates.  
  
  Moreover, if $\bm v, d$ are real, and $\mathcal{F}(\mathbb{R})$ is the set of real points in $\mathcal{F}$ (there are at least two of these, namely $[1,1,0],[1,-1,0]$) then
$$\rho:HC_F(\mathbb{R})  \rightarrow Z(\bm v,d)(\mathbb{R})$$ is surjective and the fibre over every point of $Z(\mathbb{R})-\mathcal{F}$ has exactly one point on it.  Thus, we may conclude

\begin{theorem}  If $\bm v_j \neq 0, j=1,2,$ $d \neq 0$ and $\mathcal{F}(\bm v,d)= Z(\bm v,d) \cap G(\bm v,d)$ then $\mathcal{F}$ is a finite set of points, and $\rho$ induces a bijective map
$$\rho:HC_F(\bm v,d) -\rho^{-1}(\mathcal{F}(\bm v,d)) \rightarrow Z(\bm v,d) - \mathcal{F}(\bm v,d).$$  

If $\bm v,d$ are real, then there is a bijective map
$$\rho:HC_F(\bm v,d)(\mathbb{R}) -\rho^{-1}(\mathcal{F}(\bm v,d)(\mathbb{R})) \rightarrow Z(\bm v,d)(\mathbb{R}) - \mathcal{F}(\bm v,d)(\mathbb{R}).$$  

\end{theorem}

Though we have not presented all details of all calculations, given $\bm v,d$, we can figure out exactly which points on $Z$ have inverse images with more than one point on $HC_F$.  Specifying these points just means using the quadratic formula; moreover, figuring out which points are real, and which are not, will be discovered using the quadratic formula as well.  

\subsection{The singularities of $Z$} 
In this section, we do not assume that $\bm v, d$ are real.
\begin{example} \label{examplesection7} For the examples parameterized by $v,d$ from Section 7, using software (e.g. the subprogram {\bf algcurves} in Maple) one can see that, generically, $Z$ has exactly 6 singular points:
$$[1,1,0], [1,-1,0], [0,1, \pm \sqrt{-1}], [1,0, \pm \sqrt{-1}];$$ furthermore, generically,  the singularity type of each of these points is (one could use the documentation on the {\bf algcurves} subprogram in Maple to find the definitions of ``multiplicity" and ``$\delta$-invariant"):
\begin{itemize}
\item $[1,1,0], [1,-1,0]$ are points of multiplicity 4 and $\delta$-invariant 8.
\item $[1,0, \pm \sqrt{-1}], [0,1, \pm \sqrt{-1}]$ are points of multiplicity 2 and $\delta$-invariant 1.
\end{itemize}

The genus-degree formula yields the genus of the (desingularization) as
$$\frac{(8-1)(8-2)}{2} - 8 - 8 - 1-1-1-1 = 1;$$ and again we have an ``elliptic" curve.

These calculations are not true for every choice of parameters $v,d$ though; for example, 
\begin{itemize}
\item if $d^2 = v^2$, we get the same singularities, but of slightly different type.  However, the genus calculation is the same.
\item If $d^2 = 4v^2$, there is one more singularity, and the curve is of genus 0: a rational curve.
\end{itemize}
\end{example}

Let's return to the general case of $Z(\bm v,d)$ and do calculations without software.  We  assume that the lines defined by $L_1 = 0$, $L_2 = 0$ are not components of $Z$. We do not want to consider cases where we know that $Z$ is reducible; thus, we assume the conditions given in Definition \ref{noLfactors}.

Our assumptions here are:
\begin{enumerate}
\item $d \neq 0,$
\item $v_{11}^2 + v_{12}^2 \neq 0, v_{21}^2 + v_{22}^2 \neq 0,$
\item $v_{21} \neq 0$ or $v_{11}^2 -d^2 \neq 0,$ 
\item $v_{11} \neq 0$ or $v_{21}^2 -d^2 \neq 0.$
\end{enumerate}

We don't assume that ${\bf v}_j$ are real vectors unless explicitly stated.

Recall that $Z({\bf v}, d)$ is defined by
$$h \doteq (L_2^2 f_1 + L_1^2 f_2 - d^2 f_1 f_2)^2 - 4 L_1^2 L_2^2 f_1 f_2 = 0,$$
where 
$$L_1 = v_{11}(u-y_1) - v_{12}y_2,  \qquad L_2 = - v_{21}(u + y_1) - v_{22}y_2,$$
$$f_1 = (u-y_1)^2 + y_2^2, \qquad  f_2 = (u+y_1)^2 + y_2^2.$$

We use the notation $\partial g$ to denote any partial derivative of $g$, if $g$ is a function of $u, y_1,y_2$.

Define $q \doteq L_2^2 f_1 + L_1^2 f_2 - d^2 f_1 f_2.$
Now, 
$$\partial h = 2 q \partial q - 4[L_1^2 L_2^2 f_1 \partial f_2 + L_1^2 L_2^2 f_2 \partial f_1 +  L_2^2 f_1 f_2 \partial(L_1^2) + L_1^2 f_1 f_2 \partial(L_2^2)] $$
$$ = 2q \partial q - 4[L_1^2 L_2^2 f_1  \partial f_2 + L_1^2 L_2^2 f_2 \partial f_1 + 2L_1 L_2^2 f_1 f_2 \partial L_1 + 2 L_2 L_1^2 f_1 f_2 \partial L_2].$$

We first consider points where $f_1 = f_2 = 0$.   At such a point, we must have $h = 0$,  $q = 0$ and also $\partial h = 0$.  Therefore, these points are always singular points of $Z$; it's not hard to see that this gives us four singular points
$$[1,0, \pm \sqrt{-1}], [0,1, \pm \sqrt{-1}].$$  For real vectors ${\bf v}_j$, these points will never show up on the real plot, however, they must be taken into account in order to determine the nature of the complex projective curve.   This curve in turn gives geometric information about its real part.

Next, consider points $[u,y_1,y_2]$ such that $f_1 = 0, L_1 = 0$, or $f_2 = 0, L_2 = 0$. At such a point, note that $h = 0$ and $q = 0$; also $\partial h = 0$.  Thus, any such point is a singular point.

In the case of a point where $L_1 = 0, f_1 = 0$,  since $(u-y_1)^2 + y_2^2 = 0,$ $u-y_1 = \pm \sqrt{-1} y_2$. Therefore, $L_1 = 0$ means that
$$\pm \sqrt{-1}v_{11}y_2 + v_{12} y_2 = 0.$$ 
If $y_2 = 0$, then $u = y_1$, and we have the point $[1,1,0]$. If $y_2 \neq 0$, then $ \pm \sqrt{-1}v_{11} + v_{12} = 0$, which implies that $v_{11} ^2 + v_{12}^2 = 0$, and we are assuming this is not true.  This case yields one singular point:  $[1,1,0]$.

Similarly, at a point $[u,y_1,y_2]$ where $L_2 = f_2 = 0$, we have that $y_1 = -u$ and $y_2 = 0$, yielding the point $[1,-1,0]$ or, if $y_2 \neq 0$, we have $\pm \sqrt{-1} v_{12}  + v_{22} = 0$.  This can't happen since $v_{21}^2 + v_{22}^2 \neq 0$. This case yields one singular point: $[1,-1,0]$.

\subsubsection{ $L_1 = 0, L_2^2 = d^2 f_2$}
We suppose that $[u, y_1,y_2]$ is a point where $L_1 = 0, L_2^2 = d^2 f_2$.  In this case,
$$v_{11} (u-y_1) - v_{12}y_2 = 0,$$ so that $v_{11}(u-y_1) = v_{12}y_2.$

Also,
$$0 = L_2^2 - d^2 f_2 = (-v_{21}(u+y_1) - v_{22}y_2)^2- d^2((u+y_1)^2 + y_2^2).$$ Therefore,

\begin{equation} \label{qe1}  0 = (v_{21}^2- d^2)(u+y_1)^2  + 2 v_{21}v_{22}(u+y_1)y_2 + (v_{22}^2-d^2) y_2^2 \doteq p. \end{equation} 

The polynomial $p$ on the right side of the equation above is a nonzero homogeneous quadratic polynomial in $u, y_1, y_2$ (it's nonzero as a polynomial because if $v_{2j}^2 = d^2 $ for $j = 1,2$, then using our assumption $d \neq 0$,  for $j = 1,2$, $v_{2j} \neq 0$,  so the cross term is nonzero).  

Since $L_1 = v_{11}(u-y_1) -v_{12}y_2 $ is a nonzero linear homogenous polynomial in $u, y_1, y_2$, we either have $L_1$ dividing $p$ as a polynomial, or we get at most two solutions $[u, y_1, y_2]$ to the system $L_1 = 0, L_2^2 - d^2 f_2 = 0$.    However, our hypotheses do not allow $L_1$ to divide $p$ as a polynomial.

Therefore, we've found at most two new singular points here.  Generically, we get two distinct points, and in order to compute exactly what these points are we need only apply the quadratic formula.  

\begin{example}However, note that, for example, if $v_{11} = 0 = v_{21}$ and $v_{12} = v_{22} = v \neq 0$ (Example \ref{examplesection7}) then our quadratic $p$ is
$$-d^2 (u+ y_1)^2 + (v^2 - d^2) y_2^2 = 0$$ and $L_1$ is
$$-vy_2 = 0.$$  Since $v \neq 0$, $y_2 = 0$.    Therefore, since $d \neq 0$, we get the point $[1, -1, 0]$ only.  This point has already been found, so we don't really find two new points here: in fact, for this case, no new points are found.  
\end{example}

For the purposes of understanding which real plots are useful to see, let's now assume that ${\bf v}_j$ are real, $j = 1,2$.  
In this case, equation \eqref{qe1} is a quadratic with discriminant
$$\Delta = 4d^2(v_{21}^2 + v_{22}^2 - d^2).$$  Thus, the singularities obtained here are real if and only if
$$v_{21}^2 + v_{22}^2 - d^2 \geq 0,$$ and there are exactly two real singularities if and only if we have a strict inequality.  It's not difficult to write down exactly the singularities:  first use the quadratic formula to determine the solutions to equation \eqref{qe1} (which will give two linear relationships between $u+y_1$ and $y_2$), and then use 
$$L_1  = 0$$ to find the points.

\subsubsection{$L_2 = 0, L_1^2 = d^2 f_1$}
 In this case, the system to solve is:
 $$v_{21}(u+y_1) + v_{22}y_2 = 0$$ and
 \begin{equation} \label{qe2}
 (v_{11}^2 - d^2) (u-y_1)^2 - 2 v_{11}v_{12}(u-y_1)y_2 + (v_{12}^2 - d^2) y_2^2 = 0.
 \end{equation}
Similarly, analyzing this as in the previous section, we get at most two new singular points here.  Generically, we get two new (i.e., not found in any previous section), distinct points.

Note that, again, in the non-generic case ${\bf v}_1 = (0, v) = {\bf v}_2$, we get no new singular points, obtaining only the solution $[1,1,0]$.  

Also, for the purposes of understanding which real plots are useful to see, let's now assume that ${\bf v}_j$ are real, $j = 1,2$.  In this case, equation \eqref{qe2} is a quadratic with discriminant
$$\Delta = 4d^2(v_{11}^2 + v_{12}^2 - d^2).$$  Thus, the singularities obtained here are real if and only if
$$v_{11}^2 + v_{12}^2 - d^2 \geq 0,$$ and there are exactly two real singularities if and only if we have a strict inequality.  It's not difficult to write down exactly the singularities:  first use the quadratic formula to determine the solutions to equation \eqref{qe2} (which will give two linear relationships between $u-y_1$ and $y_2$), and then use 
$$L_2  = 0$$ to find the points.

\subsubsection{Conclusion}

Generically, we have found at least 10 distinct singular points for $Z$.  Moreover, we've explicitly determined a minimum collection of conditions, using calculations (or directions on how to finish such calculations),  which ``genericity" must include.

However, for example, note that in the case ${\bf v}_1 = (0, v) = {\bf v}_2$ of  Example \ref{examplesection7}, only 6 distinct points have been found.  

In any case, we've found between 6 and 10 distinct singular points for $Z$.
\begin{remark}
Although we do not present  a complete proof, we claim that, generically,  $Z$ is irreducible and there are exactly 10 distinct singular points on $Z$, which can be explicitly found using the previous sections.  Note that we can't invoke Bertini's Theorems for the family $\mathcal{Z}$ since this is not a linear system of plane curves.

Furthermore, 2 of these points have multiplicity 4 and $\delta$ value equal to $\frac{4(4-1)}{2} = 6$, and for the remaining 8 points, the multiplicity is 2, with $\delta$ value  1.  This gives the genus as
$$21-(12+8) = 1.$$
\end{remark}

\begin{remark} In the case ${\bf v}_1 = (0, v) = {\bf v}_2$, we claim to have exactly 6 singular points on $Z$, for generic choices of $v,d$; the $\delta $ value for 2 of these points is 8, and for the remaining 4 points, the $\delta$ value is 1.  This gives the genus as
$$21-(16+4) = 1.$$
\end{remark}

\subsection{The FDOA isocurve $A_0$ and the FDOA variety $Z (\bm v,d)$}

Now, we think about our starting point, 
namely the non-algebraic FDOA isocurve  $A_0(\bm v,d)$.  We can identify $A_0(\bm v,d)$ with its homeomorphic image in $\mathbb{R}P^2$, using the identification $(y_1,y_2) \mapsto [1,y_1,y_2]$ of $\mathbb{R}^2$ with the affine open set $U_0$  in $\mathbb{R}P^2$.   [We recall that $U_0$ was  defined in \eqref{R2inP2}].

Recall from \eqref{Ldef}  that 
$$L_2 = (-1-y_1,-y_2)\cdot \bm v_2, \quad L_1 = (1-y_1,-y_2) \cdot \bm v_1$$ 
and define functions $g_j(\bm v,d):\mathbb{R}^2 \rightarrow \mathbb{R}$, $1 \leq j \leq 4$ as in \eqref{factor}  by 

\begin{definition}\label{gs}
\begin{align} 
g_1(\bm v,d) (y_1,y_2) &= L_2R_1- L_1 R_2 - d R_1R_2 ,\cr
g_2(\bm v,d)(y_1,y_2) &= L_2(-R_1)- L_1 (-R_2) - d(- R_1)(-R_2),\cr
g_3(\bm v,d)(y_1,y_2) &= L_2(-R_1)- L_1R_2 - d(- R_1)R_2 ,\cr
g_4(\bm v,d)(y_1,y_2)& = L_2R_1- L_1(-  R_2) - d R_1(-R_2);
\end{align}
\end{definition}
where 
\begin{equation} 	\label{Rdef}
R_1 = |(1-y_1, -y_2)| = \sqrt{(1-y_1)^2 + y_2^2}, \quad  R_2 = |(1+y_1, -y_2)| = \sqrt{(1+y_1)^2 + y_2^2}.
\end{equation}
Here the quantities formerly denoted by $r_j$ are now denoted by $R_j$ because they are no longer independent variables and instead are now non-algebraic functions of $y_1,y_2$. 

Then, define non-algebraic subsets $A_{\pm,\pm}(\bm v,d)$ of $\mathbb{R}^2 = U_0 \subseteq \mathbb{R}P^2$ by
\begin{definition} \label{Aplusminus}  
$$A_{++}(\bm v,d) \doteq \{(y_1,y_2) \in \mathbb{R}^2 \mid g_1(y_1,y_2) = 0\},  \quad 
A_{--}(\bm v,d) = \{(y_1,y_2) \in \mathbb{R}^2 \mid g_2(y_1,y_2) = 0\},$$
$$A_{-+}(\bm v,d) = \{(y_1,y_2) \in\mathbb{R}^2  \mid g_3(y_1,y_2) = 0\},  \quad 
A_{+-}(\bm v,d) = \{(y_1,y_2) \in \mathbb{R}^2 \mid g_4(y_1,y_2) = 0\}.$$
\end{definition} 
  Here the subscripts  $\pm$  are used because they correspond exactly to the choices of positive or negative square roots for \eqref{Rdef}.

Using \eqref{factor} and the definition \eqref{R2inP2} of $U_0$, 
$$Z(\bm v,d)(\mathbb{R})  \cap U_0 \cong  A_{++}(\bm v,d) \cup A_{--}(\bm v ,d) \cup A_{-+}(\bm v,d)   \cup A_{+-}(\bm v,d).$$ 
  Here $U_0$ specifies exactly which open piece of $\mathbb{R}P^2$  our non-algebraic FDOA curve $A_0$ lies in and ``$\cong$" refers to the identification of $U_0$ with $\mathbb{R}^2$.

So, after identifying $\mathbb{R}^2$ with $U_0 \subset \mathbb{R}P^2$,  $$A_0(\bm v,d) \cong A_{++}(\bm v,d) - \{(1,0),(-1,0)\}.$$ 
This makes explicit the precise relationship between the FDOA isocurve $A_0$ and the real FDOA variety $Z(\bm v, d)(\mathbb{R})  \subset U_0 \subset \mathbb{R}P^2 \subset \mathbb{C}P^2$.   In particular, we see that  the $U_0$ part of $Z(\mathbb{R})$ (a real part of a complex algebraic curve) breaks up into four nonalgebraic pieces, namely the $A_{\pm,\pm}$.  
This shows exactly how  $ \mathbb{R}P^2$ can be used to study the FDOA isocurve $A_0$.   
 We have also seen, via  section \ref{realSummary},  exactly how points on the Ho-Chen FDOA variety $HC_F(\bm v,d)(\mathbb{R})$ are related to points on the FDOA isocurve $A_0(\bm v,d)$.

Thus  we have shown in this section that there is a complex algebraic variety $Z(\bm v, d) \subset \mathbb{C}P^2$,  whose real points $Z(\bm v, d)(\mathbb{R})$ contain,  as a proper subset, a non-algebraic curve $A_{++}$, lying in $U_0$, that is (mostly) bijective with $A_0$.  Moreover, we've shown exactly how to obtain the ``other" points of $Z(\mathbb{R})$, that are not on $A_0$, thought of as a subset of $U_0$ using the explicit homeomorphism $\mathbb{R}^2 \cong U_0$.


\subsection{Some plots for $Z(\mathbb{R})$}
This section  considers only the cases ${\bf v}_1 = (0, \pm v), {\bf v}_2 = (0, \pm v)$, with $v \neq 0, d \neq 0$, $v,d \in \mathbb{R}$.   One should not consider these plots representative of what happens for other choices of $\bm v_j$.

 We should choose coordinates to measure ${\bf v}_j$; so, we assume that $v >0$.  Thus our choices of velocities are 
 \begin{definition}\label{velchoice} Velocity choices for this section.
\begin{enumerate}
\item ${\bf v}_1 = (0,v), {\bf v}_2 = (0,v)$
\item ${\bf v}_1 = (0,-v), {\bf v}_2 = (0,-v)$
\item ${\bf v}_1 = (0,v), {\bf v}_2 = (0,-v)$
\item ${\bf v}_1 = (0,- v), {\bf v}_2 = (0,v).$
\end{enumerate}
\end{definition} 
 In the following,  set $\alpha=d/v $.  The Cauchy-Schwarz bound \ref{CSbound} tells us that $-2 \leq \alpha \leq 2$.  Note that the real parameter $\alpha$ is of course not the same as the function $\alpha$ considered in section \ref{alphabeta}.
 
We use the $[u, y_1,y_2]$ coordinates, and look only in the chart where $u = 1$.    For all pairs ${\bf v}_j, j = 1,2$ in the list above, the equation defining $Z$  in this chart is 
$$h \doteq (v^2y_2^2(f_1 + f_2)- d^2f_1f_2)^2 - 4 d^4 y_2^4f_1f_2 = 0$$   or, equivalently,
$$(y_2^2(f_1 + f_2)- \alpha^2f_1f_2)^2 - 4 \alpha^4 y_2^4f_1f_2 = 0,$$ where
$$f_1 = (1-y_1)^2 + y_2^2, \qquad f_2 = (1 + y_1)^2 + y_2^2.$$ 
This defines a family of plane algebraic curves $Z(\alpha)$ in $\mathbb{R}^2$.

Given $(y_1,y_2) \in \mathbb{R}^2$, define real numbers
$$R_1 = \sqrt{(1-y_1)^2 + y_2^2}, \qquad R_2 = \sqrt{(1+y_1)^2 + y_2^2}.$$ 
These are nonnegative square roots by definition. 

Using Definition \ref{gs}, one sees that
$$g_1((0,v),(0,v), d) \doteq (-vy_2)R_1 - (-vy_2)R_2 - dR_1R_2 = vy_2(R_2-R_1) - dR_1R_2$$
$$g_2((0,v),(0,v), d) \doteq vy_2(R_1-R_2) - dR_1 R_2 = g_1((0,-v),(0,-v), d)$$
$$g_3((0,v),(0,v), d) \doteq vy_2R_1  + vy_2R_2+dR_1R_2) = -g_1((0,-v),(0,v), d)$$
$$g_4((0,v),(0,v), d) \doteq -vy_2R_1 - vy_2R_2 + dR_1 R_2 =  -g_1((0,v),(0,-v), d).$$

Using \ref{factor},   $Z$ is also defined by
$$ g_1((0,v), (0,v), d) g_2((0,v), (0,v), d)g_3((0,v), (0,v), d)g_4((0,v),(0,v), d) = 0$$ or
$$g_1((0,v), (0,v), d)g_1((0,-v), (0,-v), d)g_1((0,v), (0,-v), d)g_1((0,-v), (0,v), d) = 0.$$

Define
$A_j((0,v), (0,v), d)$ to be the solution set for $g_j((0,v), (0,v), d) = 0$, $1 \leq j \leq 4$.  These curves were previously labelled labeled $A_{++}, A_{- -}, A_{- +}, 
A_{+ -}$.  

We know that $Z = A_1 \cup A_2 \cup A_3 \cup A_4$.  None of the curves $A_j = A_j((0,v), (0,v), d)$ is an algebraic curve.  

\subsubsection{Plots}
In the plots, each $A_j((0,v), (0,v), d)$ is given a different color;  in order, Red, DarkGreen, NavyBlue, Teal.

  Note that for $j \leq 2$, $A_j((0,v),(0,v), d)$ is equal to some $A_1$ for another pair in Definition \ref{velchoice}.  Thus, even if we are only interested in the behavior of $A_1$ for a pair of vectors in Definition \ref{velchoice},  it really is useful to see the plots of each $A_j((0,v),(0,v), d)$ as these are plots of $A_1$ for a different choice of vectors in \ref{velchoice}.

In terms of $\alpha$, $A_j((0,v),(0,v), d)$ is the zero set of
$$\tilde{g}_j(\alpha) = \frac{g_j((0,v),(0,v),d)}{v};$$
for example,
$$\tilde{g}_1(\alpha) = y_2(R_2-R_1) - \alpha R_1 R_2.$$

Thus the multicolored curve $Z$ is the union of the colored pieces (in this chart, and for real values of $y_1,y_2$ only).

\begin{remark}. Some comments on the plots.  
\begin{enumerate}
\item All plots only plot real points and the parameter $a$ in a plot label is equal to $\alpha$.    We used Maple to make the plots.  
\item The colors are $A_1$-Red, $A_2$-DarkGreen, $A_3$-NavyBlue, $A_4$-Teal.
\item If $\alpha \geq 2$, the only real points of $Z$ or any of the $A_j$, in the $u=1$ affine chart, are the sensor points.  They are on every plot and colored in every color, even if not apparent from the Maple depiction.  This is due to the Cauchy-Schwarz bound $-2v \leq d \leq 2v$.
\item In all of the plots, $Z$ is the multicolored curve with four colors Red, DarkGreen, NavyBlue, Teal.  If it seems that Red and DarkGreen are absent in a plot, that's because the only real points with that color are the points $(-1,0),(1,0)$.  
\item 
\begin{enumerate}
\item In the equal velocity cases $(0,v), (0,v)$ or $(0,-v), (0,-v)$, the $A_1$ piece (which is Red in the first case and DarkGreen in the second) of $Z$ consists of two loops (so it's bounded) that shrink to single points (either $(1,0)$ or $(-1,0)$) as $\alpha$ approaches 1. 
\item In this case, the $A_1$ piece has only these two points where $1 \leq \alpha <2$.
\end{enumerate}
\item Scales change as you go from page to page, so note when this happens.  In particular, as $\alpha$ goes from $0 $ to 1, the size of the Red and DarkGreen loops shrinks, and the NavyBlue and Teal parts expand. 
\item  All four parts of the curve $Z$ are present up until $\alpha = 1$, and as $\alpha$ approaches 1,  the Red and DarkGreen loops shrink rapidly to a single point.  The NavyBlue and Teal parts remain as $\alpha$ travels between 1 and 2.
\item  The (real part of the) curve $Z$ is always unbounded; however,  in these plotted cases, the piece $A_1$ is always bounded.  (Look at the Red and DarkGreen pieces).  One can show that this is not true for more general choices of $\bm v_j$.
\end{enumerate}
\end{remark}
\newpage
\begin{figure}[h!]
\includegraphics{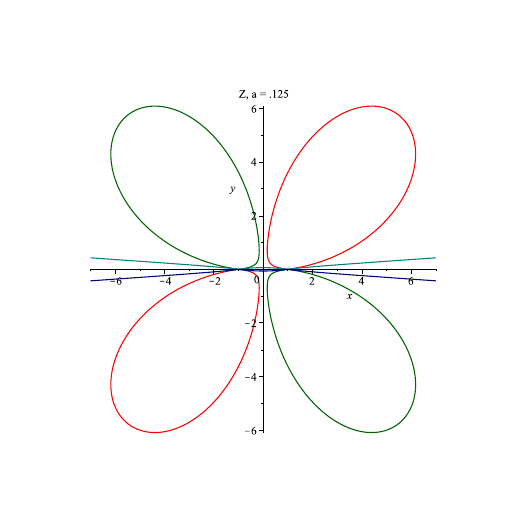}
\includegraphics{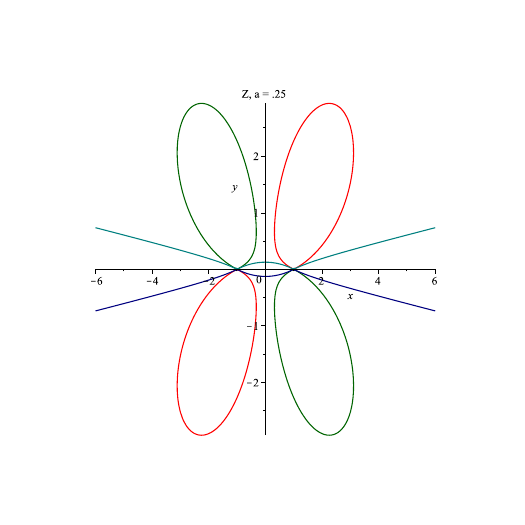}
\end{figure}
\begin{figure}[h!]
\includegraphics{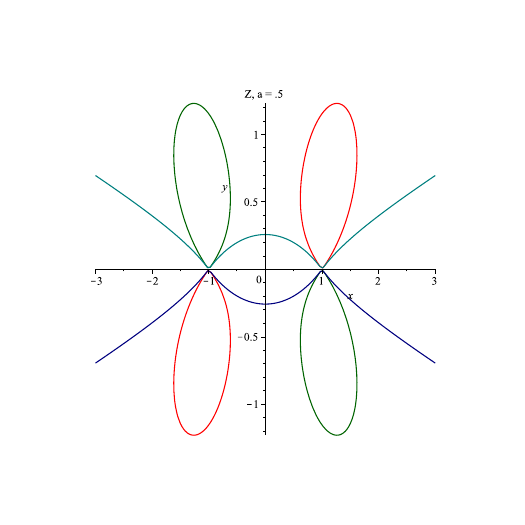}
\includegraphics{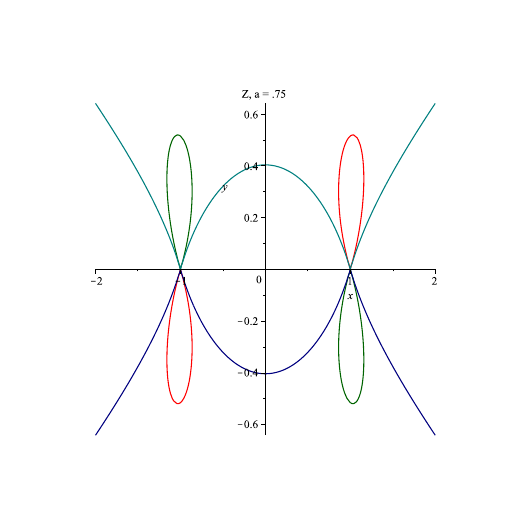}
\end{figure}
\newpage
\begin{figure}[h!]
\includegraphics{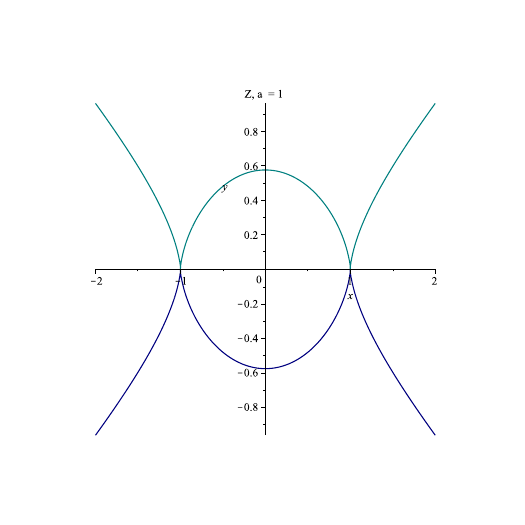}
\includegraphics{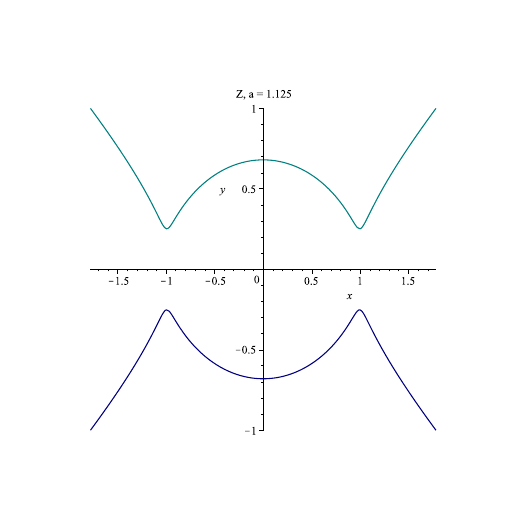}
\end{figure}
\begin{figure}[h!]
\includegraphics{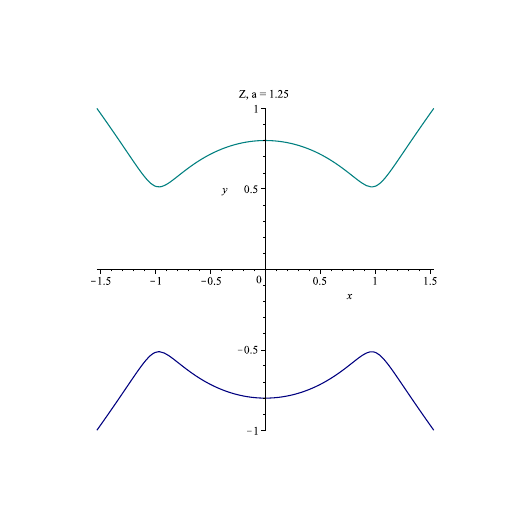}
\includegraphics{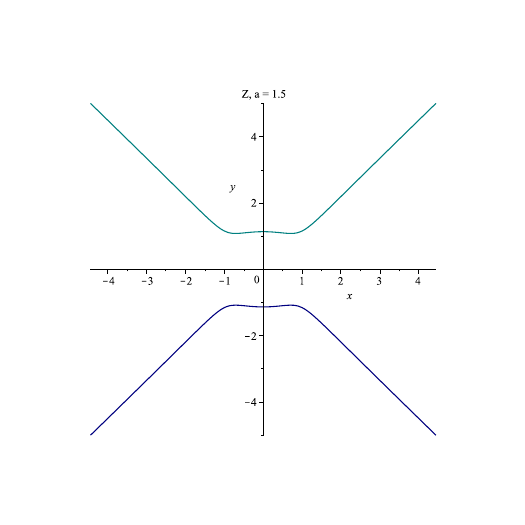}
\end{figure}
\newpage
\begin{figure}[h!]
\includegraphics{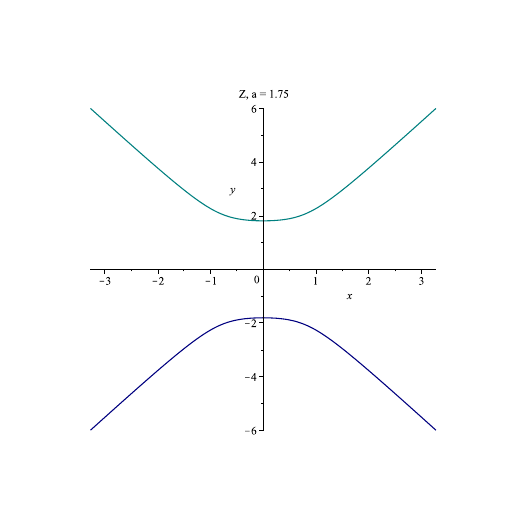}
\includegraphics{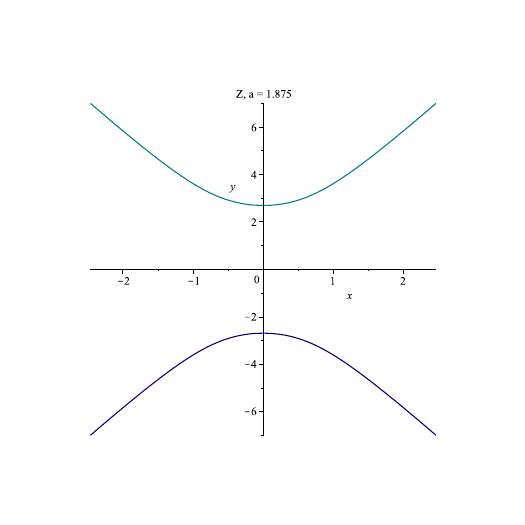}
\end{figure}

\section{Conclusions}

We've accomplished the following in this paper.  

\begin{enumerate}
\item In Section 2, after first describing the starting problem   \ref{sumFDOA}.1,  determined by \eqref{FDOAdef}, we enunciated a modified problem  \ref{sumFDOA}.2,  in terms of the Ho-Chen FDOA subvariety
$HC_F(\bm v,d),$ which is a subvariety of complex projective four-space $\mathbb{C}P^4$.

\item The curves $HC_F(\bm v,d)$ were  described as curves lying on a complex surface $Y =Y(Q_1,Q_2)$, the ambient variety,  in $\mathbb{C}P^4$.  The classification of this ambient variety $Y$, as a complex algebraic surface, were described in Section 4.  In particular, in Section 4, we established that $Y(Q_1,Q_2)$ was a rational surface and wrote down the rational parameterization of its points in terms of a particular open subset of $\mathbb{C}P^2$.  

\item   In Section 5,  still working in complex projective space,  we established that the entire family of curves $\mathcal{H}$ consisting  of all curves $HC_F(\bm v, d)$, as $[\bm v,d]$ varies over $\mathbb{C}P^4$, form a linear system of curves on the ambient variety $Y$.  This enabled us to use one of  Bertini's Theorems, BTGS \cite{EOM}, to conclude that, for most choices of parameters $\bm v,d$, $HC_F(\bm v,d)$ has only the four singularities shared with the ambient variety.   The use of Bertini's Theorem, though powerful, does not tell us exactly what ``most choices of parameters" means.

\item  In Section 6, we introduced a linear system of plane curves $\mathcal{V}$, depending on the same parameters $\bm v,d$ and described the relationship between $HC_F(\bm v, d)$ and the corresponding curve  $V(\bm v,d)$ in the linear system $\mathcal{V}$:  namely, the curves $HC_F(\bm v,d)$ and $V(\bm v,d)$ are {\bf birational} to each other, and the explicit birational map is described.  The curves $V(\bm v,d)$ were determined by using the rational parameterization of the ambient surface $Y$ given in Section 4.   Advantages of working with the plane curves $V(\bm v,d)$  are 1) every curve $V(\bm v,d)$ is a plane quartic curve, and there is some useful theory that helps to understand the nature of such curves; 2) the entire collection $\mathcal{V}$ of all these plane quartics is a linear system, giving us a way to understand 3) the nature of the singularities of a ``generic" such curve $V(\bm v,d)$.  Then, calculations showed that this enabled us to show that, generically, $V(\bm v,d)$ was irreducible, and had only two very nice singularities.  This enabled us to conclude what we expected given Bertini's Theorem on generic irreducibilty \cite{EOM}(without invoking the theorem): generically, $HC_F(\bm v,d)$ is an irreducible curve. Moreover, some conditions for which ``genericity" is achieved were explicitly determined.  Finally, as expected, we could classify the sort of curve one gets (generically) by resolving the singularities of $V(\bm v,d)$:  namely, a smooth curve, of genus one.  Sometimes this is referred to as an ``elliptic curve".  

This is significant because it has implications for the topological type of the real points on $HC_F(\bm v,d)$, though the nature of these implications is not explored in this paper.  Moreover, the fact that these curves are (generically) genus one curves tells us that we can expect the families of curves $\mathcal{H}, \mathcal{V}, \mathcal{Z}$ as well as  any particular random curve (whether considering real points or not) in any of these families,  to be geometrically complicated.  

\item In Section 7, an extended Example was worked through:  $\bm v_1 = \bm v_2 = (0,v).$ 
In this case, there are only two parameters, namely 
$v$ and $d$.  One point of working through this Example is to avoid entirely the use of Bertini's theorems, and to classify the singularities of $HC_F(v,d)$ for various $v$ and $d$, as well as  to describe explicitly the smooth curves $\tilde{V}(v,d)$ of genus 1 that one gets after resolving the singularities of $HC_F(v,d)$, when those curves are irreducible.  Moreover, these curves $\tilde{V}(v,d)$   are also plane curves, indeed, smooth irreducible cubic curves.  Such curves have a {\bf Weierstrass normal form} which could be determined using the appropriate software.  

\item In Section 8, we described explicitly another  plane curve $Z(\bm v,d)$, a projection of $HC_F(\bm v,d)$, whose real points are more directly connected to the  original non-algebraic  FDOA isocurves $A_0$ in $\mathbb{R}^2$.  In fact, the plane curves $Z(\bm v,d)$ are the smallest projective algebraic variety containing the 
FDOA isocurves $A_0$.  Unfortunately, this set of plane curves does not vary in a linear system of curves, and has more, and ``worse" singularities than $HC_F(\bm v,d)$ does, as discussed in the section about the singularities of $Z$.   Finally,  for real choices of $\bm v,d$, the connection between the real points on $HC_F(\bm v,d)$ and the real points on the curves $Z(\bm v,d)$ was explained, precisely, in terms of the projection.

\end{enumerate}
Problems that we do not address in this paper,  
but are of interest,  include
\begin{enumerate} 
\item Studying  intersections of  TDOA hyperbolas and FDOA varieties $Z(\mathbb{R})$ in $\mathbb R^2$. 
\item Repeating the work in this report, ``one dimension up": in other words, study the ``3D-2 sensor" FDOA problems parallel to the work done in this report.  The starting point here is equation \eqref{FDOAdef}, where the sensor locations, velocities and source are in  $\mathbb R^3$.  This is quite an interesting problem, from the geometric point of view.
\item For multiple sensor pairs, studying the intersections of the associated FDOA varieties and the combinations of values of FDOA that can arise.
\end{enumerate}


\section{Acknowledgements}
This material is based upon work supported in part by the Air Force Office of Scientific Research under award numbers FA9550-18-1-0087 and FA9550-21-1-0169, and in part by the U.S. Office of Naval Research under award number N00014-21-1-2145.   Any opinions, findings, and conclusions or recommendations expressed in this material are those of the authors and do not necessarily reflect the views of the Office of Naval Research or of the  United States Air Force.


\appendix

\section{Example:  The smooth quadric surface in $\mathbb{C}P^3$}  \label{quadric}


Giving $\mathbb{C}P^3$ the coordinates $w = [w_0,w_1,w_2,w_3]$,  and consider the degree 2 polynomial
$$Q = w_0w_3-w_1w_2.$$  It's not hard to see that this polynomial is irreducible.  The smooth quadric surface in $\mathbb{C}P^3$ is
$$Y(Q) \doteq \{w \in \mathbb{C}P^3 \mid w_0w_3-w_1w_2=0 \}.$$ This is a smooth projective variety; we call it ``the" smooth quadric surface in $\mathbb{C}P^3$ because any other such is isomorphic (with a complex, though not necessarily real, projective isomorphism) to $Y(Q)$.

As a topological space  $Y(Q) \cong \mathbb{C}P^1 \times \mathbb{C}P^1$ using the Segre embedding
$$\sigma:\mathbb{C}P^1 \times \mathbb{C}P^1 \rightarrow \mathbb{C}P^3,$$
$$\sigma([u_0,u_1],[v_0,v_1]) = [u_0v_0,u_0v_1,u_1v_0,u_1v_1].$$    One can show that $\sigma$ is an embedding with image equal to $Y(Q)$; for example, a local inverse to $\sigma$ on the open set of $Y(Q)$ where $w_0 \neq 0$ has the description $[w_0,w_1,w_2,w_3] \mapsto ([w_0,w_2],[w_0,w_1])$.    Also, the following diagram commutes, and both horizontal arrows are homeomorphisms:
$$\begin{array}{ccc}
\mathbb{C}P^1 \times \mathbb{C}P^1 & \stackrel{\sigma}{\longrightarrow} & Y(Q) \\
\cup & & \cup \\
\mathbb{R}P^1 \times \mathbb{R}P^1 & \stackrel{\sigma}{\longrightarrow} & Y(Q)(\mathbb{R}).
\end{array}
$$
This identifies the topological type of $Y(Q)$ as $S^2 \times S^2$, since $\mathbb{C}P^1$ is homeomorphic to $S^2$; and the topological type of $Y(Q)(\mathbb{R})$ as $S^1 \times S^1$ (the torus) since $\mathbb{R}P^1$ is homeomorphic to $S^1$.  

The Segre isomorphism $Y(Q) \cong \mathbb{C}P^1 \times \mathbb{C}P^1$ also allows us to give a definition of an $(m,n)$-curve on $Y(Q)$:  we'll define an $(m,n)$-curve on $\mathbb{C}P^1 \times \mathbb{C}P^1$ and then use the isomorphism $\sigma$ to say that an $(m,n)$-curve on $Y(Q)$ is $\sigma(E)$, for an $(m,n)$-curve $E$ on $\mathbb{C}P^1 \times \mathbb{C}P^1$.  So, suppose that $F(u_0,u_1,v_0,v_1)$ is a nonzero polynomial in $\mathbb{C}[u_0,u_1,v_0,v_1]$ that is ``bihomogeneous" of degree $m \geq 0$ in $u_0,u_1$ and of degree $n \geq 0$ in $v_0,v_1$, where at least one of $m,n$ is positive.  Then, for every $[u_0,u_1] \in \mathbb{C}P^1, [v_0,v_1] \in \mathbb{C}P^1$, whether or not $F(u_0,u_1,v_0,v_1) = 0$ is independent of the choice of $u_0,u_1$ and the choice of $v_0,v_1$.  Therefore, it makes sense to define
$$E = E(F) = \{([u_0,u_1],[v_0,v_1])\in \mathbb{C}P^1 \times \mathbb{C}P^1 \mid F(u_0,u_1,v_0,v_1) = 0\}.$$  In this case, we say that $E$ is an $(m,n)$-curve in $\mathbb{C}P^1 \times \mathbb{C}P^1.$  

For example, if $g(w_0,w_1,w_2,w_3)$ is a homogeneous polynomial of degree $d>0$ in the $w$s, then making the substitution $w_0 = u_0v_0,w_1 = u_0v_1,w_2 = u_1v_0,w_3 = u_1v_1$, we see that if 
$$F(u_0,u_1,v_0,v_1) \doteq  g(u_0v_0,u_0v_1,u_1v_0,u_1v_1),$$ then $F$ defines a $(d,d)$-curve $E(F)$ in $\mathbb{C}P^1 \times \mathbb{C}P^1$ such that $\sigma(E(F)) = V_Q(F),$ and therefore $V_Q(F)$ is a $(d,d)$-curve on $Y(Q)$.  It's not as straightforward to write down an $(m,n)$-curve in $Y(Q)$ where $m \neq n$, since $\sigma^{-1}$ has different definitions locally.  

If we make the (real) change of coordinates  $w = w_3-w_0, x = w_2-w_1, y = w_0+w_3, z = w_1+w_2$, we see that 
$$Y(Q) = \{[w,x,y,z] \in \mathbb{C}P^3 \mid -w^2 + x^2 + y^2 -z^2 = 0\}.$$ 
 In real Euclidean space $\mathbb{R}^3$, obtained by setting $w=1$,  
 $Y(Q)$ corresponds to the familiar quadric surface from our undergraduate days:  the hyperboloid of one sheet, with axis the $z$-axis. 
 


\begin{thebibliography}{XXX}

\bibitem{BHPV}  
Barth, Wolf, Klaus Hulek, Chris Peters, and Antonius Van de Ven. {\it Compact complex surfaces}. Vol. 4. Springer, 2015.

\bibitem{BKT} 
Bashelor, Andrew, Amy Ksir, and Will Traves. ``Enumerative algebraic geometry of conics." The American Mathematical Monthly 115, no. 8 (2008): 701-728.
 
\bibitem{AB} 
Beauville, Arnaud. {\it Complex algebraic surfaces}. No. 34. Cambridge University Press, 1996.

\bibitem{GH}  Griffiths, Philip, and Harris, Joe.  {\it Priniciples of Algebraic Geometry}, John Wiley \& Sons, 2014.


\bibitem{KC} 
Cameron, Karleigh J. ``FDOA-based passive source localization: a geometric perspective." PhD diss., Colorado State University. Libraries, 2018.

%

\bibitem{CoNoAnSa} 
Compagnoni, Marco, Roberto Notari, Fabio Antonacci, and Augusto Sarti. ``A comprehensive analysis of the geometry of TDOA maps in localization problems." Inverse Problems 30, no. 3 (2014): 035004.

\bibitem{CNRAS} 
Compagnoni, Marco, Roberto Notari, Andrea Alessandro Ruggiu, Fabio Antonacci, and Augusto Sarti. ``The algebro-geometric study of range maps." Journal of nonlinear science 27, no. 1 (2017): 99-157.

\bibitem{HoCh} Ho, K.C. and Chan, Y.T.,  ``Geolocation of a Known Altitude Object from TDOA and FDOA Measurements", \textit{IEEE Transactions on Aerospace and Electronic Systems}, Vol. 33, No.3 (1997).

\bibitem{EOM} Bertini theorems. {\it Encyclopedia of Mathematics.} URL:  http://encyclopediaofmath.org/index.php?title=Bertini\_theorems\&oldid=41168.  

\bibitem{E95} Eisenbud, David,   {\it Commutative Algebra with a View Toward Algebraic Geometry.} Graduate Texts in Mathematics. 150. Springer-Verlag (1995).

\bibitem{K} 
Koll\'ar, J\'anos. {\it Rational curves on algebraic varieties}. Vol. 32. Springer Science \& Business Media, 1999.

\bibitem{L}  N. Levanon and E. Mozeson, {\em Radar signals}, John Wiley \& Sons, 2004.


\bibitem{SMGC}  L.L. Scharf, L. T. McWhorter, J. Given, and M. Cheney, ``General First-Order Framework for Passive Detection with Two Sensor Arrays," Asilomar Conference on Signals and Systems, Pacific Grove, CA, Nov 2019.


\end{thebibliography}
\end{document}